\documentclass[review,onefignum,onetabnum]{siamonline220329}

\usepackage{lipsum}
\usepackage{amsfonts}
\usepackage{graphicx}
\usepackage{epstopdf}
\usepackage{algorithmic}
\ifpdf
  \DeclareGraphicsExtensions{.eps,.pdf,.png,.jpg}
\else
  \DeclareGraphicsExtensions{.eps}
\fi

\newcommand{\cO}{\mathcal{O}}

\newcommand{\cX}{\mathcal{X}}

\newcommand{\cD}{\mathcal{D}}
\newcommand{\cA}{\mathcal{A}}

\newcommand{\sX}{\mathsf X}
\newcommand{\sY}{\mathsf Y}

\newcommand{\bbE}{\mathbb{E}}

\newcommand{\bbP}{\mathbb{P}}

\newcommand{\bbR}{\mathbb{R}}

\newcommand{\bbV}{\mathbb{V}}

\DeclareMathOperator*{\argmax}{arg\,max}

\newcounter{hypA}
\newenvironment{hypA}{\refstepcounter{hypA}\begin{itemize}
  \item[({\bf A\arabic{hypA}})]}{\end{itemize}}
\newcounter{hypB}

\newcounter{hypD}

\usepackage{enumitem}
\setlist[enumerate]{leftmargin=.5in}
\setlist[itemize]{leftmargin=.5in}

\newsiamremark{remark}{Remark}
\newsiamremark{hypothesis}{Hypothesis}
\crefname{hypothesis}{Hypothesis}{Hypotheses}
\newsiamthm{claim}{Claim}

\headers{
Bayesian Deep Learning with
Multilevel Trace-class  
Neural Networks}
{N. K. Chada, A. Jasra, K. J. H. Law and S. S. Singh}

\title{Bayesian Deep Learning with
Multilevel Trace-class  
Neural Networks
\thanks{Submitted to the editors DATE.
}}

\author{Neil K. Chada\thanks{Department of Mathematics, City University of Hong Kong, 83 Tat Chee Avenue, Kowloon Tong, Kowloon, Hong Kong SAR
  (\email{neilchada123@gmail.com}).}
\and Ajay Jasra\thanks{School of Data Science, The Chinese University of Hong Kong,  Shenzhen, CN
  (\email{ajayjasra@cuhk.edu.cn}).}
  \and Kody J. H. Law\thanks{School of Mathematics, University of Manchester, Manchester, M13 9P, UK
  (\email{kody.law@manchester.ac.uk})}
\and Sumeetpal S. Singh\thanks{NIASRA, School of Mathematics and Applied Statistics, University of Wollongong, NSW 2522, Australia
  (\email{sumeetpals@uow.edu.au}).}}%\footnotemark[3]}

\usepackage{amsopn}

\ifpdf
\hypersetup{
  pdftitle={An Example Article},
  pdfauthor={D. Doe, P. T. Frank, and J. E. Smith}
}
\fi

\PassOptionsToPackage{table}{xcolor}
\usepackage{subcaption}
\begin{document}

\maketitle

\begin{abstract}
\textcolor{black}{In this article we consider Bayesian inference associated to deep neural networks (DNNs) and in particular, trace-class neural network (TNN) priors  \cite{sell2020dimension} which can be preferable to traditional DNNs because they (a) are identifiable and (b) possess desirable convergence properties. 
TNN priors are defined on functions with infinitely many hidden units, and have strongly 
convergent Karhunen-Loeve-type approximations with finitely many hidden units.  
A practical hurdle is that the Bayesian solution is computationally demanding, requiring 
%Markov 
simulation methods, so approaches to drive down the complexity are needed. 
In this paper,  we leverage the strong convergence of TNN
in order to apply Multilevel Monte Carlo (MLMC) to these models. 
In particular, an MLMC method that was introduced in 
\cite{beskos2018multilevel} %can be 
is used to approximate posterior expectations of Bayesian TNN models with optimal computational complexity, 
and this is mathematically proved. 
The results are verified with several 
numerical experiments on model problems arising in machine learning, including %. These include Bayesian
\textcolor{black}{some toy regression and classification models, %as well as Bayesian
MNIST image classification, and a challenging reinforcement learning problem.
Furthermore, we illustrate the practical utility of the method on MNIST as well as IMDb sentiment classification.}}
\end{abstract}

% REQUIRED

\section{Introduction}

%Machine learning \cite{bishop,mackay,murphy} has emerged as a topic of interest, in a wide class of mathematical and statistical disciplines. This is largely due to the amount of machinery
%in terms of advanced numerical algorithms, but also the amount of readily available data. Common examples of machine learning tasks, or methodologies,
%include clustering, classification and reinforcement learning. Typically many of these methodologies aim to optimize a function, which is used for learning and prediction purposes,
%such as image processing, handwriting recognition and natural language processing. Specifically, these problems are solved with deterministic methods, based on variational techniques.
\textcolor{black}{
Deep neural networks (DNNs) \cite{dnn}
are a popular and powerful parametric model class which can be used for 
solving a variety of machine learning problems. 
These architectures are very much applicable to a wide array of disciplines and applications,
where a major advantage of DNNs is that one attains the universal approximation theorem, 
which in brief states that %irrespective of the function, 
the NN can approximate %the 
a wide class of target functions \cite{hornik}.
%of interest. 
For this work our interest in DNNs are w.r.t.~generating random processes. 
%which are usually constructed with either the limit of the length or width 
There has been extensive
work on connecting DNN with Gaussian and non-Gaussian processes \cite{bach,lee2018deep,matthews2018gaussian,neal}. 
%In a Bayesian setting, this can be exploited as prior forms which has been tested firstly on 
%In the case of Deep Gaussian processes, 
%which are no longer Gaussian,
%variational approaches are typically employed \cite{lawrence}.
Deep Gaussian and $\alpha$-stable process priors for 
fully Bayesian edge-preserving inversion 
have been employed 
in the context of inverse problems in 
\cite{DGS18,DLS21}, but these 
studies have typically been limited to low input dimensionality
due to the debilitating computational complexity.
%. This is in the context of Deep Gaussian and $\alpha$-stable processes for edge-preserving inversion. 
%Secondly they have 
%Recently  been developed for machine learning tasks \cite{sell2020dimension,VVM19}. 
%For the former, 
Sell et al. introduce trace-class DNN (TNN) 
priors \cite{sell2020dimension}, %,VVM19}, 
which are non-stationary, non-Gaussian, and well-defined in the 
infinite-width limit, 
yet scale well with input dimension.
%For example, they have been successfully
%applied to problems with up to $n=17$ dimensions.
Fitting such model often requires Monte Carlo (MC) methods.}
\\
\textcolor{black}{
Monte Carlo methods are well-known class of methods aimed to solve stochastic computation problems. Such developments have been primarily in the fields of computational physics, statistics and numerical analysis. In particular one methodology which improves on vanilla MC is multilevel Monte Carlo (MLMC). MLMC aims to reduce the computational cost, and complexity, to attain a particular order of mean square error (MSE), i.e.~$\mathcal{O}(\varepsilon^2)$ for $\varepsilon>0$. First introduced in \cite{MBG08,MBG15,hein}  and primarily applied to diffusion processes in mathematical finance, it has since seen various extensions to other fields. 
Related to this work, it has been applied to MC methods within computational statistics which includes sequential Monte Carlo (SMC), Markov Chain Monte Carlo (MCMC) \cite{beskos2017multilevel,beskos2018multilevel,heng2021unbiased,JKL17,jasra2021unbiased} and other related methods, based on sampling from a distribution of interest. However in terms of the application of MLMC to machine learning, there has been limited work on this. 
%Some works have applied
%Notable works include applying 
Notable works along this direction include %the use of MLMC to improve 
the improvement of complexity 
in gradient estimators within variational inference \cite{FS21,SC21}, 
%MLMC techniques to 
and improved complexity for data-driven surrogate modelling 
%forward solvers for 
of high-dimensional PDE and SDE models \cite{EJJ22,GHR21,LMM21}.
%, but also applying it to variational inference problems \cite{FS21,SC21}. 
However as of yet this has not been exploited for the use of statistical inference problems in machine learning.}
\textcolor{black}{Therefore our motivation in this paper is to answer the following question: \\ \\
\centering
\textit{Can one improve the computational complexity for Bayesian deep learning?}}
%Can one provide a Bayesian methodology, that reduces the computational complexity to obtain a certain level of accuracy in terms of the mean square error?}}
\\ \\
\textcolor{black}{In this paper our focus is on fully Bayesian inference with DNNs %related to 
for machine learning problems, through the aid and use of %advanced 
%stochastic algorithms, which have been exploited for 
Monte Carlo  simulation methods.
%, but also deep neural networks. 
%Before describing in detail our contributions, we first present a general literature overview on these entities, and how they are related to Bayesian machine learning \cite{HTF01,neal,bishop,murphy}.
%A popular tool for machine learning problems are d
We will provide a new methodology, 
%aimed at Bayesian machine learning. This methodology we will refer to as 
multilevel Bayesian deep neural networks,
%where we 
which combines TNN %trace-class neural network 
prior models with an advanced multilevel sequential Monte Carlo (MLSMC) algorithm \cite{beskos2018multilevel}. 
This methodology combines the advantages of the TNN 
model in terms of scalability and flexibility, with %achieves the 
the %power of population SMC and the advantage 
efficiency of MLSMC, to achieve the canonical asymptotic 
computational complexity of $1/$MSE, hence answering the question above in the affirmative. 
%which have been developed for posterior sampling. 
%Our intention from this is to improve
%on computational efficiency, which is related to attaining the canonical Monte Carlo rate, which 
For problems of this form, where MC is used to simulate from an 
approximate distribution, 
such complexity can only be attained in a MLMC framework,
and it is impossible to do better.}

\subsection{Contributions}

Our contributions of this work are summarized through the following points.
\begin{itemize}
\item We show how to use MLMC in the context of Bayesian inference for deep neural networks.
The specific method we use is MLSMC samplers. Our motivation is to reduce the computational complexity, where we show it is possible %hope 
to attain the canonical rate of convergence $\textrm{Cost} = \textrm{MSE}^{-1}$.
\item We %provide and 
prove two key results within our work. 
%The first being that, 
First, %given our framework, 
we are able to derive 
%canonical rates of convergence using such 
{strong rates of convergence for TNN priors} in a multilevel setting. 
%The second
%key result is 
Secondly, we establish a bound on the MSE, 
which can be decomposed into a variance and bias term. 
{By leveraging these two results, we establish that our method converges 
to the true underlying posterior with the canonical rate}. 
%This result establishes the convergence of our numerical scheme to the true underlying posterior
%\blu{with the canonical rate}.
\item Various numerical experiments are presented to verify the theoretical findings discussed above. This is related to both demonstrating that the %canonical 
strong convergence rate is attained, 
and also that the MLSMC sampler with TNN priors 
can reduce the cost to attain a particular order of MSE. 
\textcolor{black}{The experiments verifying the theory
are conducted on machine learning examples including simple toy regression and classification problems,  MNIST classification and reinforcement learning. 
%\textcolor{blue}{
For MNIST classification, 
%experiment, 
%we provide a higher-dimensional example of 100, where 
we also provide a sensitivity analysis on the parameter choices.
In addition to these examples, we look at a more practical example of sentiment classification for IMDb dataset. 
%high-dimensional un-structured data example }
}
\end{itemize}
\subsection{Outline}
The outline of this article is structured as follows. In Section \ref{sec:model} we present preliminary material related to both the model problem and multilevel Monte Carlo. 
This leads onto Section \ref{sec:algo_math},  where we present our numerical algorithm coupled with our main mathematical result, which is a bound on the MSE establishing
convergence of our method. We then discuss and introduce our multilevel trace-class priors in Section \ref{sec:tnn} where we demonstrate they attain the canonical rate of convergence,
with numerical verification. In Section \ref{sec:numerics} our numerical results are presented on a range of machine learning tasks, verifying the improved gains on computational
cost using the MLSMC sampler, combined with TNN priors. Finally we conclude, and remark on future directions, in Section \ref{sec:conc}. 
The proof of our main theorem is deferred to the Supplementary Material.

\section{Model Formulation}\label{sec:model}

In this section we provide a preliminary background on the setting, and formulation of the model problem.
This will include an initial discussion on our setup, which will discuss neural networks and how they related to Bayesian modeling,
as well as introducing the key concepts of multilevel Monte Carlo.

Suppose we have data $\mathcal{D} = \left((x_1,y_1),\dots,(x_N,y_N)\right)$, $N\in\mathbb{N}$,
where, for $i\in\{1,\dots,N\}$, $x_i \in \sX$ and $y_i\in \sY$. The objective is
to infer a predictive model $f: \sX \rightarrow \sY$ based on the data.
One way to do this is with a parametric model
of the form $f: \sX\times\Theta \rightarrow \sY$, with $\Theta\subseteq\mathbb{R}^{d_{\theta}}$.
It is assumed that $x_{1:N}$ are deterministic.
%Following standard discriminative modelling convention, $x$ will be considered
%deterministic and know, and will thus not be modelled statistically.

Suppose further that $\mathsf{Y}=\mathbb{R}^m$, then we will assume that for $i\in\{1,\dots,N\}$
\begin{equation}\label{eq:modelreg}
y_i = f(x_i,\theta)+ \epsilon_i \, , \qquad
\epsilon_i \stackrel{\textrm{ind}}{\sim} \mathcal{N}_m(0,\Sigma_i),
\end{equation}
where $\textrm{ind}$ denotes independence across the indices $i\in\{1,\dots,N\}$ and $\mathcal{N}_m(\mu,\Sigma)$ denotes the $m-$dimensional
Gaussian distribution with mean $\mu$ and covariance matrix $\Sigma$. \textcolor{black}{From this we have $\textcolor{black}{y_i|\theta \sim \mathcal{N}_m( f(x_i,\theta), \Sigma_i)}$}.

 Now consider the case that $\sY = \{1,\dots,m\}$, for some $m\in\mathbb{N}$.
{\color{black}Let $f:\mathsf{X}\times\Theta\rightarrow\mathbb{R}^m$, with $f(x,\theta)=(f_1(x,\theta),\dots,f_m(x,\theta))$, and we define the so-called {\em softmax} function as
\begin{equation}\label{eq:class}
%\mathbb{P}(Y_i=k|x_i) =
h_k(x_i,\theta) :=
\frac{\exp\{f_k(x_i,\theta)\}}{\sum_{j=1}^K \exp\{f_j(x_i,\theta)\}}\, , \qquad k\in\mathsf{Y},
\end{equation}
We will assume that $y_i \sim h(x_i,\theta)$, independently for $i\in\{1,\dots,N\}$,
\\ where $h(x,\theta)=(h_1(x,\theta),\dots,h_m(x,\theta))$
denotes a categorical distribution on $m$ outcomes, given input $x$.
}
%In order to allow the presentation that follows to generalize easily,
%we adopt the common convention of redefining the output
% as the categorical distribution itself, i.e.
%$\sY = [0,1]^m \subset \bbR^m$,
%%as the vector
%%$(\mathbb{P}(Y_i=1), \dots, \mathbb{P}(Y_i=m))$,
%and a categorical observation $y_i = k$ is mapped to $y_i = e_k \in \bbR^m$,
%the $k^{\rm th}$ standard basis vector in $\bbR^m$.

A popular parametric model is deep neural networks (DNN) \cite{dnn}. 
In this scenario, affine functions are composed with simple element-wise {\em activation functions}
%\begin{eqnarray}\label{eq:activation}
{\color{black} 
%$\textcolor{blue}
$\nu:\bbR \rightarrow \bbR.$
For  $z \in \bbR^n$, including the limit $n\rightarrow \infty$, we define 
${\sigma_n:\bbR^n \rightarrow \bbR}$ as follows
$\sigma_n(z) := (\nu(z_1),\dots,\nu(z_n))^{\top}$.
We also refer to $\sigma_n$ as the activation function.
%(\sigma(z_1),\dots,\sigma(z_k))^{\top}$.}
%We make the definition that for $k\in\mathbb{N}$ given
%and $z_{1:k}\in\mathbb{R}^k$, $\sigma_k:\mathbb{R}^k\rightarrow\mathbb{R}^k$
%with $\sigma_k(z_{1:k})=(\nu(z_1),\dots,\nu(z_k))^{\top}$.
%\\&:& z \mapsto
%\end{eqnarray}
%which map $z = (z_1,z_2,\dots,z_n)\in \bbR^n$ to
%$\sigma(z) := (\sigma(z_1),\sigma(z_2),\dots,\sigma(z_n))\in \bbR^n$.
An example is $\nu(z)=\max\{0,z\}$ for $z\in \bbR$, the so-called
{\em ReLU activation}.}
Let $\sX = \bbR^n$ and $\sY=\bbR^m$.
The DNN itself can be defined in the following way.
Let $D\in\mathbb{N}$, $(n_0,\dots,n_D)\in\mathbb{N}^{D+1}$ be given,
with the constraint that
$n_0=n$ for the input layer and $n_D=m$ for the output layer. Now for weights $A_d\in \mathbb{R}^{n_{d} \times n_{d-1}}$ and biases $b_d\in \bbR^{n_d}$,
$d\in\{1,\dots,D\}$ also given, we use the notation $\theta := \left((A_1,b_1),\dots,(A_D,b_D)\right)$ and so $\theta\in\Theta=\bigotimes_{d=1}^D\{ \mathbb{R}^{n_d\times n_{d-1}}\times\mathbb{R}^{n_d}\}$. Now set
\begin{eqnarray}\nonumber
g_0(x,\theta) & := & A_1x + b_1,\\ \nonumber
g_d(x,\theta) & := & A_d\sigma_{n_{d-1}}(g_{d-1}(x)) + b_d\, , \qquad d\in\{1,\dots D-1\}, \\
%\end{eqnarray*}
%Then in the case of regression we have
%\begin{equation}
\label{eq:DNN}
f(x,\theta) &:=&  A_D\sigma_{n_{D-1}}(g_{D-1}(x)) + b_D,
\end{eqnarray}
where $f(x,\theta)$ is the output of the final layer of the DNN.
For any given $f(x,\theta)$ and parameter space $\Theta$ one can place a prior $\overline{\pi}$ on $\Theta$. Given the structure in \eqref{eq:modelreg} one then has a posterior
\begin{equation}\label{eq:bnn_posterior}
\pi(\theta | y_{1:N}) \propto p(y_{1:N}|\theta) \overline{\pi}(\theta) \, ,
\end{equation}
assuming it is well-defined.
The likelihood function, in the case \eqref{eq:modelreg} is exactly
\begin{equation}\label{eq:reglikelihood}
p(y_{1:N}|\theta) = \prod_{i=1}^N \phi_m(y_i;f(x_i,\theta),\Sigma_i),
\end{equation}
where $\phi_m(y;\mu,\Sigma)$ is the $m-$dimensional Gaussian density of mean $\mu$, covariance $\Sigma$ evaluated at $y$ given.
In the case of classification, the likelihood function associated to \eqref{eq:class}
is exactly
{\color{black}\begin{equation}\label{eq:classlikelihood}
p(y_{1:N}|\theta) = \prod_{i=1}^N \prod_{k=1}^m h_{k}(x_i,\theta)^{\mathbb{I}_{[y_i=k]}} \, .
\end{equation}
}

\subsection{Multilevel Bayesian Neural Networks}\label{sec:mlbnn}
We shall begin with a short review of MLMC. Let us assume that we are given a probability density $\Psi$,
on a state-space $\mathsf{U}$ and it is of interest to compute expectations of $\Psi-$integrable functions,
$\varphi:\mathsf{U}\rightarrow\mathbb{R}$; $\Psi(\varphi):=\int_{\mathsf{U}}\varphi(u)\Psi(u)du$ with $du$ a dominating
$\sigma-$finite measure (often Lebesgue). Now, we assume that working with $\Psi$ is not computationally feasible (e.g.~has an infinite computational cost)
but there exist a scalar parameter $l\in\mathbb{N}$ which parameterises an approximation $\Psi_l$ of $\Psi$, where $\Psi_l$ a density on
a state-space $\mathsf{U}_l\subseteq\mathsf{U}$ such that:
\begin{enumerate}
\item{For any $\varphi:\mathsf{U}\rightarrow\mathbb{R}$ that is both $\Psi_l$ and $\Psi-$integrable we have $\lim_{l\rightarrow\infty}\Psi_l(\varphi)=\Psi_l(\varphi)$, where $\Psi_l(\varphi)=\int_{\mathsf{U}_l}\varphi(u)\Psi_l(u)du$.}
\item{The cost of computing with $\Psi_l$ is increasing in $l$.}
\end{enumerate}
Let $L\in\{2,3,\dots\}$ be given, we have the telescoping sum identity
\begin{equation}
\label{eq:tele}
\Psi_L(\varphi) = \Psi_1(\varphi) + \sum^L_{l=2}[\Psi_l-\Psi_{l-1}](\varphi),
\end{equation}
with $[\Psi_l-\Psi_{l-1}](\varphi)=\Psi_l(\varphi)-\Psi_{l-1}(\varphi)$.
The idea behind MLMC is to try and approximate the R.H.S.~of \eqref{eq:tele} in such a way as to reduce the cost of a particular order of mean square error (MSE) versus
approximating the L.H.S.~of \eqref{eq:tele}. The way in which this is achieved to construct couplings of $(\Psi_l,\Psi_{l-1})$, $l\in\{2,\dots,L\}$, that is, joint densities $\check{\Psi}_l$, on $\mathsf{U}_l\times\mathsf{U}_{l-1}$ such that $\int_{\mathsf{U}_l}\check{\Psi}_l(u_l,u_{l-1})du_l=\Psi_{l-1}(\tilde{u}_{l-1})$ and $\int_{\mathsf{U}_{l-1}}\check{\Psi}_l(u_l,u_{l-1})du_{l-1}=\Psi_{l}(u_{l})$.

The MLMC estimator can be constructed in the following way.
\begin{enumerate}
\item{Sample $U_1^1,\dots,U_1^{P_1}$ i.i.d.~from $\Psi_1$, with $P_1\in\mathbb{N}$ given.}
\item{For $l\in\{2,\dots,L\}$ independently of all other random variables, sample \\ $(U_l^1,\tilde{U}_{l-1}^1),\dots,(U_l^{P_l},\tilde{U}_{l-1}^{P_{l}})$ i.i.d.~from $\check{\Psi}_l$, with $P_l\in\mathbb{N}$ given.}
\end{enumerate}
Then, we have the approximation
\begin{equation}\label{eq:psi_ml}
\Psi_L^{ML}(\varphi) := \frac{1}{P_1}\sum_{i=1}^{P_1}\varphi(U_1^i) + \sum_{l=2}^L \frac{1}{P_l}\sum_{i=1}^{P_l}\{\varphi(U_l^i)-\varphi(\tilde{U}_{l-1}^i)\}.
\end{equation}
Now, alternatively, one can use i.i.d.~samples $U_L^1,\dots,U_L^P$ from $\Psi_L$ to approximate $\Psi_L(\varphi)$:
\begin{equation}\label{eq:psi_iid}
\textcolor{black}{\Psi^{IID}_L(\varphi) := \frac{1}{P}\sum_{i=1}^{P}\varphi(U_L^i).}
\end{equation}
In either scenario one would have a standard, variance (assuming it exists) plus square bias decomposition of the MSE; e.g.~for the MLMC estimator:
\begin{equation}
\label{eq:MSE}
\mathbb{E}[(\Psi_L^{ML}(\varphi)-\Psi(\varphi))^2] = \mathbb{V}\textrm{ar}[\Psi_L^{ML}(\varphi)] +
[\Psi_L-\Psi](\varphi)^2,
\end{equation}
with $\mathbb{V}\textrm{ar}$ representing the variance operator. The bias is of course the same for both estimators \eqref{eq:psi_ml} and \eqref{eq:psi_iid} and so, if there is
to be an advantage for using \eqref{eq:psi_ml} under an MSE criterion, it would be via the variance.
The variance of \eqref{eq:psi_ml} is
\begin{equation}
\label{eq:varr}
\mathbb{V}\textrm{ar}[\Psi_L^{ML}(\varphi)] = %\Bigg(
\frac{\mathbb{V}\textrm{ar}[\varphi(U_1^1)]}{P_1}+\sum_{l=2}^L\frac{
\mathbb{V}\textrm{ar}[\varphi(U_l^1)-\varphi(U_{l-1}^1)]
}{P_l},
%\Bigg),
\end{equation}
whereas for \eqref{eq:psi_iid}
$$
\mathbb{V}\textrm{ar}[\Psi^{IID}_L(\varphi)] = \frac{\mathbb{V}\textrm{ar}[\varphi(U_L^1)]}{P}.
$$
Now, if the couplings $\check{\Psi}_l$ are constructed so that $\mathbb{V}\textrm{ar}[\varphi(U_l^1)-\varphi(U_{l-1}^1)]$ falls sufficiently quickly with $l$ then
it can be possible to achieve an MSE (as in \eqref{eq:MSE}) for the estimator \eqref{eq:psi_ml} which is of the same order as \eqref{eq:psi_iid} except for a cost that is lower.
This is characterized in the following result. Below $\mathrm{Cost}(\check{\Psi}_l)$ is the cost of one sample from $\check{\Psi}_l$.

 \begin{theorem}[Giles \cite{MBG08}]
\label{thm:VMLMC} Suppose that there exists constants
$(\alpha,\beta,\gamma) \in \mathbb{R}_+^3$ with $\alpha \geq
\frac{\min(\beta,\gamma)}{2}$ such that
\begin{itemize}
\item[(i)] $ |\Psi_l(\varphi)  - \Psi(\varphi)| =\mathcal{O}(2^{-\alpha l})$.
\item[(ii)] $\mathbb{V}\textrm{\emph{ar}}[\varphi(U_l^1)-\varphi(U_{l-1}^1)]=\mathcal{O}(2^{-\beta l})$.
\item[(iii)] $\mathrm{Cost}(\check{\Psi}_l)=\mathcal{O}(2^{\gamma l})$.
\end{itemize}
Then for any $\varepsilon < 1$ and $L := \lceil \log(1/\varepsilon) \rceil$,
there exists $(P_{1},\dots,P_L) \in \mathbb{N}^{L}$ such
that
$$
\mathrm{MSE} =\mathbb{E}[(\Psi_L^{ML}(\varphi)-\Psi(\varphi))^2]=\mathcal{O}(\varepsilon^2),
$$
and
\begin{equation}\label{eq:mlmcCost}
\mathrm{Cost(MLMC)} := \sum_{l=0}^L P_l C_l =
\begin{cases} \mathcal{O}(\varepsilon^{-2}), \quad &\mathrm{if} \ \beta
>\gamma,\\ \mathcal{O}(\varepsilon^{-2}( \log \varepsilon)^2), \quad
&\mathrm{if} \ \beta =\gamma,\\
\mathcal{O}(\varepsilon^{-2-\frac{(\gamma-\beta)}{\alpha}}), \quad
&\mathrm{if} \ \beta <\gamma,
\end{cases}
\end{equation}
\textcolor{black}{where $C_l$ is the cost for one sample of $\varphi(U_l^i)-\varphi({U}_{l-1}^i)$}.
\end{theorem}
{\color{black}
We remark that the cost in \eqref{eq:mlmcCost} can be lower than that of using the estimator \eqref{eq:psi_iid}, depending upon the parameters $(\alpha,\beta,\gamma)$, and simultaneously, the estimator \eqref{eq:mlmcCost}  having an MSE of $\mathcal{O}(\varepsilon^2)$. The case of $\beta > \gamma$ is referred to as the optimal, or canonical, rate of convergence, i.e. $\mathcal{O}(\varepsilon^{-2})$, in other words this is the best
rate that one can obtain.
In the original work of Giles \cite{MBG08} the methodology of MLMC was motivated and applied to diffusion processes
with applications in financial mathematics. In \cite{MBG08}, $\beta \in \mathbb{R}$ relates to the strong rate of convergence, and $\alpha\in\mathbb{R}$ the weak
rate of convergence. In our context, as we do not work with diffusion processes, not only is strong and weak convergence not required, i.i.d.~sampling of couplings is not achievable in our context (to be defined below) and thus an alternative methodology to reducing the variance of the higher level $l$ differences are
required.}

\subsubsection{Multilevel Neural Networks}

%{\color{black}
We now consider the question of choosing the dimension of $\theta$,
in particular $n_d$ for $d\in\{1,\dots,D-1\}$ (since $n_0$ and $n_D$ are fixed to the input and output dimensions). % and $D$.
For simplicity, here we assume $D$ is fixed and the NN width is the same for all layers other than the input or the output, i.e.
$n_d=n_{d'}$ for  $d,d' \in \{1,\dots,D-1\}$,
but dependent upon some resolution parameter $l\in\mathbb{N}$.
We can now re-define $n_l=2^l$, noting that $n_0=n$ and $n_D=m$,
so those variables are no longer needed.
We will denote the corresponding vector of parameters by
\textcolor{black}{$$\theta_l := \left((A_{1}^l,b_{1}^l),\dots,(A_{D}^l,b_{D}^l)\right)
\in \Theta_l :=%\bigotimes_{d=2}^{D-1}
\{\mathbb{R}^{n_l\times n}\times\mathbb{R}^{n_l}\}\otimes \{ \mathbb{R}^{n_l\times n_{l}}\times\mathbb{R}^{n_l}\}^{\otimes D-2} \otimes\{\mathbb{R}^{m\times n_{l}}\times\mathbb{R}^{m}\}.$$
We place a prior distribution $\overline{\pi}_{l}$ on $\Theta_l$, and note that its properties
are crucial for determining whether the multilevel method works --
in particular it must place vanishing mass on rows and columns $n_l$ as $l \rightarrow \infty$. 
%We will postpone such a discussion to 
This will be made precise in Section \ref{sec:tnn}.}
%with prior $\overline{\pi}_{l}(\theta_l)$, and the parameter space is $\Theta_l=%\bigotimes_{d=2}^{D-1}
%\{\mathbb{R}^{n_l\times n}\times\mathbb{R}^{n_l}\}\otimes \{ \mathbb{R}^{n_l\times n_{l}}\times\mathbb{R}^{n_l}\}^{\otimes D-2} \otimes\{\mathbb{R}^{m\times n_{l}}\times\mathbb{R}^{m}\}$, thus $\theta_l \in \Theta_l $.
The NN's output function $f(x,\theta)$ 
%(see 
\eqref{eq:DNN} for resolution $l$ is denoted by  $f_l(x,\theta_l)$, the  likelihood by
$p_l(y_{1:N}|\theta_l)$, and the posterior distribution by
\begin{equation}\label{eq:bnn_approx}
\pi_l(\theta_l | y_{1:N}) \propto p_l(y_{1:N}|\theta_l)\overline{\pi}_{l}(\theta_l) \, ,
\end{equation}
where $p_l(y_{1:N}|\theta_l)$ is the likelihood function using $\theta_l$ parameters as in \eqref{eq:reglikelihood},
and \eqref{eq:classlikelihood} for the second example discussed.
We view this posterior (and the corresponding $f_l(x,\theta_l)$) as a finite approximation of the posterior associated to the non-parametric limiting
DNN as $l \rightarrow \infty$, assuming it exists. 
%and denoted $\pi_{\star}$.

%\begin{remark}
%\textcolor{blue}{We note that thus far, we have introduced the concept of a multilevel deep neural network. An important quantity is the prior $\bar{\pi}_l(\theta_l)$, which is
%relevant to our study for the posterior. We will postpone such a discussion to Section \ref{sec:tnn}.}
%\end{remark}

\section{Algorithm and Main Result}\label{sec:algo_math}

In this section we introduce our methodology related to the Bayesian machine learning tasks, which we numerically test in Section \ref{sec:numerics}.
Namely we consider the ML sequential Monte Carlo method, and present it in our given framework. This will lead to the our main mathematical result, which
is the convergence of our ML estimator, provided in terms of a bound on the MSE.

{\color{black}
We begin this section by firstly presenting our algorithm for approximating 
%quantities 
functional expectations, with respect to the posterior $\pi$, such as
%for $(x,L)\in\mathsf{X}\times\mathbb{N}$ and $$
%given:
$$
\mathbb{E}_{\pi}[f(x,\theta)],
$$
as well as some mathematical results which justify their implementation.
The choice of 
%One relevant example of 
posterior predictive expectation as a quantity of interest
is motivated by the goal of our inference procedure, i.e.
to make predictions related to the output of the neural network $f$.
The form of our results are such that they extend trivially to
objective functions $\varphi \circ f$, for Lipschitz $\varphi$,
for example. 
We note that 
our results extend to a whole class of 
%there are other 
possible test functions %one can consider 
beyond the neural network itself.
The results will be presented first for the posterior predictive,
%functions, which will be presented 
and then extended as a corollary.
%However we aim to explore this for future work. }
Noting that 
$\mathbb{E}_{\pi}[f(x,\theta)] = 
\mathbb{E}_{\pi_L}[f_L(x,\theta_L)] + 
(\mathbb{E}_{\pi}[f(x,\theta)]-\mathbb{E}_{\pi_L}[f_L(x,\theta_L)])$,
our main objective is to construct a multilevel Monte Carlo estimator:
\begin{equation}\label{eq:ml_id_desc}
\mathbb{E}_{\pi_L}[f_L(x,\theta_L)] = \sum_{l=2}^L\left\{\mathbb{E}_{\pi_l}[f_l(x,\theta_l)] - \mathbb{E}_{\pi_{l-1}}[f_{l-1}(x,\theta_{l-1})]\right\}  + \mathbb{E}_{\pi_1}[f_1(x,\theta_1)],
\end{equation}
whose variance matches the discretization bias 
$(\mathbb{E}_{\pi}[f(x,\theta)]-\mathbb{E}_{\pi_L}[f_L(x,\theta_L)])$,
by approximating each summand on the R.H.S., as well as
$ \mathbb{E}_{\pi_1}[f_1(x,\theta_1)]$, using a suitable simulation method.
Then we will show that the computational cost for doing so, versus, simply approximating
$\mathbb{E}_{\pi_L}[f_L(x,\theta_L)]$ can be lower, when seeking to achieve a pre-specified mean square error.}
\textcolor{black}{The variance reduction from using MLMC, in the context of \eqref{eq:ml_id_desc}, can be seen through \eqref{eq:varr}, where the variance between the couplings at high-levels 
is reduced. For classic MLMC, levels correspond to the time-discretization of diffusion processes, while here they are related to the width, which we will discuss in Section \ref{sec:tnn}.
It is well-known within neural networks \cite{calib}, that various metrics can be improved on, by considering a larger width.}

\subsection{Algorithm}

The approach we construct follows that in \cite{beskos2017multilevel, beskos2018multilevel} and in order \textcolor{black}{to directly apply some of its results},  we introduce additional notation.
%{\color{blue} Capital letters will be used to denote {\em random variables}, e.g. $U$, whereas lowercase letters will be used to denote {\em realizations}, e.g. $u$.}
Throughout this exposition $(x,L)\in\mathsf{X}\times\mathbb{N}$ are fixed and given.
{\color{black} Define the state-spaces for $l\in\{2,\dots,L\}$
\begin{eqnarray*}
\tilde{\Theta}_l &:= & \{\mathbb{R}^{(n_l-n_{l-1})\times n}\times\mathbb{R}^{n_l-n_{l-1}}\}\otimes \{ \mathbb{R}^{(n_l-n_{l-1})\times n_{l-1}}
\times\mathbb{R}^{n_l\times(n_l-n_{l-1})}
\times \\
& &
\mathbb{R}^{n_l-n_{l-1}}\}^{\otimes D-2} \otimes\{\mathbb{R}^{m\times (n_{l}-n_{l-1})}\},
\end{eqnarray*}
and we set $\tilde{\Theta}_1=\Theta_1$; the significance of these spaces will be made clear below.
Set $\theta_1=\tilde{\theta}_{1}\in\Theta_1$ and for $l\in\{2,\dots,L\}$
$$
\theta_l = (\tilde{\theta}_1,\dots,\tilde{\theta}_{l}) \in\Theta_l.
$$
The quantities $(\tilde{\theta}_2,\dots,\tilde{\theta}_{l})$ will be used to denote the increase in dimension between
$\theta_{l-1}$ and $\theta_{l}$ say,  which is pivitol to the algorithm that we will present.
We will suppress the data $(x_{1:N},y_{1:N})$ from the notation
and set for $l\in\{1,\dots,L\}$
$$
\pi_l(\theta_{l}) \propto p_l(y_{1:N}|\theta_{l})\overline{\pi}_l(\theta_{l}) =: \kappa_{l}(\theta_{l}).
$$
For any $l\in\{2,\dots,L\}$ and  $\theta_{l-1}\in\Theta_{l-1}$, let $q_l(\cdot|\theta_{l-1})$ be a positive probability density on $\tilde{\Theta}_l$ and $q_1(\theta_1)$ a positive probability density on $\Theta_1$.
Set $G_1(\theta_1) = \kappa_1(\theta_1)/q_1(\theta_1)$
and for $l\in\{2,\dots,L\}$
$$
G_l(\theta_l) = \frac{\kappa_{l}(\theta_l)}{\kappa_{l-1}(\theta_{l-1})q_{l}(\tilde{\theta}_{l}|\theta_{l-1})},
$$
define our weights.
Let $(K_l)_{l\in\{1,\dots,L-1\}}$ be a sequence of $(\pi_l)_{l\in\{1,\dots,L-1\}}-$invariant Markov kernels and define
for $l\in\{1,\dots,L-1\}$
\begin{equation}\label{eq:kernel}
M_l(\theta_{l},d\theta_{l+1}') = K_l(\theta_{l},d\theta_{l}')q_{l+1}(\tilde{\theta}_{l+1}|\theta_{l}')d\tilde{\theta}_{l+1},
\end{equation}
 where $\theta_{l+1}' = (\theta_{l}',\tilde{\theta}_{l+1})$ and $d\tilde{\theta}_{l}$ is the appropriate dimensional Lebesgue measure. {This kernel is invoked in Step 4 of Algorithm \ref{alg:mlsmc_dnn}.}
}
%Let $l\in\{0,\dots,L-1\}$ and set $\mathsf{E}_l=\Theta_{l+1}$. For $l\in\{2,\dots,L\}$, we will write $\theta_l=(\theta_{l-1},\tilde{\theta}_l)$, where
%$\tilde{\theta}_l\in\Theta_l\setminus\Theta_{l-1}$, also set $\theta_1=\tilde{\theta}_1\in\Theta_1$. We seek to connect these model notations to that used in
%Feynman-Kac formulae, so we will define for $l\in\{0,\dots,L-1\}$:
%$$
%u_l := (\tilde{\theta}_1,\dots,\tilde{\theta}_{l+1}) = (u_{l-1},\tilde{\theta}_{l+1})\in \mathsf{E}_l,
%$$
%and note that $u_l =\theta_{l+1}$.
%We will suppress the data $(x_{1:N},y_{1:N})$ from the notation
%and set for $l\in\{1,\dots,L\}$
%$$
%\pi_l(\theta_{l}) \propto p_l(y_{1:N}|\theta_{l})\overline{\pi}_l(\theta_{l}) =: \kappa_{l-1}(u_{l-1}).
%$$
%For any $l\in\{2,\dots,L\}$ and  $\theta_{l-1}\in\Theta_{l-1}$, let $q_l(\cdot|\theta_{l-1})$ be a positive probability density on $\Theta_l\setminus\Theta_{l-1}$ and $q_1(u_0)$ a positive probability density on $\mathsf{E}_0$.
%Set $G_0(u_0) = \kappa_0(u_0)/q_0(u_0)$
%and for $l\in\{1,\dots,L-1\}$
%$$
%G_l(u_l) = \frac{\kappa_{l}(u_l)}{\kappa_{l-1}(u_{l-1})q_{l+1}(\tilde{\theta}_{l+1}|\theta_l)}.
%$$
%Let $(K_l)_{l\in\{1,\dots,L-1\}}$ be a sequence of $(\pi_l)_{l\in\{1,\dots,L-1\}}-$invariant Markov kernels and define
%\begin{equation}\label{eq:kernel}
%M_l(u_{l-1},du_l) = K_l(u_{l-1},du_{l-1}')q_{l+1}(\tilde{\theta}_{l+1}|u_{l-1}')d\tilde{\theta}_{l+1},
%\end{equation}
% where $u_l = (u_{l-1}',\tilde{\theta}_{l+1})$ and $d\tilde{\theta}_{l+1}$ is the appropriate dimensional Lebesgue measure. {This kernel is invoked in Step 4 of Algorithm \ref{alg:mlsmc_dnn}.}

{\color{black}
Set $\eta_1(\theta_1)=q_1(\theta_1)$ and for $l\in\{2,\dots,L\}$
$$
\eta_l(\theta_l) = \pi_{l-1}(\theta_{l-1})q_l(\tilde{\theta}_l|\theta_{l-1}).
$$
The algorithm to be presented will provide a sample-based approximation of these probability densities
and in particular,  associated expectations.  In other words we will have access to $P\in\mathbb{N}$
samples $(\theta_l^1,\dots,\theta_l^P)\in\Theta_l^P$ so that
\begin{equation}
 \eta_l^{P}(\varphi_l) := \frac{1}{P}\sum_{i=1}^{P}\varphi_l(\theta_l^i) \, .
% \approx \pi_l(\varphi_l) \, .
\label{eq:eta_empirical}
\end{equation}
will be an almost surely convergent estimator of $\int_{\Theta_l}\varphi_l(\theta_l)\eta_l(\theta_l)d\theta_l$,
where $\varphi_l:\Theta_l\rightarrow\mathbb{R}$ is $\eta_l-$integrable.  In the subsequent exposition given
$(\theta_l^1,\dots,\theta_l^P)\in\Theta_l^P$  $\eta_l^P$ will denote the so-called $P-$empirical measure.
We stress that the superscript of $\theta_l^i$ is used to denote a simulated sample (the simulation is given in the algorithm to be described) and the subscript relates to the subscript associated to the posterior $\pi_l$.
%Furthermore, $q_l(\cdot|\theta_{l-1})$ be a positive probability density on $\tilde{\Theta}_l$.}
We are now in a position to define our algorithm which can approximate
expectations w.r.t.~the sequence of posteriors $(\pi_l)_{l\in\{1,\dots,L\}}$  
and this is given in Algorithm \ref{alg:mlsmc_dnn}. 

%{We introduce the `predictors' $\eta_{l}(u_{l})=\pi_{l}({u_{l-1}})q_{l+1}(\tilde{\theta}_{l+1}\vert u_{l-1} )$, with $\eta_0(u_0)$  being the prior $q_1(u_0)$, which are computed by Algorithm \ref{alg:mlsmc_dnn}. }
%We will describe below how {samples from these predictors} this can be used to approximate the multilevel identity \eqref{eq:ml_id_desc}. Let $l\in\{1,\dots,L-1\}$, $\varphi_l:\Theta_l\rightarrow\mathbb{R}$ be $\pi_l-$integrable; we use the short-hand
%$\pi_l(\varphi_l)=\int_{\Theta_l}\varphi_l(\theta_l)\pi_l(\theta_l)d\theta_l$.
%{Denote the $P_l-$empirical measure of the samples
%$U_l^1, 
%%=((u_{l-1}')^1,\tilde{\theta}_{l+1}^1),
%\dots,
%U_l^{P_l}$
%%=((u_{l-1}')^{P_l},\tilde{\theta}_{l+1}^{P_l})$
%from Step 4 of Algorithm \ref{alg:mlsmc_dnn} as 
%%$\eta_l^{P_l}$
%%can be used to approximate 
%%, to approximate 
%%$\pi_l(\varphi_l)$. 
%%we can use
%\begin{equation}
% \eta_l^{P_l}(\varphi_l) := \frac{1}{P_l}\sum_{i=1}^{P_l}\varphi_l(U_l^i) \, .
%% \approx \pi_l(\varphi_l) \, .
%\label{eq:eta_empirical}
%\end{equation}}
%%where one should recall the notation in \eqref{eq:ul_def} given in Algorithm \ref{alg:mlsmc_dnn} and, as a convention, we will use $\eta_l^{P_l}(\varphi_l)$ instead of $\pi_l^{P_l}(\varphi_l)$ from herein.
%%This estimate can be justified in the sense that there are 
%There are formal results which prove that $\eta_l^{P_l}(\varphi_l)$ will converge to  $\pi_l(\varphi_l)$, almost surely as $P_l\rightarrow+\infty$; see for instance \cite{delmoral}.

\begin{algorithm}[h]
\begin{enumerate}
\item{\textbf{Input}: $L\in\mathbb{N}$ the highest resolution and the number of samples at each level $(P_0,\dots,P_{L-1})\in\mathbb{N}^{L}$, with $+\infty>P_1\geq P_2\geq \cdots \geq P_{L}\geq 1$.}
\item{\textbf{Initialize}: {\color{black}For $i\in\{1,\dots,P_1\}$ independently sample $\theta_1^i$ using $q_1(\theta_1)d\theta_1$. Set $l=1$. 
Go to step 3..}}
\item{\textbf{Iterate}: {\color{black}If $l=L$ go to step 4.. Otherwise for $i\in\{1,\dots,P_{l+1}\}$ sample $
\theta_{l+1}^i|\theta_{l}^1,\dots,\theta_{l}^{P_{l}}$ independently using:
$$
\sum_{j=1}^{P_{l}}\frac{G_{l}(\theta_{l}^j)}{\sum_{s=1}^{P_{l}}G_{l}(\theta_{l}^s)}M_l(\theta_{l}^j,d\theta_{l+1}).
$$}
%{Let $\eta_l^{P_l}$ denote the empirical measure of $\{U_l^i\}_{i=1}^{P_l}$ (see \eqref{eq:eta_empirical}).
%We will use the notation
%%\begin{equation}\label{eq:ul_def}
%%u_l^i=(\tilde{\theta}_l^i(1),\dots,\tilde{\theta}_l^i(l+1)).
%%\end{equation}
%{
%\begin{equation}\label{eq:ul_def}
%u_l^i=((u_{l-1}')^i,\tilde{\theta}_{l+1}^i).
%\end{equation}
%}
%{\color{red}***SSS: Should the above be $(\tilde{\theta}_1^i,\dots,\tilde{\theta}_{l+1}^i)$ since the tilde is the increment? We can use notation of \eqref{eq:eta_empirical}. Also it is a good place to define what is meant by $\eta_l^{P_l}$. Also define $\eta_0^{P_0}$.***}
Set $l=l+1$ and go to the start of step 4.}
{\color{black} \item{\textbf{Output}: $(\theta_1^1,\dots,\theta_1^{P_1},\dots,\theta_L^1,\dots,\theta_L^{P_{L}})$, from which \eqref{eq:ml_est} is constructed.}}
\end{enumerate}
\caption{Multilevel Sequential Monte Carlo Sampler for Deep Neural Networks.}
\label{alg:mlsmc_dnn}
\end{algorithm}

{\color{black}
Now, recalling \eqref{eq:ml_id_desc}, our objective is to approximate the difference, for $(x,l)\in\mathsf{X}\times \{2,\dots,L\}$:
$$
\mathbb{E}_{\pi_l}[f_l(x,\theta_l)] - \mathbb{E}_{\pi_{l-1}}[f_{l-1}(x,\theta_{l-1})] =: \pi_l(f_l) - \pi_{l-1}(f_{l-1}) ,
$$
where $\mathbb{E}_{\pi_l}$ denotes expectation w.r.t.~$\pi_l$. Now, we have the simple identity
\begin{eqnarray}
\mathbb{E}_{\pi_l}[f_l(x,\theta_l)] - \mathbb{E}_{\pi_{l-1}}[f_{l-1}(x,\theta_{l-1})] & = &  \mathbb{E}_{\pi_{l-1}\otimes q_l}\left[\frac{\kappa_{l}(\theta_l)Z_{l-1}}{\kappa_{l-1}(\theta_{l-1})q_l(\tilde{\theta}_l|\theta_{l-1})Z_l}f_l(x,\theta_l)-f_{l-1}(x,\theta_{l-1})\right] \nonumber \\
& = &   \mathbb{E}_{\pi_{l-1}\otimes q_l}\left[\frac{Z_{l-1}}{Z_l}G_{l}(\theta_{l})f_l(x,\theta_l)-f_{l-1}(x,\theta_{l-1})\right] ,
\label{eq:ml_id}
\end{eqnarray}
where $\mathbb{E}_{\pi_{l-1}\otimes q_l}$ denotes expectation w.r.t.~$\pi_{l-1}(\theta_{l-1})q_l(\tilde{\theta}_l|\theta_{l-1})=\eta_l$ and for any $l\in\{1,\dots,L\}$ and $Z_l=\int_{\Theta_l}\kappa_{l}(\theta_l)d\theta_l$.
One can approximate the R.H.S.~of \eqref{eq:ml_id}
as
$$
\frac{\eta_{l}^{P_{l}}(G_{l}f_l)}{\eta_{l}^{P_{l}}(G_{l})}-\eta_{l}^{P_{l}}(f_{l-1}).
$$
The justification of this estimator is informally as follows. $\eta_{l}^{P_{l}}(f_{l-1})$ will converge in probability (as $P_{l}\rightarrow\infty$) to $\pi_{l-1}(f_{l-1})=\int_{\Theta_{l-1}}f_{l-1}(x,\theta_{l-1})\pi_{l-1}(\theta_{l-1})d\theta_{l-1}$ (see \cite{delmoral}),
which justifies the term $\eta_{l}^{P_{l}}(f_{l-1})$.
Then $\eta_{l}^{P_{l}}(G_{l})$ will converge to
\begin{eqnarray*}
\int_{\Theta_l}\pi_{l-1}(\theta_{l-1})q_l(\tilde{\theta}_l|\theta_{l-1})\frac{\kappa_{l}(\theta_l)}{\kappa_{l-1}(\theta_{l-1})q_l(\tilde{\theta}_l|\theta_{l-1})}d\theta_l & = &
\frac{1}{Z_{l-1}}\int_{\Theta_l}\kappa_{l}(\theta_l)d\theta_l \\
& = & \frac{Z_l}{Z_{l-1}}.
\end{eqnarray*}
Similarly $\eta_{l}^{P_{l}}(G_{l}f_l)$ converges to
$$
\int_{\Theta_l}\pi_{l-1}(\theta_{l-1})q_l(\tilde{\theta}_l|\theta_{l-1})\frac{\kappa_{l}(\theta_l)}{\kappa_{l-1}(\theta_{l-1})q_l(\tilde{\theta}_l|\theta_{l-1})}f_l(x,\theta_l)d\theta_l.
$$
which yields the appropriate identity on the R.H.S.~of \eqref{eq:ml_id}. As a result of this exposition, one can use the following approximation of $\pi_L(f_L)$:
\begin{equation}\label{eq:ml_est}
\widehat{\pi}_L(f_L) = \sum_{l=2}^L\left\{
\frac{\eta_{l}^{P_{l}}(G_{l}f_l)}{\eta_{l}^{P_{l}}(G_{l})}-\eta_{l}^{P_{l}}(f_{l-1})
%\frac{\eta_{l-1}^{P_{l-1}}(G_{l-1}f_l)}{\eta_{l-1}^{P_{l-1}}(G_{l-1})}-\eta_{l-1}^{P_{l-1}}(f_{l-1})
\right\} + \frac{\eta_{1}^{P_{1}}(G_{1}f_1)}{\eta_{1}^{P_{1}}(G_{1})}.
\end{equation}
}

\subsection{Mathematical Result}

We consider the convergence of \eqref{eq:ml_est} in the case $m=1$; this latter constraint can easily be removed with only minor changes to the subsequent notations and arguments. The analysis of Algorithm \ref{alg:mlsmc_dnn} has been considered in \cite{beskos2017multilevel,beskos2018multilevel}. However, there are some nuances that
need to be adapted for the context under study. Throughout, we will suppose that for each $l\in\{2,3,\dots\}$ we have
chosen $q_l$ {\color{black} so that for each $\theta_l\in\Theta_l$
\begin{equation}\label{eq:cond_disc}
\overline{\pi}_l(\theta_l) = \overline{\pi}_{l-1}(\theta_{l-1})q_l(\tilde{\theta}_l|\theta_{l-1}).
\end{equation}
This means that
$$
G_{l}(\theta_{l}) = \frac{p_l(y_{1:N}|\theta_l)}{p_{l-1}(y_{1:N}|\theta_{l-1})}.
$$}
This convention is not entirely necessary, but it will facilitate a simplification of the resulting calculations.
{\color{black} We discuss this point after stating our assumptions.}
We will use the following assumptions, which are often used in the analysis of approaches of the type described in Algorithm \ref{alg:mlsmc_dnn}. {\color{black} Below $\mathcal{B}(\Theta_l)$ is the Borel $\sigma-$field associated to $\Theta_{l}$. We also consider,  as is the case for the DNN model and priors,  that the indices associated to $p_l,\overline{\pi}_l,f_l,q_l,K_l$ can be indefinitely extended (i.e.~beyond $L$). 

\begin{hypA}\label{ass:1}
\begin{enumerate}
\item{There exists a $0<\underline{C}<\overline{C}<+\infty$ such that for any $x\in\mathsf{X}$:
\begin{align}
\label{eq:min}
\inf_{l\in\mathbb{N}}\inf_{\theta_{l}\in\Theta_l}\min\{p_l(y_{1:N}|\theta_{l}),\overline{\pi}_l(\theta_{l})\} & \geq  \underline{C} \\
\label{eq:max}
\sup_{l\in\mathbb{N}}\sup_{\theta_{l}\in\Theta_{l}}\max\{p_l(y_{1:N}|\theta_{l}),\overline{\pi}_l(\theta_{l}),f_l(x,\theta_{l})\} & \leq  \overline{C}.
\end{align}
}
\item{There exists a $\rho\in(0,1)$ such that for any $(l,\theta_{l},\theta_{l}',\mathsf{B})\in\mathbb{N}\times
\Theta_l^2\times\mathcal{B}(\Theta_{l+1})$:
$$
\int_{\textcolor{black}{\mathsf{B}}} M_l(\theta_{l},d\theta_{l+1}) \geq \rho \int_{B_l} M_l(\theta_{l}',d\theta_{l+1}).
$$
}
\item{There exists a $C<\infty$ and \textcolor{black}{$\beta >0$} such that for any $(l,x)\in\{2,3,\dots\}\times\mathsf{X}$:
$$
\int_{\Theta_{l}}(f_l(x,\theta_{l})-f_{l-1}(x,\theta_{l-1}))^2\overline{\pi}_l(\theta_l)d\theta_l \leq Cn_l^{-\beta}.
$$
}
\item{There exists a $C<\infty$ 
%and $\beta>0$ 
such that for any $(l,\theta_{l},x)\in\{2,3,\dots\}\times\Theta_{l}\times\mathsf{X}$, %$u_{l-1}=(u_{l-2},\tilde{\theta}_l)$:
$$
|p_l(y_{1:N}|\theta_{l})-p_{l-1}(y_{1:N}|\theta_{l-1})| \leq C|f_l(x,\theta_{l-1})-f_{l-1}(x,\theta_{l-1})|.
$$
}
\item{There exists a $r\geq 3$, $C<\infty$, possibly depending on $r$, such that for $\beta>0$ as in 3.~and any $(l,\theta,x)\in\{2,3,\dots\}\times\Theta_{l-1}\times\mathsf{X}$:
$$
\left(\int_{\Theta_{l}}|f_l(x,\theta_{l})-f_{l-1}(x,\theta_{l-1})|^r M_{l-1}(\theta,d\theta_{l})\right)^{1/r} \leq Cn_l^{-\beta/2}.
$$
}%with the notation $u_{l-1}=(u_{l-2},\tilde{\theta}_l)$.}
\end{enumerate}
\end{hypA}
}
{\textcolor{black}{Admittedly, the minimum bound in \eqref{eq:min}} is not satisfied for most applications but does significantly simplify the proof. 
{\color{black} This assumption can be relaxed by leveraging the heavy machinery of multiplicative drift conditions as in \cite{kontoyiannis2005large}.
This machinery was used to relax the assumptions for the original MLSMC algorithm in 
\cite{del2017multilevel}. 
This is possible to do, but it is highly technical and would detract from the current presentation, so it is deferred to future work.} 
{\color{black}To deal with the general case of $G_l$ (i.e.~without the condition \eqref{eq:cond_disc})
one could specify $q_l$ so that (A\ref{ass:1}) 4.~can be modified to: There exists a $C<\infty$ 
such that for any $(l,\theta_{l},x)\in\{2,3,\dots\}\times\Theta_{l}\times\mathsf{X}$,
$$
|p_l(y_{1:N}|\theta_{l})\overline{\pi}_l(\theta_l)-p_{l-1}(y_{1:N}|\theta_{l-1})
\overline{\pi}_{l-1}(\theta_{l-1})q_l(\tilde{\theta}_l|\theta_{l-1})
| \leq C|f_l(x,\theta_{l-1})-f_{l-1}(x,\theta_{l-1})|.
$$
Such a specification may require considerable work, but essentially the point is that one must be able to sufficiently control $|G_l(\theta_l)-1|$. This is what the combination of \eqref{eq:cond_disc} and (A\ref{ass:1}) 4.~allow us to do with the current proof.}
We then have the following result, whose proof is in Supplementary material \ref{app:proofs}.

\begin{proposition}\label{prop:main_res}
Assume (A\ref{ass:1}). Then there exists a $C<\infty$ and $\zeta\in(0,1)$ such that for any $L\in\{2,3,\dots\}$
$$
\mathbb{E}[(\widehat{\pi}_L(f_L)-\pi_L(f_L))^2] \leq C\left(\frac{1}{P_1} + \sum_{l=2}^{L} \frac{1}{P_{l}n_l^{\beta}}+ \sum_{l=2}^{L-1}\sum_{q=l+1}^L \frac{1}{(n_ln_q)^{\beta/2}}
\left\{\frac{\zeta^{q-1}}{P_{l}} + \frac{1}{P_{l}^{1/2}P_{q}}\right\}
\right).
$$
\end{proposition}

{\color{black} Connecting to the discussion of Section \ref{sec:mlbnn}  if $\beta$ is large enough (we have $\beta=3$ later,
which is large enough) then as $\zeta=1$ one can choose $P_1,\dots,P_L$ to achieve a MSE of $\mathcal{O}(\varepsilon^{2})$ for the optimal cost of $\mathcal{O}(\varepsilon^{-2})$.
}

%\begin{proposition}\label{prop:main_res}
%\textcolor{blue}{Assume (A\ref{ass:1}). Then there exists an arbitrary constant $C < \infty$, whose exact value may change line by line, but which does not depend on the 
%underlying $C$ in Corollary \ref{cor:main_res}, and $\zeta\in(0,1)$ such that for any $L\in\{2,3,\dots\}$}
%$$
%\mathbb{E}[(\widehat{\pi}_L(f_L)-\pi_L(f_L))^2] \leq C\left(\frac{1}{P_0} + \sum_{l=2}^L \frac{1}{P_{l-1}n_l^{\beta}}+ \sum_{2\leq l<q\leq L} \frac{1}{(n_ln_q)^{\beta/2}}
%\left\{\frac{\zeta^{q-1}}{P_{l-1}} + \frac{1}{P_{l-1}^{1/2}P_{q-1}}\right\}
%\right).
%$$
%{\color{black} Consequently, if the cost per forward simulation is proportional
%to $n_l^\gamma$, with $\beta>\gamma$, then we are able to achieve 
%$$
%\mathbb{E}[(\widehat{\pi}_L(f_L)-\pi(f))^2] \leq C\varepsilon^2 \, , 
%$$
%for a cost proportional to $\varepsilon^{-2}$.}
%\end{proposition}

%{\color{blue}
%\begin{corollary}\label{cor:main_res}
%Assume (A\ref{ass:1}), and additionally that 
%%$\varphi_l \rightarrow \varphi$, where 
%a family of functions 
%$\varphi_l: \Theta_l \rightarrow \bbR$
%satisfy $\sup_l \|\varphi_l\|_\infty <{C}$
%and assumptions (A\ref{ass:1}).3 and (A\ref{ass:1}).5.
%Then we are able to achieve 
%$$
%\mathbb{E}[(\widehat{\pi}_L(\varphi_L)-\pi(\varphi))^2] \leq \textcolor{blue}{{C}}\varepsilon^2 \, , 
%$$
%for a cost proportional to $\varepsilon^{-2}$.
%\end{corollary}
%}

\section{Trace class priors}
\label{sec:tnn}

In this section we briefly discuss our trace class priors, which we aim to analyze
and motivate for numerical experiments later within the manuscript. We will begin with
a formal definition, before providing a result related to convergence of the predictive
model $f$. Given the rate obtained, we will then verify this rate through the demonstration
of a simple numerical experiment.

The neural network priors we consider are the trace class neural networks,
first proposed by Sell et al. \cite{sell2020dimension}. These priors were introduced
to mimic Gaussian priors $\overline{\pi}_0 \sim \mathcal{N}(0,\mathcal{C})$, for function-space inverse problems,
where a common way to simulate Gaussian random fields is to use the Karhunen-Lo\`{e}ve expansion
{\color{black}
\begin{equation}
\label{eq:kle}
 f = \sum_{j \in \mathbb{Z}^+} \sqrt{\lambda_j} \iota_j \Phi_j, \qquad \iota_j \sim \mathcal{N}(0,1),
\end{equation}
where $(\lambda_j,\Phi_j)_{j \in \mathbb{Z}^+}$ are the associated eigenbasis of the covariance \textcolor{black}{operator} $\mathcal{C}$,
and $\{\iota_j\}_{j \in \mathbb{Z}^+}$} is Gaussian white noise. For a detail description of the \textcolor{black}{derivation of \eqref{eq:kle}}
and its application to stochastic numerical problems\textcolor{black}{,} we refer the reader to \cite{LPS14}.
An issue that can arise with using \eqref{eq:kle} is that using such priors do not scale well with
high-dimensional problems. 
This acts as the initial motivation for trace-class neural network (TNN) priors.
%, where
The work of Sell et al. \cite{sell2020dimension} provided a justification for such priors,  
%Specifically if one considers  then 
where the variances are 
%can be 
summarized in a trace-class diagonal covariance operator $\mathcal{C}$ as 
%an infinite width network  
$l \rightarrow \infty$.
Such a result is presented as [Theorem 1., \cite{sell2020dimension}].
%If we consider the 
The TNN prior for $\theta_l := \left((A_{1}^l,b_{1}^l),\dots,(A_{D}^l,b_{D}^l)\right)$ is defined as follows
%then for $d=1,\dots,D$, the 
%associated weights and bias are defined as
\begin{equation}\label{eq:tnn}
\textcolor{black}{A^l_{ij,d} \sim \mathcal{N}(0, |ij|^{-\alpha}),
\quad b_{i,d}^{l} \sim \mathcal{N}(0, i^{-\alpha})} \, ,
\end{equation}
where at level $l$ a coupling can be constructed by letting
$A^l_{ij,d}=A^{l-1}_{ij,d}$ and $b_{i,d}^{l}=b_{i,d}^{l-1}$
for $i,j\in\{1,\dots,n_{l-1}\}$ and $d\in\{2,D-1\}$;
a similar assignment between layers $l$ and $l-1$ is to be adopted for
$(A_{1}^l,b_{1}^l)$ and $(A_{D}^l,b_{D}^l)$.
The notation $A_i,b_i$, for $i=1,\dots,D$, and
$\theta$, will be used to denote the limits
$\lim_{l\rightarrow \infty} A_i^l, b_i^l, \theta^l$.

{\color{black} The limiting TNN prior 
uses \eqref{eq:DNN} in place of 
\eqref{eq:kle}, with parameters defined in \eqref{eq:tnn}.} 
%with $u=f$,the notation in Section \ref{sec:model} and \eqref{eq:bnn_posterior}
%(except for $n_l=2^l$) and above in \eqref{eq:tnn},
%for $d=2,\dots,D$
%%is given in the form
%\begin{equation}
%\label{eq:tnn_final}
%g_d(x,\theta) = %\sum_{j=1}^{\infty}
%A_{d}\sigma(g_{d-1}(x)) + b_{d}.
%\end{equation}
The approximation at level $l$, $f_l$,
simply replaces all limiting objects with the level $l$ approximations
(with width $n_l=2^l$ for all hidden layers).
The tuning parameter $\alpha$ controls how much information one believes concentrates on the first nodes.
In the case of $\alpha>1$, we refer to the prior as trace-class, which results in the term trace-class neural network prior.

%As we are concerned with using 
%%these 
%%combining such priors to 
%the MLSMC sampler with TNN priors, 
%%it would be of interest to 
We must establish strong convergence of TNN priors (Assumption 3)
if we are to successfully implement the MLSMC method.
This is provided by the following proposition.
%%of such neural networks which use \eqref{eq:tnn}. 
%%Therefore t
%The following proposition
%%we provide a result 
%%{\color{black}establishing 
%{verifies Assumption 3}, 
%%in terms of the
%%parameters 
%$l$ and $\alpha$.
%, which is given through the following proposition. 

\begin{proposition}
\label{prop:tnnconv}
Assume that for all $z\in \bbR$,
$\sigma(z) \leq |z|$.
Then for $\alpha > 1/2$, the priors defined as in \eqref{eq:tnn}
are trace-class. Furthermore, let $x \in \sX$,
and consider the TNN $f_l(x,\theta_l)$ truncated at $n_l=2^l$ terms, for $l \in \mathbb{N}$,
with the limit denoted $f(x,\theta)=\lim_{l \rightarrow \infty}f_l(x,\theta_l)$, as described above.
Then there is a $C(x)<+\infty$ such that
\begin{equation}\label{eq:tnnconv}
\bbE\left[| f_l(x, \theta_l)  - f(x,\theta) |^2\right] \leq C 2^{-(2\alpha-1) l} \, .
\end{equation}
\end{proposition}
\begin{proof}

First assume that $D=2$, i.e. there is a single hidden layer, and $m=1$.
Let $x\in \sX$ be fixed, but allow the possibility that $n \rightarrow \infty$,
provided that $\sum_j |x_j|^2 < \infty$ almost surely.

Let $A_{k,1}^l$ denote the $k^{\rm th}$ row of $A_1^l$,
allowing the possibility that $l \rightarrow \infty$
(denoted simply $A_1$, as described above),
and observe that
\begin{equation}\label{eq:firstlevel}
\bbE\left[A_{k,1}^l x + b_{k,1}^l\right] = 0 \, , \qquad  
\bbE\left[(A_{k,1}^l x + b_{k,1}^l )^2\right] =
k^{-\alpha}\left( \sum_{j=1}^n j^{-\alpha} x_j^2 + 1 \right)
\leq k^{-\alpha}( |x|^2 + 1 ) \, ,
\end{equation}
where $|x|^2 := \sum_{j=1}^n x_j^2$.

Let $A_{k,2}^l$ denote the $k^{\rm th}$ entry of the vector
$A_2^l \in \bbR^{n_l}$, again allowing $l \rightarrow \infty$
(and denote the limit by $A_2$).
Using the shorthand notation
$\xi_{k,1}^l = A_{k,1}^l x + b_{k,1}^l$
(and $\xi_{k,1} = A_{k,1} x + b_{k,1}$),
we are concerned with
\begin{eqnarray}\label{eq:traceclass}
\bbE\left[f^2(x,\theta)\right] &=& \sum_{k=1}^\infty \bbE\left[A_{k,2}\right]^2 \bbE\left[\sigma(\xi_{k,1})^2\right] + \bbE\left[b_2\right]^2
\\ \nonumber
&\leq& \sum_{k=1}^\infty \bbE\left[A_{k,2}\right]^2 \bbE\left[\xi_{k,1}\right]^2 + 1 \\ \nonumber
&\leq & C \sum_{k=1}^\infty k^{-2\alpha} + 1\, .
\end{eqnarray}
This shows that the output has finite second moment for $\alpha>1/2$.
Furthermore, a simple extension of the above calculation shows that
the rate of convergence is given by
\begin{equation}\label{eq:trivial}
\bbE\left[|f_l(x,\theta_l)-f(x,\theta)|^2\right] \leq C \sum_{k\geq 2^l} k^{-2\alpha}
\leq C 2^{-(2\alpha-1)l} \, .
\end{equation}

Now, regarding multiple levels, we can proceed by induction.
Assume we have \eqref{eq:traceclass} at level $d$,
with $\xi^l_d \in \bbR^{n_l}$ denoting the output at level $d$
(e.g. $\xi_{i,2}$ replaces $f$ in \eqref{eq:traceclass},
for $i=1,\dots$).
As above, let
$A_{k,d+1}^l$ denote the $k^{\rm th}$ row of the matrix $A_{d+1}^l$.
Then, by iterating expectations,
we have a $C<+\infty$ such that
$$
\bbE\left[(A_{k,d+1}^l \xi^l_d + b_{k,d+1}^l )^2\right] \leq   C k^{-\alpha} \, ,
$$
as in \eqref{eq:firstlevel}.
This brings us back to \eqref{eq:traceclass} at level $d+1$,
and we are done for $D>2$.

{\color{black} The extension to finite $m>1$ is trivial by repeating the arguments above.
For example, the result \eqref{eq:trivial} holds for each of the $m$ outputs.}
\end{proof}

{The result above establishes Assumption (A\ref{ass:1}.3) for $\beta=2\alpha-1$.
It also provides (A\ref{ass:1}.5) for $r=2$ when
combined with (A\ref{ass:1}.1), (A\ref{ass:1}.2) and the definition \eqref{eq:kernel}.}

{\color{black} 

\begin{remark}[Strong Convergence]
%\textcolor{blue}{%We note that attaining t
Strong convergence is %result above is %crucial, in order 
required to %proceed with the MLMC theory.
use the MLMC \\ method, and is one attractive property of TNN.
Standard NNGP derived as the infinite-width limit of i.i.d. weights with variance
proportional to the inverse width do not have this property.
%This is achieved through the aid of the TNN prior which enable this through its structure. Therefore this can be viewed as another
%motivation for using TNN in the context of this work.}
\end{remark}

\begin{remark}[Smoothness]
\textcolor{black}{
%We would also like the remark, that for the prior
%structure in Eqn. \eqref{eq:tnn}, it is not just the parameter $\alpha$ that controls the smoothness of the function, but also $|ij|$. For example, if these values are high, then their
%is less support on the diagonals resulting in less smooth features. However for our experiments we do not require a high number related to the architecture. For very high-dimensional examples, 
%this may differ, however we leave this for future work.} 
Notice that the variance of the parameters of the TNN decays along both rows and columns
%making 
(hence the mass decays faster closer to the diagonal).
%twice as fast along the diagonal, due to the factor $|i|^{-2\alpha}$, 
%as compared to $|ij|^{-\alpha}$ for row or column $i$ when the other is 1).
Since the support decays as the width grows, the decay rate precisely controls the smoothness, and as a consequence
the level of the neuron also determines the smoothness level it captures:
low levels capture smooth features and high levels capture progressively more rough features.
Note that this constrained functional form
does not adversely affect the reconstruction, provided that the smoothness constraint $\alpha$
is not chosen too large to capture the underlying data-generating function.
We simply have a Gaussian prior with more fine-scale diagonal variance structure than a 
typical isotropic ``weight-decay'' prior $N(0,\sigma^2 I)$, 
so it provides a more structured ``width-decay" regularization.
Incidentally, this construction also guarantees stability of the weights with 
depth, so there is no constraint preventing infinite depth networks.}
\end{remark}

\begin{corollary}
\label{cor:tnnconv}
Let $\overline{A}^l, \overline{b}^l, \overline{\theta}_l$ 
denote the embedding of $\theta_l$
into $\Theta$, by setting $\overline{A}^l, \overline{b}^l$
to be 0 for rows higher than $n_l$. Then there is a $C>0$ s.t.
$$
\bbE\left[\| (\overline{A}^l,\overline{b}^l)_d - (A,b)_d \|^2\right] \leq C 2^{-(\alpha-1)l}\, , 
$$
for all layers $d\in\{1,\dots,D\}$, 
where $\|\cdot\|$ denotes the induced operator norm.
%\overline{b}^l$
\end{corollary}

\begin{remark}
The corollary above provides an important extension from
the posterior predictive, covered in the basic propositions \ref{prop:tnnconv}
and \ref{prop:main_res} to the important and challenging 
problem of inference over the {\em parameter posterior}.
\end{remark}

}

\begin{remark}[Mutation kernel]
\textcolor{black}{
%We note in our particular setup for 
For the implementation of Algorithm \ref{alg:mlsmc_dnn}, 
%our MLSMC sampler uses a 
the kernel $M_l$ 
%defined by 
in \eqref{eq:kernel} is defined as follows:
\begin{eqnarray}
K_l &=& Q_l^{(m)} \, ,\\
Q_l(\theta_l,\cdot) &:= & N(\sqrt{1-\delta^2}\theta^i_l , \delta^2 I_{}) \, , \\
q_{l+1}(\tilde{\theta}_{l+1}|\theta_l) &:=& \overline{\pi}_{l+1}(\tilde{\theta}_{l+1}) \, .
\end{eqnarray}
$Q_l$ is known as a preconditioned Crank-Nicholson (pCN) kernel \cite{cotter},
%of the form 
%$K_l(\theta,\theta') = 
%$\theta_l' = \sqrt{1-\vartheta^2}\theta^i + \vartheta \Xi_n$, 
%$\theta_l'|\theta \sim 
% \Xi_n$, 
%where $X_i \sim \bar{\theta}$. 
%The choice of $\vartheta \in [0,1]$, 
%The choice of The parameter 
and the stepsize $\delta \in [0,1]$ 
%where we choose it later as 
%$\vartheta=0.6$. 
will often be set to $\delta=0.6$. 
So the mutation of the level $l$ parameters involves $m$ steps of a basic pCN kernel.
%Finally, in \eqref{eq:kernel} $K_l = Q_l^{(m)}$, with $Q_l$ the pCN kernel, i.e. $m-$steps
%\textcolor{red}{Neil, is this correct? You take one single step of pCN?? 
%Or is it more like $K_l = Q_l^{(m)}$, with $Q_l$ the pCN kernel, i.e. $m-$steps.}
Notice that the prior is diagonal, so that it is easy to achieve \eqref{eq:cond_disc}.
In practical terms, this means we mutate the level $l$ parameters 
$(A^{l}, b^{l})$ according to pCN, 
and then we sample the new rows and columns of $(A^{l+1}, b^{l+1})$ from the prior.
The coupling is immediate because the first $n_l$ rows and columns are identical between 
the levels $l$ and $l+1$ for a given increment. 
The level $l+1$ parameters will then be mutated again for the next increment.}
%$q_{l+1}(\tilde{\theta}_{l+1}|\theta_l)=\overline{\pi}_{l+1}(\tilde{\theta}_{l+1})$
%The Markov kernel $K_l$ enters our algorithm through Eqn \eqref{eq:kernel}. }
\end{remark}

\subsection{Numerical Results}

To verify the rate which was obtained in  Proposition \ref{prop:tnnconv}, we provide a simple numerical
example which analyzes \eqref{eq:tnnconv}, in the context of our trace class priors. For our experiment
we consider a setup of $\gamma=2$ and $\alpha=2$, implying that our decay rate is $\beta=2\alpha-1=3$. We test this on both a 2 layer and
3 layer NN, and with a ReLU activation function $\sigma(z) = \max\{z,0\}$, and $\sigma(z) = \tanh(z)$ activation function.
Our experiment is presented in Figures \ref{fig:sum} - \ref{fig:sum_2l}.

We observe that in both figures, as a result of using TNN priors, our decay rate matches that of \eqref{eq:tnnconv}, where
we obtain the canonical rate of convergence.
As a side example, to show that this is not attained with other choices, we compare
this other choices of priors. This can be seen from Figure \ref{fig:sum_2l}, where we notice that the rates are not as expected, implying
that they are sub-canonical.
Therefore this provides a motivation in using trace class NN priors, coupled with MLMC,
 which we exploit in the succeeding subsection.

%Therefore this suggests
%that perhaps the priors developed in \cite{sell2020dimension} are the best way forward.-eps-converted-to.pdf}
\begin{figure}[h!]
\centering
\includegraphics[width=0.40\textwidth]{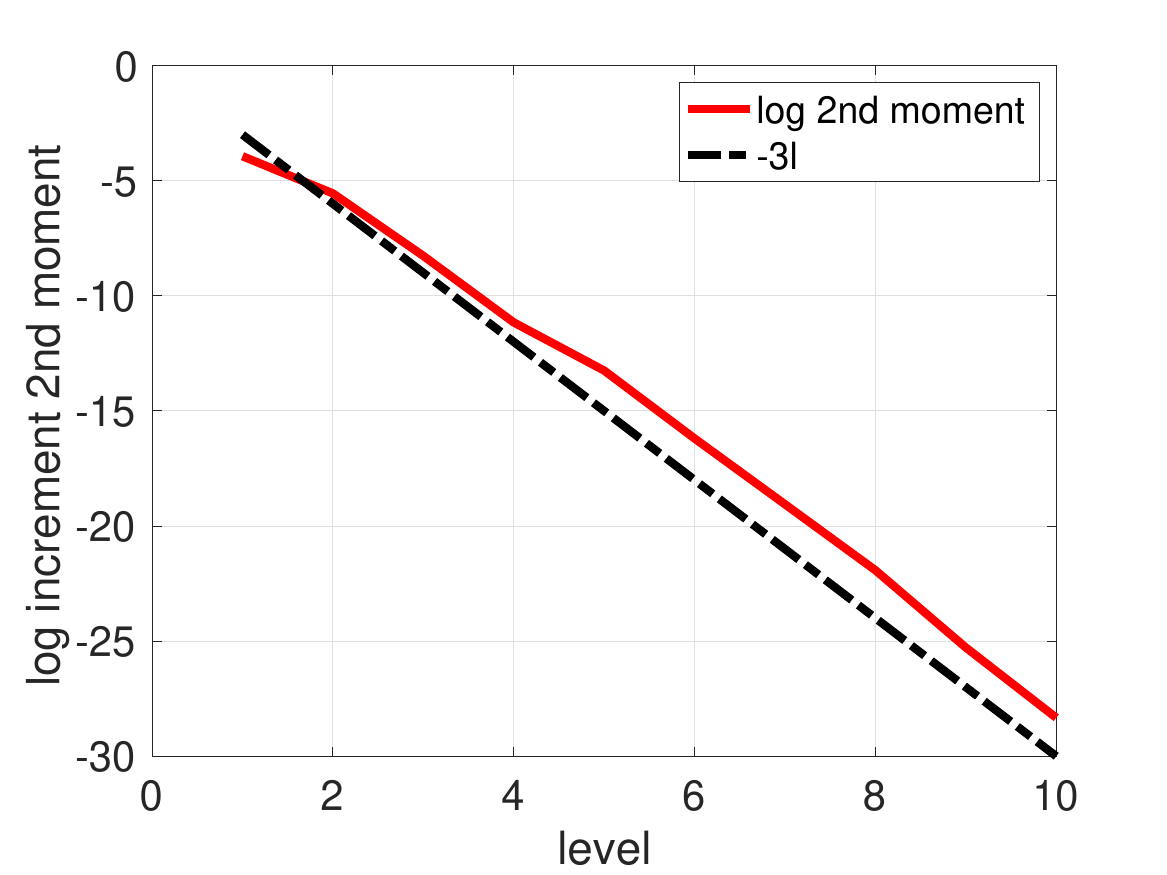}
\includegraphics[width=0.40\textwidth]{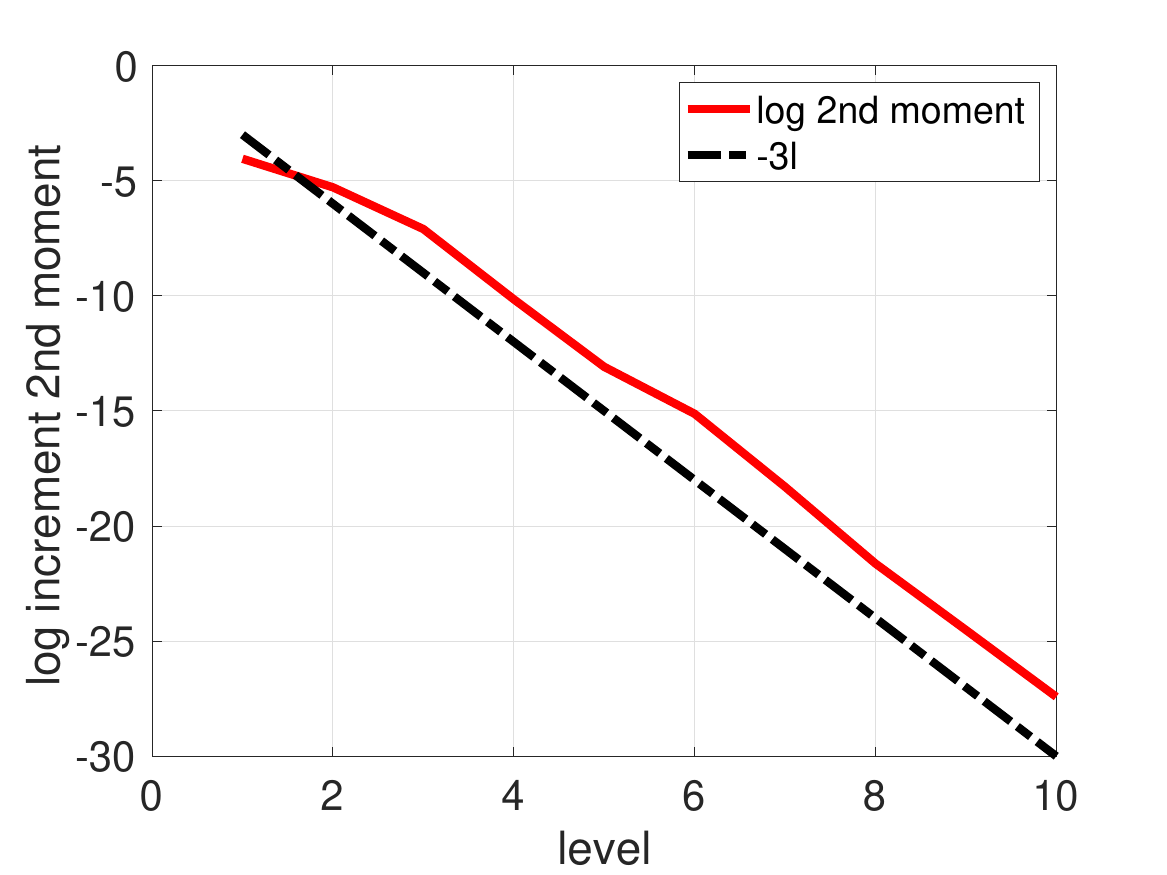}
    \caption{Increment 2nd moment vs. levels.
    The decay is $\cO(2^{-3l})$.
    %The number of weights, and hence the cost, is $\cO(2^{2l})$ {\color{red}***SSS: cost calculation needs explaining.***}.
    Therefore the variance decays faster than $1/$cost, which is the canonical regime.
    Left: activation function of $\textrm{ReLU}(z) =\max\{0,z\}$. Right: activation function of $\sigma(z) =\tanh(z)$.}
    %Absolute error of the difference between the finer and coarse level, with ReLU activation function.}
    \label{fig:sum}
\end{figure}

\begin{figure}[h!]
\centering
\includegraphics[width=0.40\textwidth]{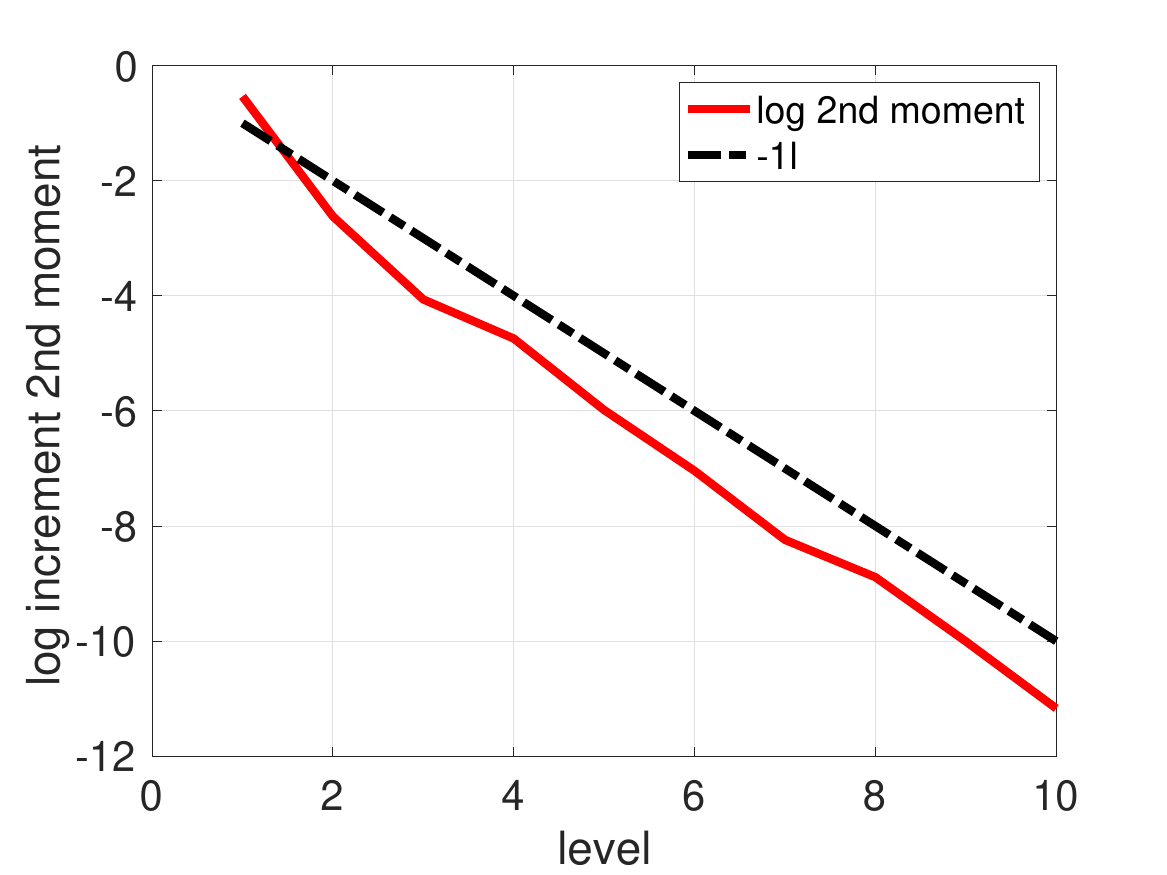}
\includegraphics[width=0.40\textwidth]{williams_3-eps-converted-to.pdf}
    \caption{Increment 2nd moment vs. levels.
    The decay is $\cO(2^{-l})$.
      %The number of weights, and hence the cost, is $\cO(2^{2l})$.
    Therefore the variance decays slower than $1/$cost, which is a sub-canonical regime.    
    Left: activation function of $\textrm{ReLU}(z) =\max\{0,z\}$. Right: activation function of $\sigma(z) =\tanh(z)$.}
    %Absolute error of the difference between the finer and coarse level, with ReLU activation function.
    \label{fig:sum_2l}
\end{figure}

\section{Numerical Experiments}
\label{sec:numerics}

\textcolor{black}{For this section we provide our main numerical experiments,
which is to use our TNN priors on various Bayesian tasks
related to machine learning.
Specifically we will aim to use these priors with the proposed methodology of Algorithm \ref{alg:mlsmc_dnn},
comparing it to single-level Monte Carlo SMC sampler which use the same priors. We will test this on a range of machine learning tasks, which include regression,
classification (spiral and MNIST data set) and reinforcement learning, where we hope to show the benefit of using MLMC, by attaining the canonical rate of convergence.}
\textcolor{black}{We provide a sensitivity analysis of our methodology w.r.t. modifications of parameter choices (such as within the TNN and also the MLMC framework)
for the MNIST problem, illustrating the robustness of the approach. 
Finally, we present a much more challenging practical example of sentiment classification for 
the IMDb dataset of 50,000 movie reviews, using a simplified version of the method presented.
The results are compared to the state-of-the-art in Bayesian deep learning for uncertainty quantification (UQ). 
This illustrates the practical value of the method,
%This example illustrates 
as well as its flexibility and generality.}
The Code can be downloaded from \href{https://github.com/NKC123/MLTNN}{\texttt{{https://github.com/NKC123/MLTNN}}}.
%We have also tested our algorithm on an MNIST binary classification problem, which can be found in the supplementary file}. 

\subsection{Regression Problem}

Our first  numerical experiment will be based on a Bayesian regression problem
given in \eqref{eq:modelreg}.
%Let $\mathcal{D}=\{x_i,y_i\}_{i=1}^N$ represent our dataset
% where $X = \{x_i\}^N_{i=1}$ and $Y = \{y_i\}_{i=1}^N$
% denotes our observations which take the general form
% $$
%y_i = f(x_i,\theta)+ \epsilon_i \, , \qquad
%\epsilon_i \stackrel{\textrm{ind}}{\sim} \mathcal{N}_m(0,\Sigma_i),
%$$
Our objective is to identify the posterior distribution on $\theta$
specified in \eqref{eq:bnn_posterior}
with the likelihood specified in \eqref{eq:reglikelihood}
%, with likelihood
%, given as
$$
\pi(\theta | y_{1:N}) \propto p(y_{1:N}|\theta) \overline{\pi}(\theta),  \qquad p(y_{1:N}|\theta) = \prod_{i=1}^N \phi_m(y_i;f(x_i,\theta),\Sigma_i) \, ,
$$
where %where we recall
%that
%$p(y_{1:N}|\theta,x_{1:N})$ denotes the likelihood function, and
$\overline{\pi}(\theta)$ is the TNN prior \eqref{eq:tnn}, discussed in Section \ref{sec:tnn}.
This can subsequently be used for prediction at a new test data point $x^*$
$$
\bbE\left[f(x^*,\theta)|y_{1:N}\right] = \int_\Theta f(x^*,\theta) \pi(\theta | y_{1:N}) d\theta \, ,
$$
as well as for the computation of other posterior predictive quantities such as
the variance
$$
\bbV\textrm{ar}\left[f(x^*,\theta)|y_{1:N}\right] = \int_\Theta \left(f(x^*,\theta)-\bbE\left[f(x^*)|y_{1:N}\right]\right)^2 \pi(\theta |y_{1:N}) d\theta.
$$
For our observational noise we take $\Sigma_i^2 =0.01^2 I$.
We use a TNN prior with a tanh activation function, i.e. $\sigma(z) = \tanh(z)$. As we have no explicit form for our ``ground-truth",
we run a high-level simulation to attain a posterior, which we take as our reference solution. For this we take $n_l=2^l$
where $l=7$.
For this model we specify $N=200$, and consider  $\sX=\bbR^{10}$ so the input NN width is $n_0=10$, and the data points are generated as normal distributed random variables, 
i.e. $x_i \sim \mathcal{N}(2,0.5)$.
Then for each $x_i$, the following model  produces $y_i$.
To compute the MSE we consider 100 replications,
and then take the MSE.
We will compare and apply an SMC sampler and multilevel counterpart, i.e. a MLSMC sampler.
The setup of the MLSMC sampler is presented in Algorithm \ref{alg:mlsmc_dnn}. In this first experiment
we will apply both these methodologies and observe the error-to-cost
rate $\xi$ such that cost $\propto$ MSE$^{-\xi}$. \textcolor{black}{Recall the cost can be computed using equation \eqref{eq:mlmcCost}}.
We will compare this for different
values of $\beta=2\alpha-1$.
Also recall that the cost (per one sample) is of the order $\cO(2^{\gamma l})$, with $\gamma=2$,
so as a rule of thumb one should expect the attain a canonical rate
of convergence when $\alpha>1.5$ and $\beta>\gamma$.

Below in Figure \ref{fig:reg_results1} we present  results of applying both SMC sampler and a MLSMC sampler.
As mentioned, further details on the MLSMC sampler can be found in \cite{beskos2018multilevel}. We conduct our numerical experiment with levels
$ L \in \{3,4,\ldots,7\}$. The first numerical results are presented in Figure \ref{fig:reg_results1} which compares
 values of $\alpha=\{1.7, 1.9, 2, 3\}$. We plot both the SMC sampler and MLSMC sampler where the prior for each method is our
 TNN prior.
 Furthermore we also plot the credible sets around the MSE values, given by the thin blue and red curves.
 
 As we can observe from the results, The MLMC methodology shows computational benefits, where for a fixed MSE the cost associated with attaining
 that order of MSE is smaller than the SMC sampler. We can see that the difference in cost for the lowest MSE
 is approximately a factor of 10. %of order $\sim \mathcal{O}(10^{1})$.
 This indicates that the error-to-cost rates are different, and to verify our theoretical findings
 we plot the canonical rate in black, which matches that of our proposed methodology. 
\begin{figure}[h!]
\centering
\includegraphics[width=0.45\textwidth]{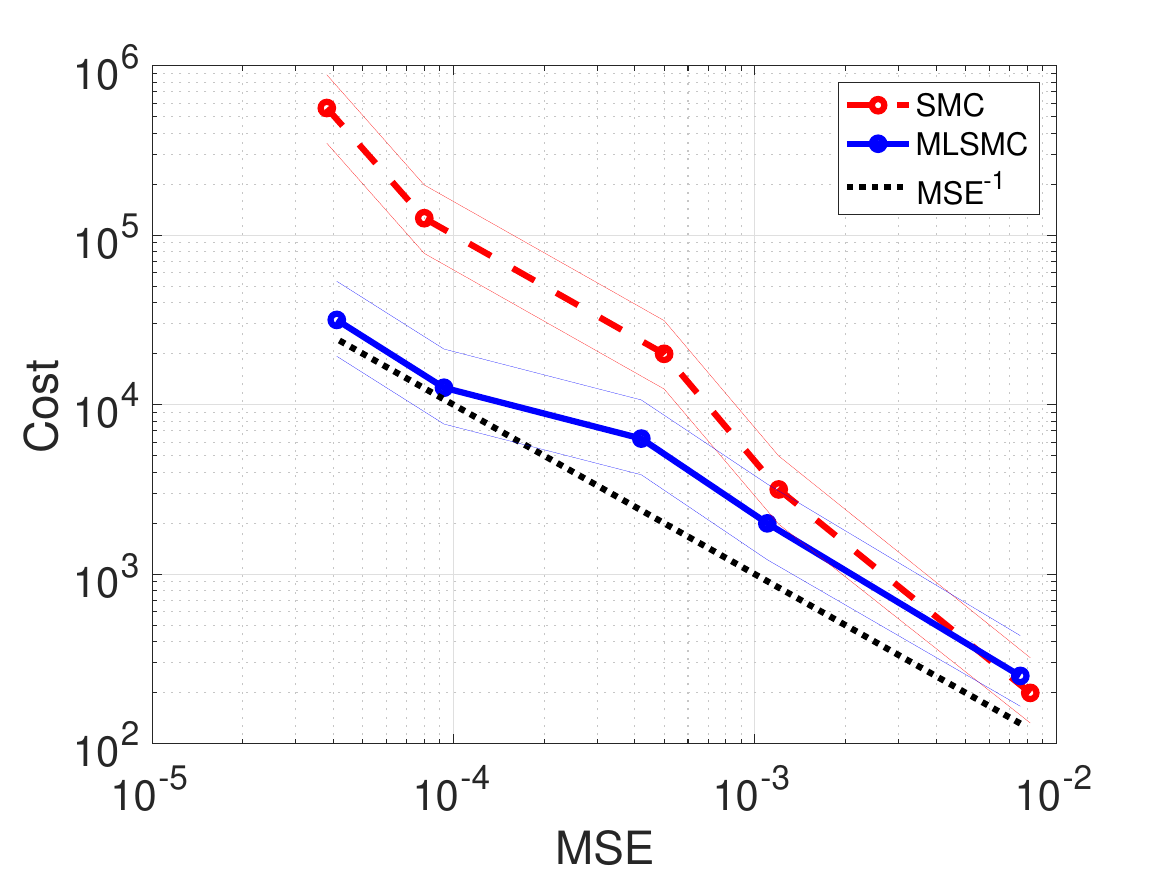}
\includegraphics[width=0.45\textwidth]{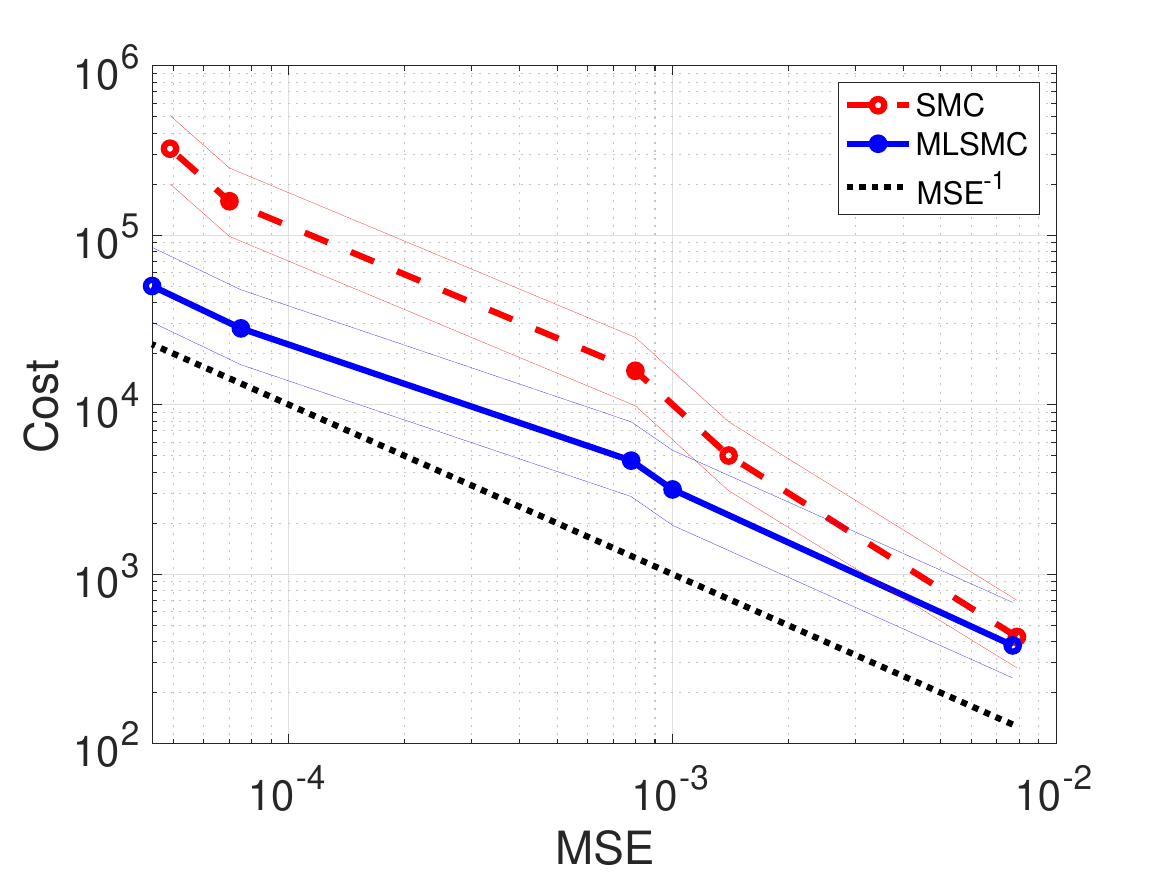}
\includegraphics[width=0.45\textwidth]{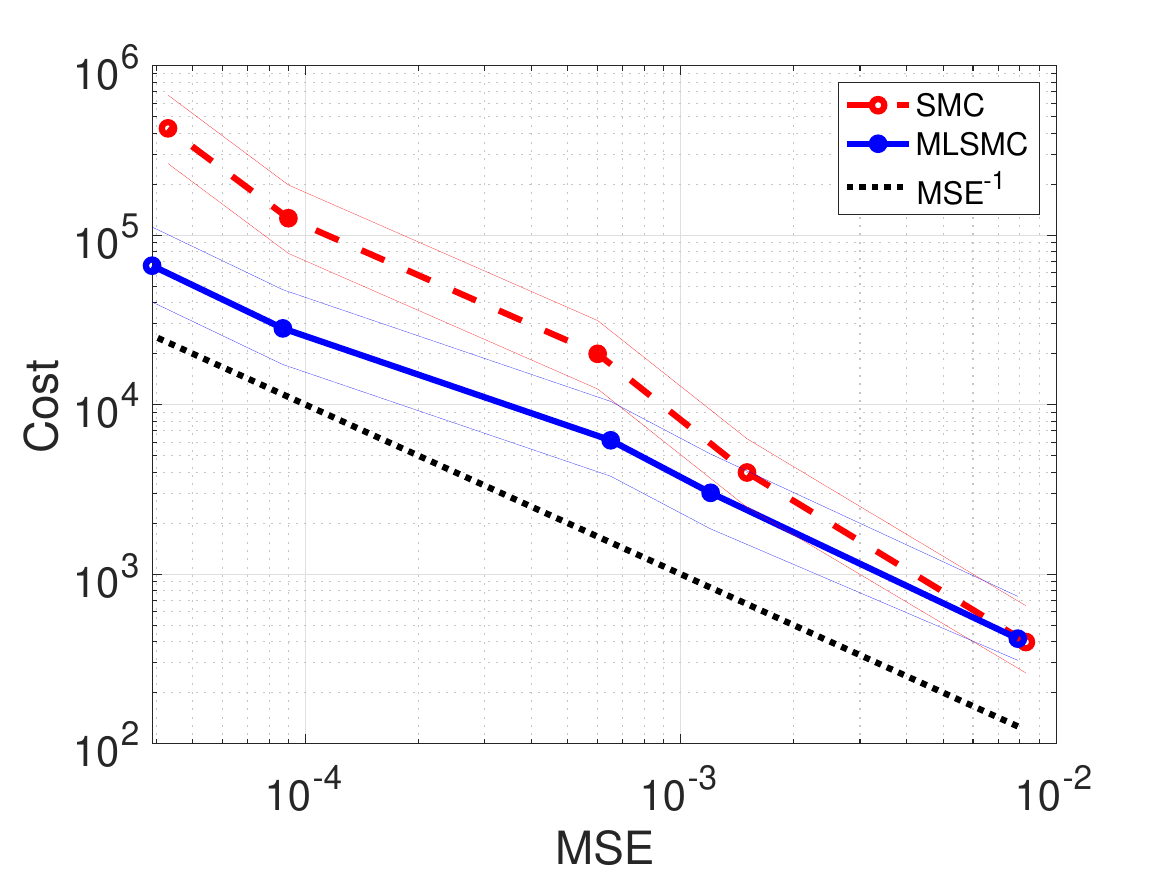}
\includegraphics[width=0.45\textwidth]{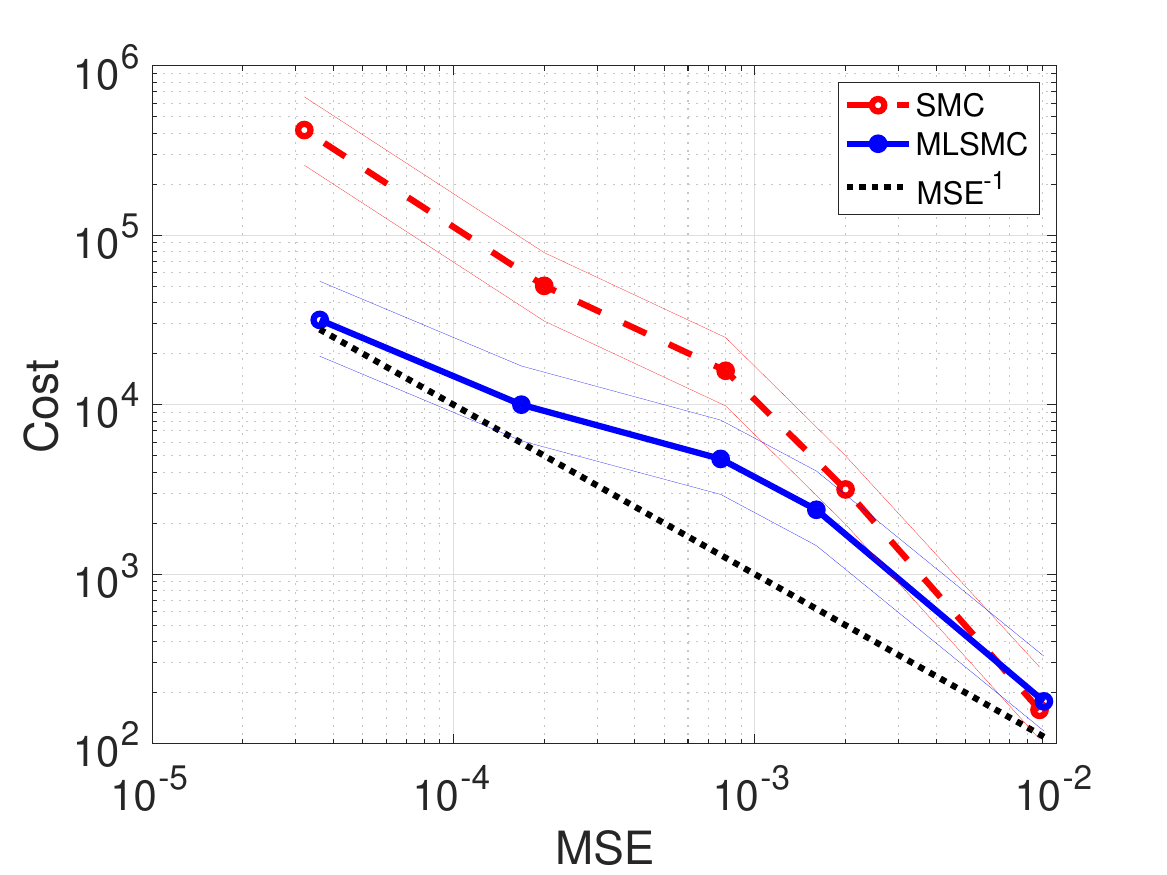}
 \caption{Regression problem: error vs cost plots for SMC and MLSMC using  TNN priors.
 Top left: $\alpha=3$. Top right: $\alpha = 2.0$. Bottom left: $\alpha=1.9$. Bottom right: $\alpha=1.7$.
 Credible sets are provided in the thin blue and red curves.}
    \label{fig:reg_results1}
\end{figure}

\subsubsection{Sub-canonical rates}
The results
 thus far are for values of $\alpha$ which attain the canonical rate. Now let us consider alternative choices of $\alpha$, which should result in
 a sub-canonical rate.
We now modify the parameter values of
$\alpha \leq 1.5$.
In this case, $\beta \leq \gamma = 2$,
where the cost of a single simulation of $f_l$ is $\cO(2^\gamma l)$,
so sub-canonical convergence is expected.
%we expect %this to %conjecture that this will
%relate to the sub-canonical case,
%i.e.
%$\gamma$ is
%the work done
%in the multilevel Monte Carlo framework from Theorem \ref{thm:VMLMC}.
We will consider two choices of
%$\alpha$ which are
$\alpha \in \{1.1,1.4\}$, where we keep the experiment and the parameter choices the same.
Our results are presented in Figure \ref{fig:reg_results2}. %As we observe t
As expected, the complexity rate is closer to the single level case,
but an improvement factor of around 3 is still observed at the resolutions considered.
%The difference in the cost for the smallest
%MSE values is lower compared to that of the canonical rate, which roughly looks of order $\sim \mathcal{O}(10^{0.5})$.

\begin{figure}[h!]
\centering
\includegraphics[width=0.45\textwidth]{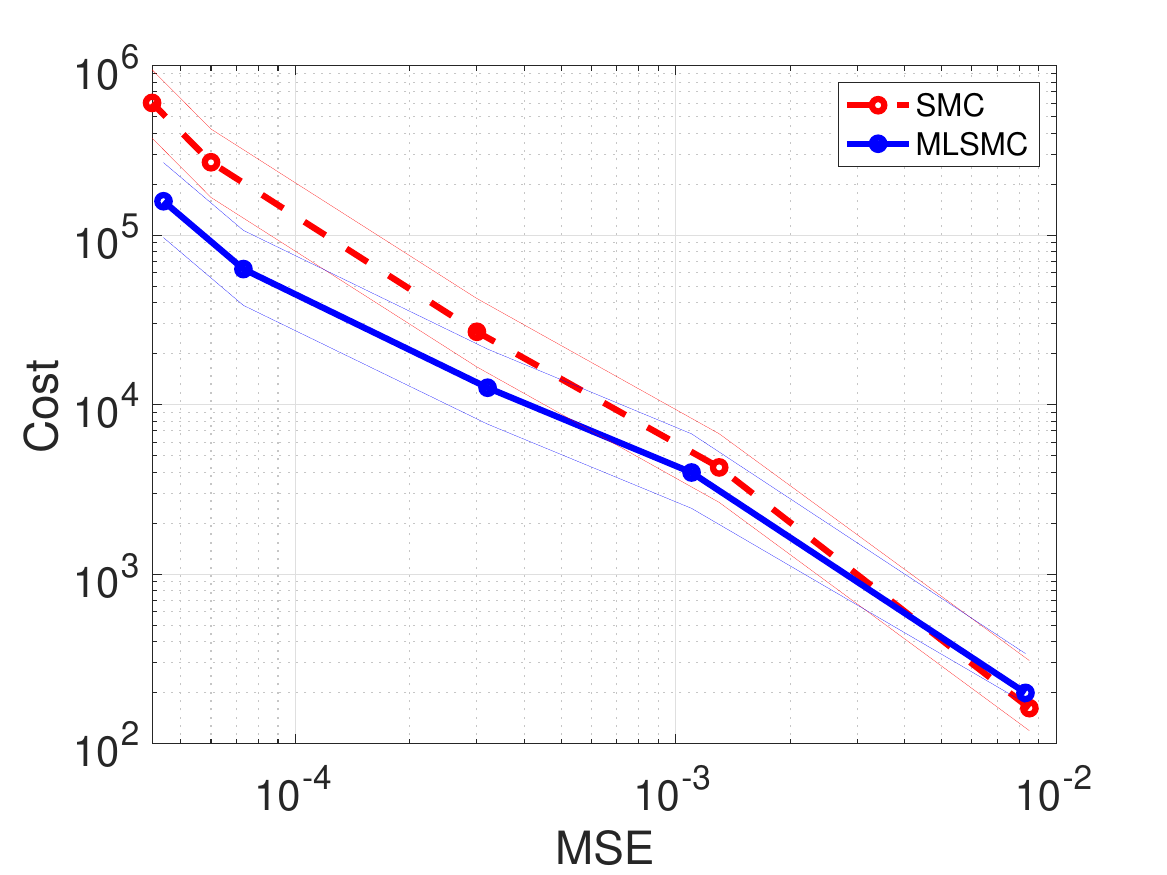}
\includegraphics[width=0.45\textwidth]{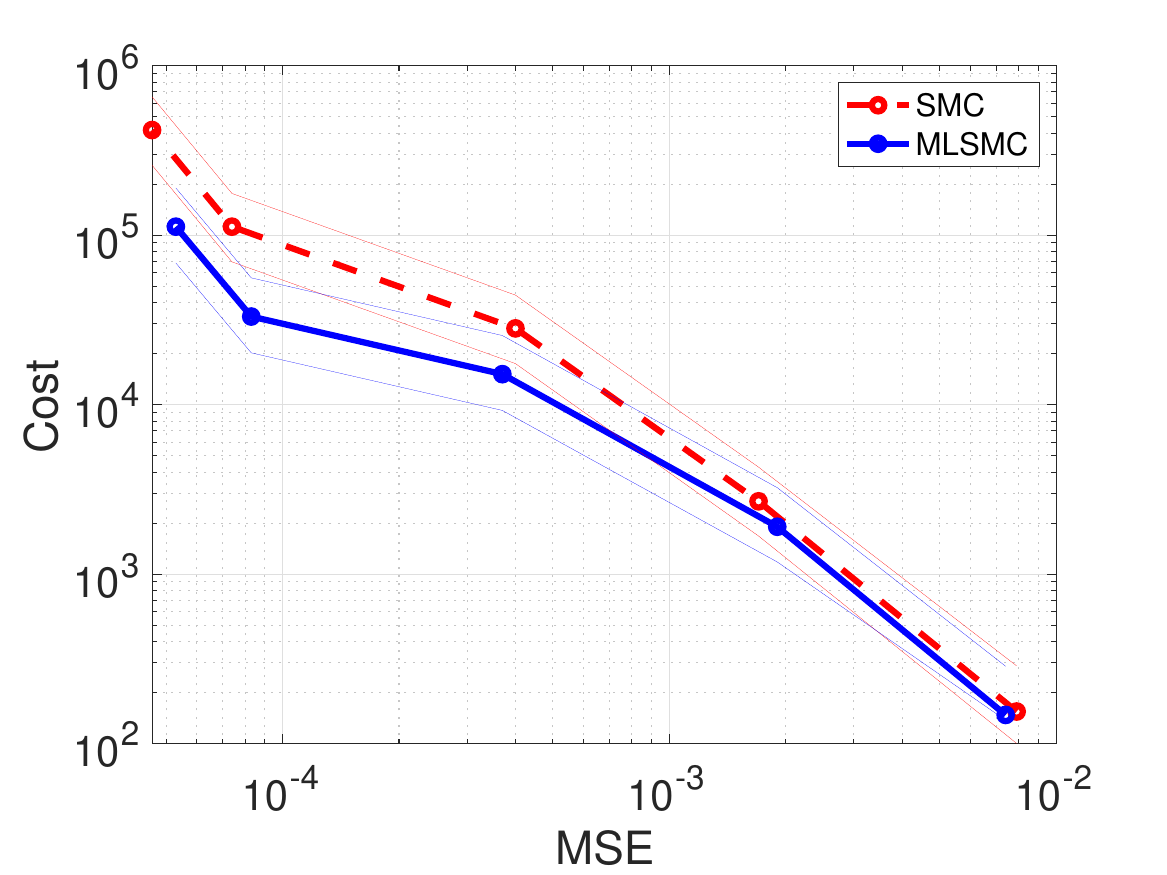}
 \caption{Regression problem: error vs cost plots for SMC and MLSMC using TNN priors.
 Left: $\alpha=1.4$.  Right: $\alpha = 1.1$.  Credible sets are provided by thin blue and red curves.}
    \label{fig:reg_results2}
\end{figure}

The results from Figure \ref{fig:reg_results2} indicate that if we consider TNN priors of lower regularity, then we can expect
to achieve sub-canonical rates, in the complexity related to the MSE-to-cost ratio. This promotes the question, of whether
canonical rates, in our setup and framework are possible for non-smooth random fields. This, and related questions, will be considered for future work.

\subsection{%Bayesian
Spiral Classification Problem}

Following the definition of the classification model
\eqref{eq:class}, just as we considered the likelihood \eqref{eq:reglikelihood}
for the regression problem, we consider the likelihood \eqref{eq:classlikelihood}
for the classification problem. The posterior is again given by \eqref{eq:bnn_posterior}.
Again predictions at an input $x^*$ not part of the data are delivered by the posterior predictive distribution.
For example, the marginal posterior class probability
$\bbP(k|x^*,y_{1:N})$ is given by
$$
\bbP(k|y_{1:N}) = \int_\Theta h_k(x^*,\theta) \pi(\theta | y_{1:N}) d\theta \, ,
$$
%Moreover, the posterior provides uncertainty quantification, where
and the variance associated to a prediction of class $k$
for input $x^*$ is given by
$$
\bbV\textrm{ar}\left[k|y_{1:N}\right] = \int_\Theta h_k(x^*,\theta)^2 \pi(\theta | y_{1:N}) d\theta - \bbP(k|y_{1:N})^2 \, .
$$
For our classification problem we are interested in classifying two data sets, where the data
is based on a 2D spiral.
In other words, the data for class $k=1$
is generated through the following equations,
for $i=1,\dots,N=500$
\begin{align*}
x_{i1} &= a \upsilon_i^p \cos (2t_i^p\pi) + \epsilon_{i1}, \\
x_{i2} &= a \upsilon_i^p \sin (2t_i^p\pi) +  \epsilon_{i2},
\end{align*}
 where
 %generated by the following uniform random variables
 $\upsilon_i, t_i \sim \mathcal{U}[0,1]$ uniform, $\epsilon_i \sim \mathcal{N}(0,0.1^2)$
 and the parameter choices are
 $a=16$ and $p=0.05$.
 The data associated to class $k=2$ is generated similarly, except with a shift of $\pi$
 in the arguments of the trigonometric functions.
 %, and we specify 500 data points.
 This data can be seen in  Figure \ref{fig:class},  
 where our two classes correspond to the colors, i.e. data labelled as
 Class $k=1$ is blue and data labelled as Class $k=2$ is yellow.
The setup for the classification problem is similar to the regression problem.
Again we conduct experiments for different choices of $\alpha$,
i.e. $\alpha=\{1.7, 1.9, 2, 3\}$ and levels $L \in \{3,4,\ldots,7\}$.
Our parameter choices are the same as the regression problem. For our reference solution we again use $n_7=2^7$ with a $tanh$ activation function, for the prior and take a high-resolution solution.
We present firstly in Figure \ref{fig:class_results1} the results related
to the canonical rates obtained for the MSE-to-cost rates. As we observe, we see a clear distinction in the difference in costs as
the MSE is decreased, with again roughly a magnitude of $\sim \mathcal{O}(10^1)$. This again
highlights the computationally efficiency of our proposed methodology, with the combination of our TNN prior.
For our final experiments for the classification problem, we consider alternative values of $\alpha$, i.e., $\alpha \in \{1.1,1.4\}$,
which are presented in Figure \ref{fig:class_results2}. Again, similar to the regression problem, we observe that sub-canonical
rates are obtained for $\alpha<1.5$, where the difference in cost for the lower value of MSEs is not as significant.

\begin{figure}[h!]
\centering
\includegraphics[width=0.45\textwidth]{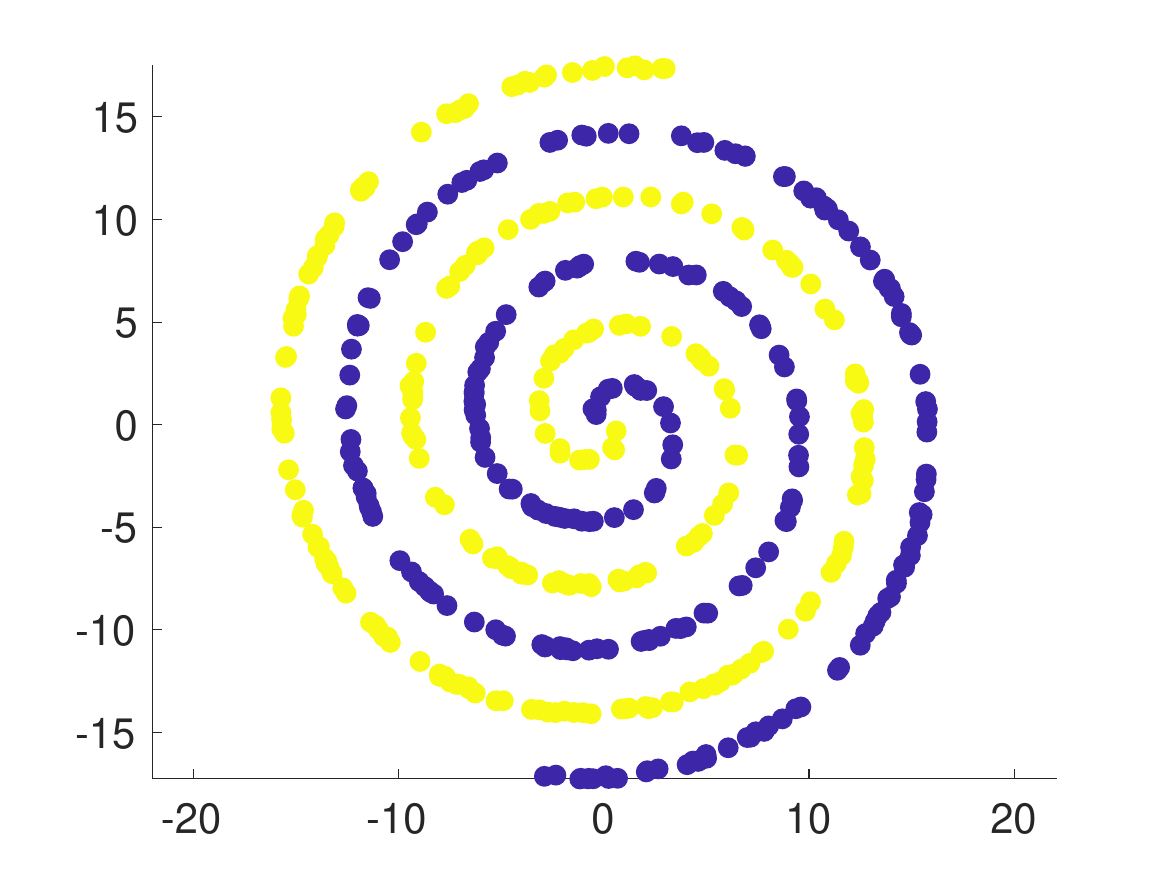}
 \caption{Classification problem: Our data is generated as a 2D spiral with two classes, Class $k=1$ being in blue and Class $k=2$ in yellow.}
 \label{fig:class}
\end{figure}

\begin{figure}[h!]
\centering
\includegraphics[width=0.45\textwidth]{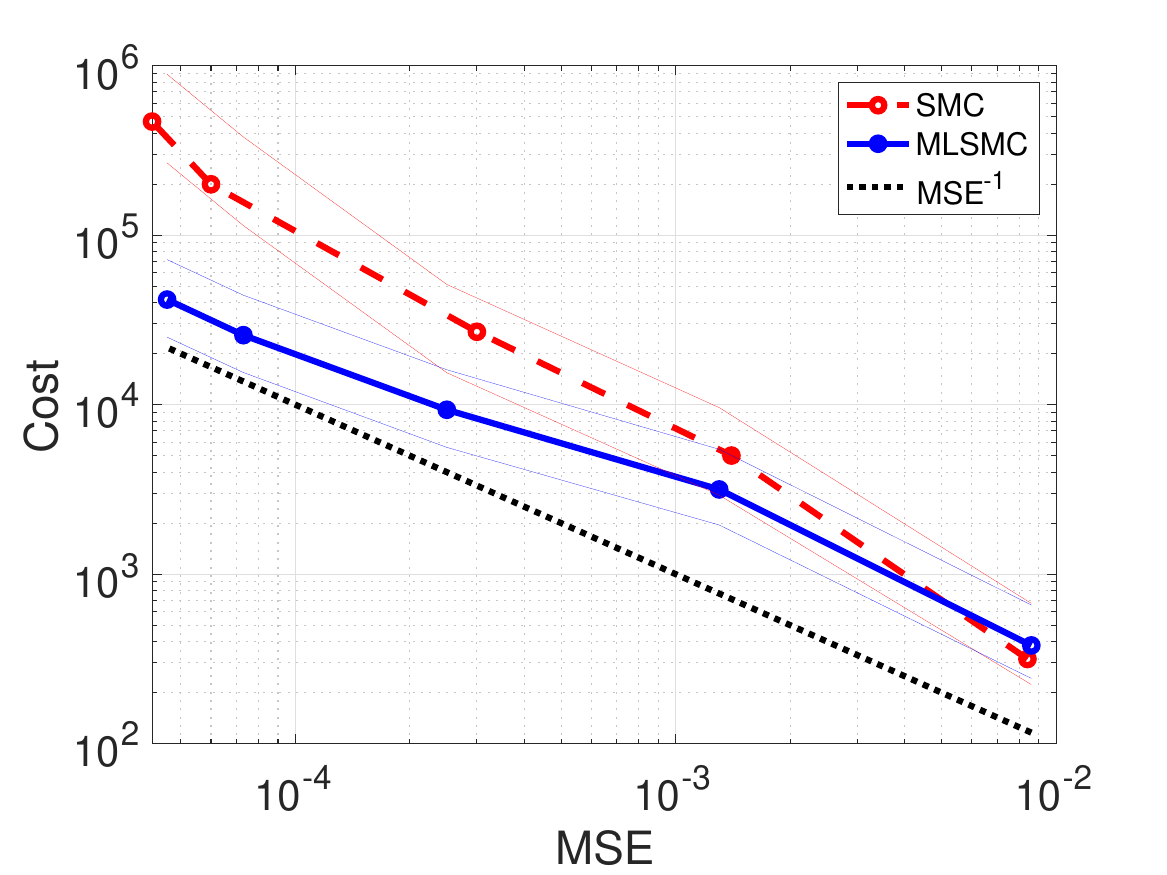}
\includegraphics[width=0.45\textwidth]{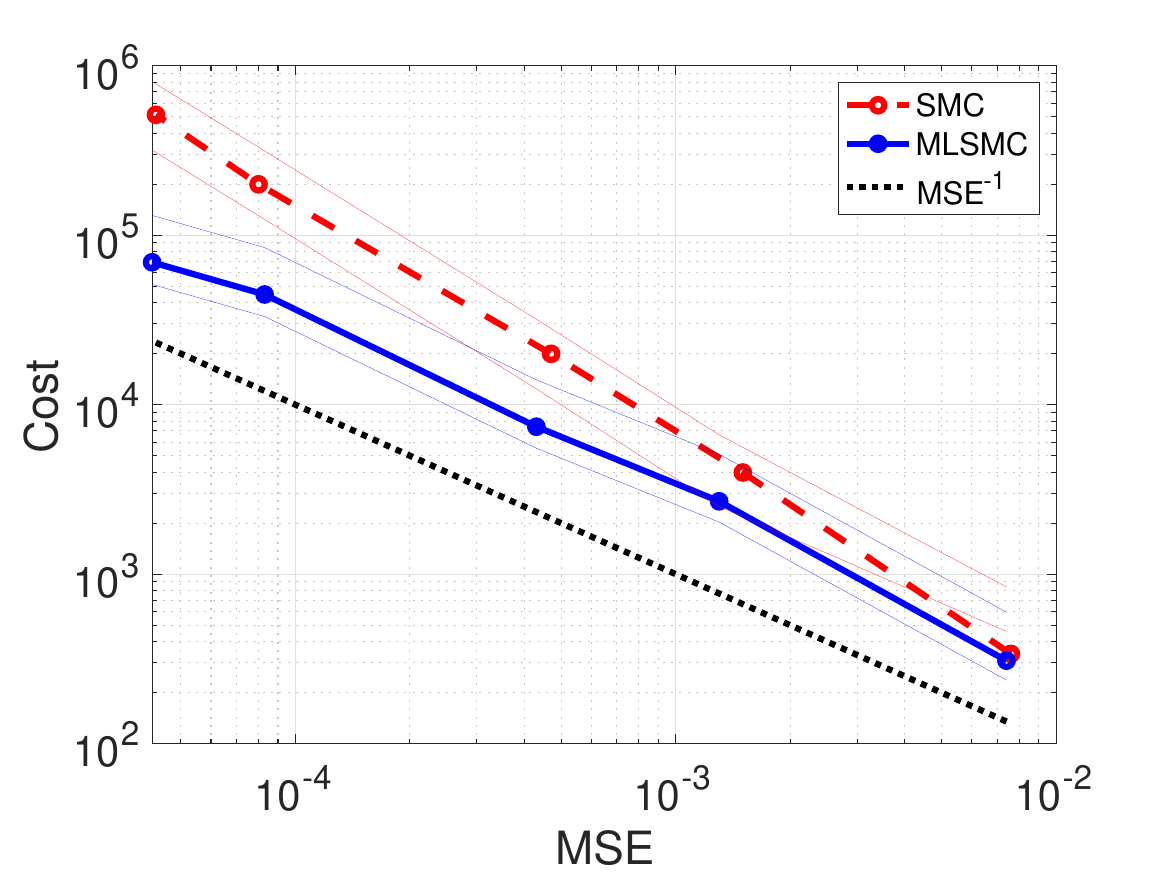}
%make this plot 3 below
\includegraphics[width=0.45\textwidth]{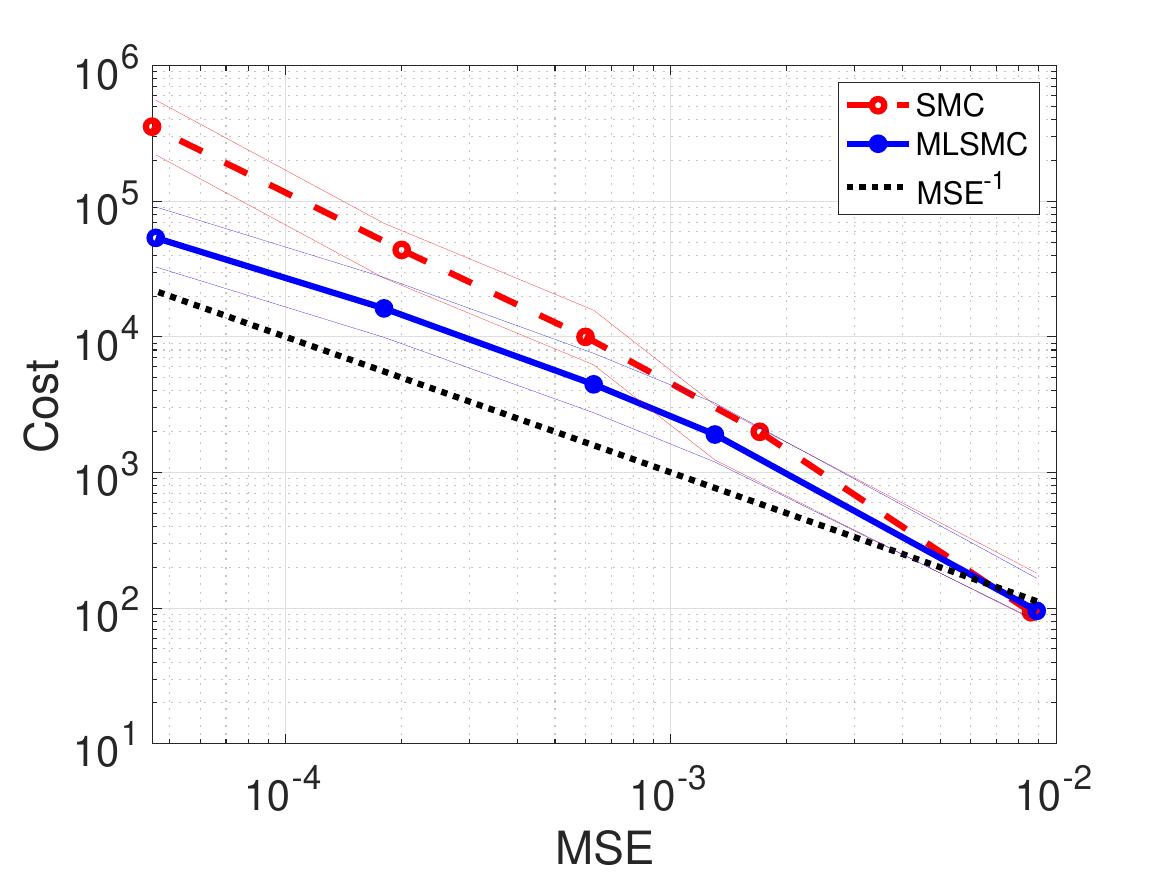}
\includegraphics[width=0.45\textwidth]{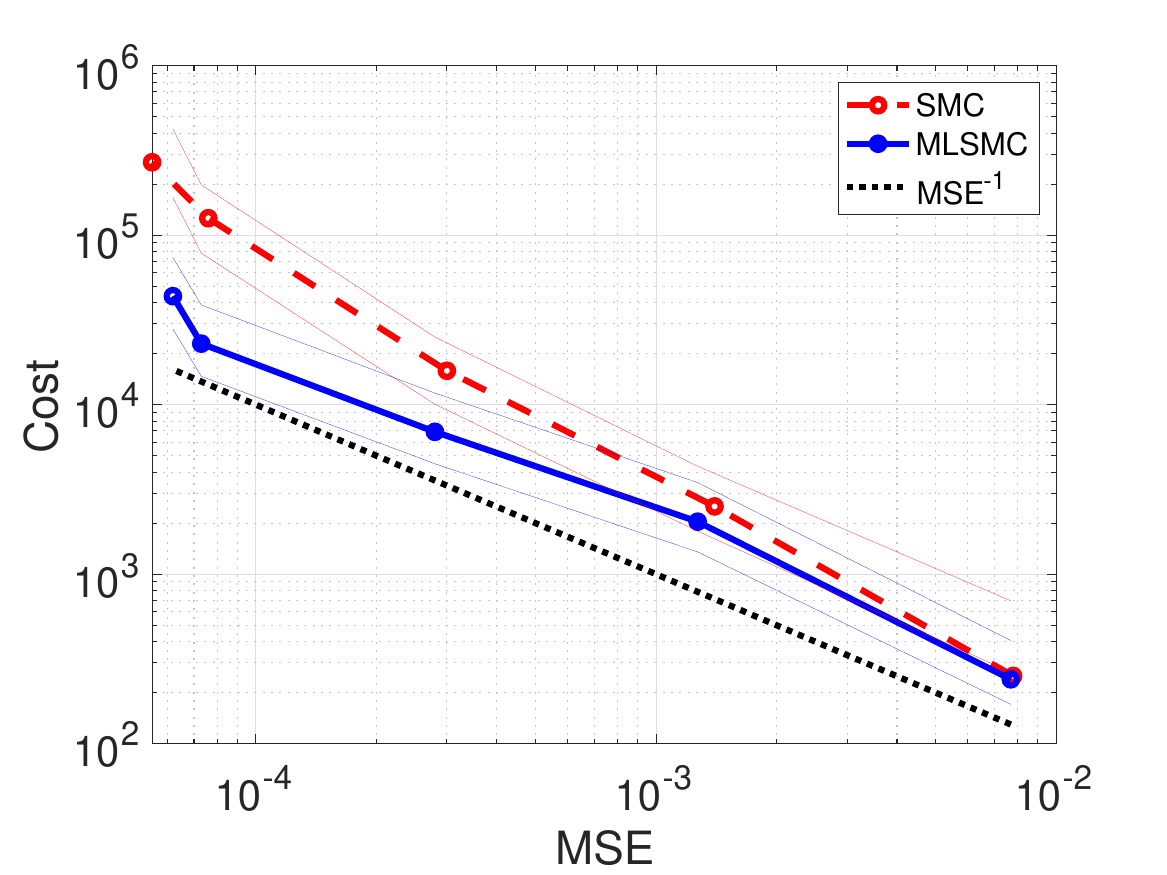}
 \caption{Classification problem: error vs cost plots for SMC and MLSMC, using TNN priors.
 Top left: $\alpha=3$. Top right: $\alpha = 2.0$. Bottom left: $\alpha=1.9$. Bottom right: $\alpha=1.7$.
  Credible sets are provided in the thin blue and red curves.}
    \label{fig:class_results1}
\end{figure}

\begin{figure}[h!]
\centering
\includegraphics[width=0.45\textwidth]{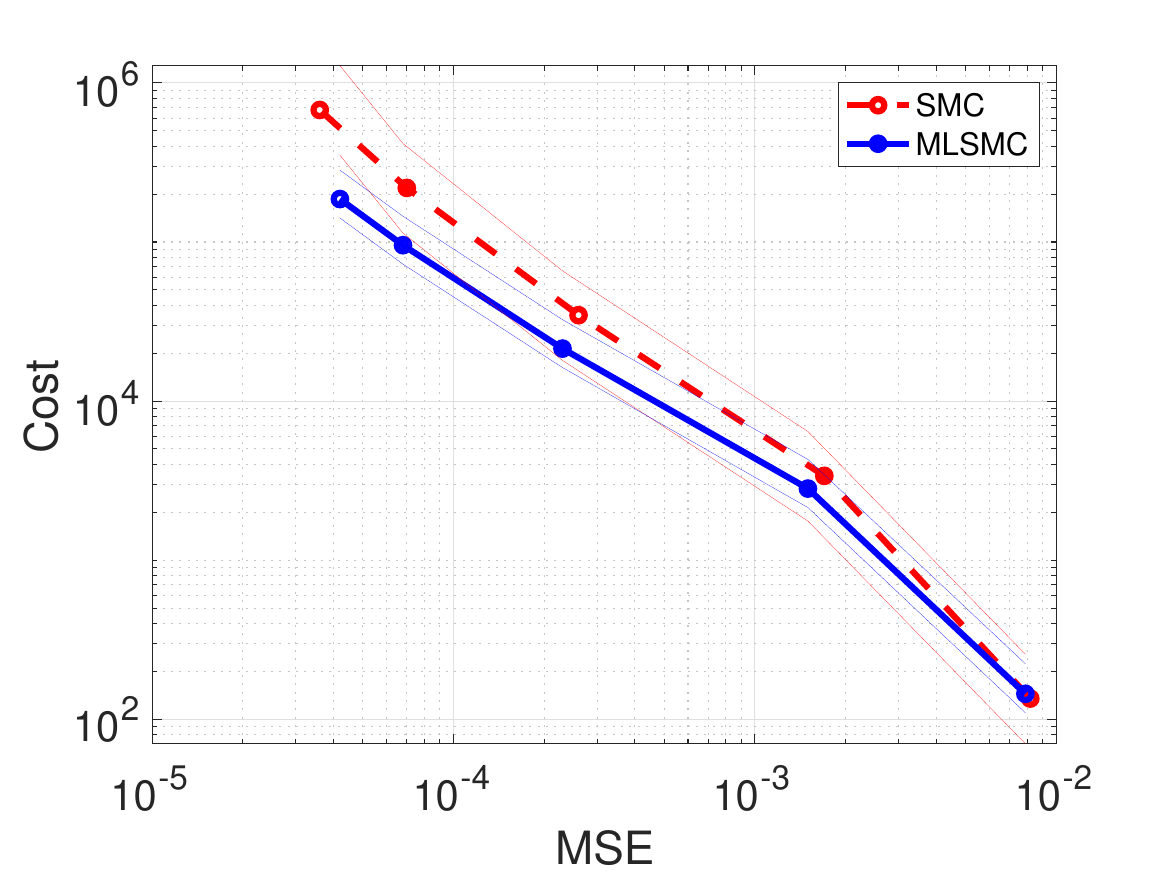}
\includegraphics[width=0.45\textwidth]{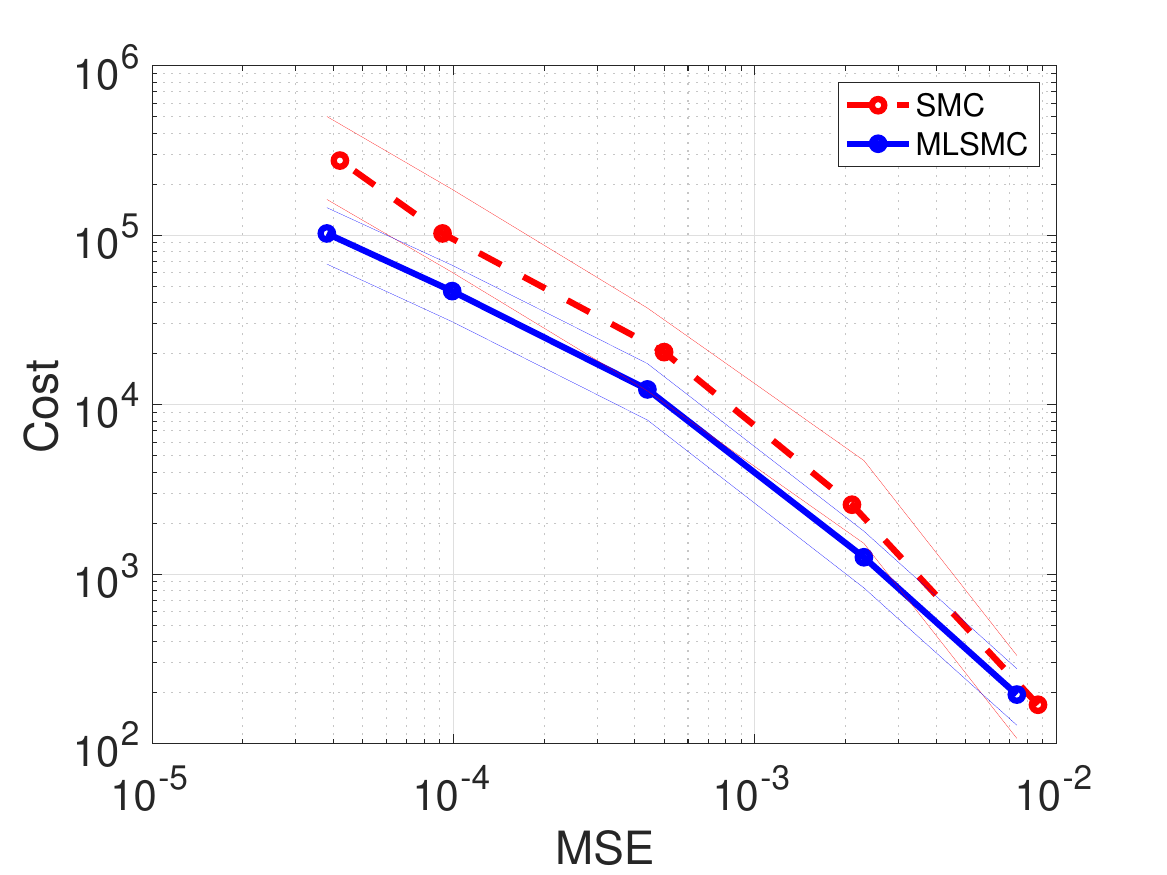}
 \caption{Classification problem: error vs cost plots for SMC and MLSMC, using TNN priors.
Left: $\alpha=1.4$. Right: $\alpha=1.1$.  Credible sets are provided in the thin blue and red curves.}
    \label{fig:class_results2}
\end{figure}

\subsection{%Bayesian
Binary MNIST Classification}
\label{sec:mnist}
\textcolor{black}{We now turn our attention to a more realistic classification example, which is based on the well-known MNIST dataset. This dataset is relatively high-dimensional dataset which contains handwritten digits from 0 to 9. Our problem that we consider now is a binary MNIST classification problem, where we are interested in two classes. One class corresponds to a collection of zeros and the other is related to the class of ones. This can be seen in Figure 
\ref{fig:mnist}. Our methodology and experimental setup remains the same, where we refer the reader to the previous subsection on specific choices on our MLSMC algorithm. The main difference we have now is related to dimensionality and number of samples, where this is choosen now as $d=784$ and $m=400$. To overcome this challenging dimensionality issues, we apply principle component analysis (PCA) which reduces the dimension to $d=100$.}
\\
\textcolor{black}{Our numerical experiments are provided in Figure \ref{fig:mnist1} - \ref{fig:mnist2}. From analyzing Figure \ref{fig:mnist1} we notice again that there is a clear difference in the MSE-to-cost plots for both methods. Specifically, by exploiting a TNN prior within our MLSMC methodology we attain the canonical rate of convergence. This is consistent amongst all subplots for values of $\alpha=\{1.7, 1.9, 2, 3\}$. Again, if we do not consider the TNN prior, what we observe is an improvement compared to the SMC sampler methodology, however we do not attain the canonical rate. This is observed in Figure \ref{fig:mnist2}, where we have $\alpha=\{1.1, 1.4\}$.}
\begin{figure}[h!]
\centering
\includegraphics[width=0.5\textwidth]{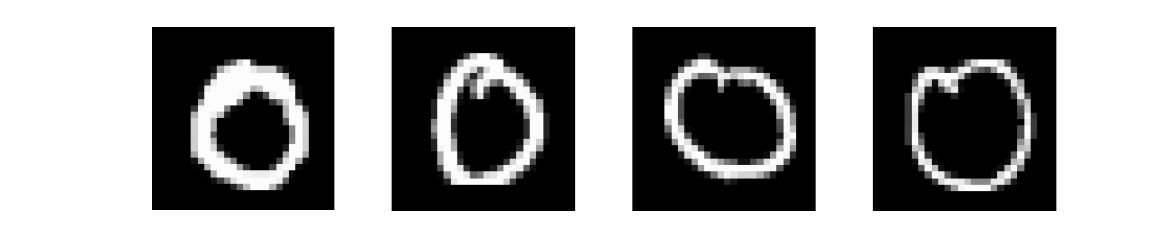} \\
 \  \includegraphics[width=0.5\textwidth]{mnist0-eps-converted-to.pdf}
\caption{Random selection of 0's and 1's from the MNIST dataset.}
\label{fig:mnist}
\end{figure}

\begin{figure}[h!]
\centering
\includegraphics[width=0.45\textwidth]{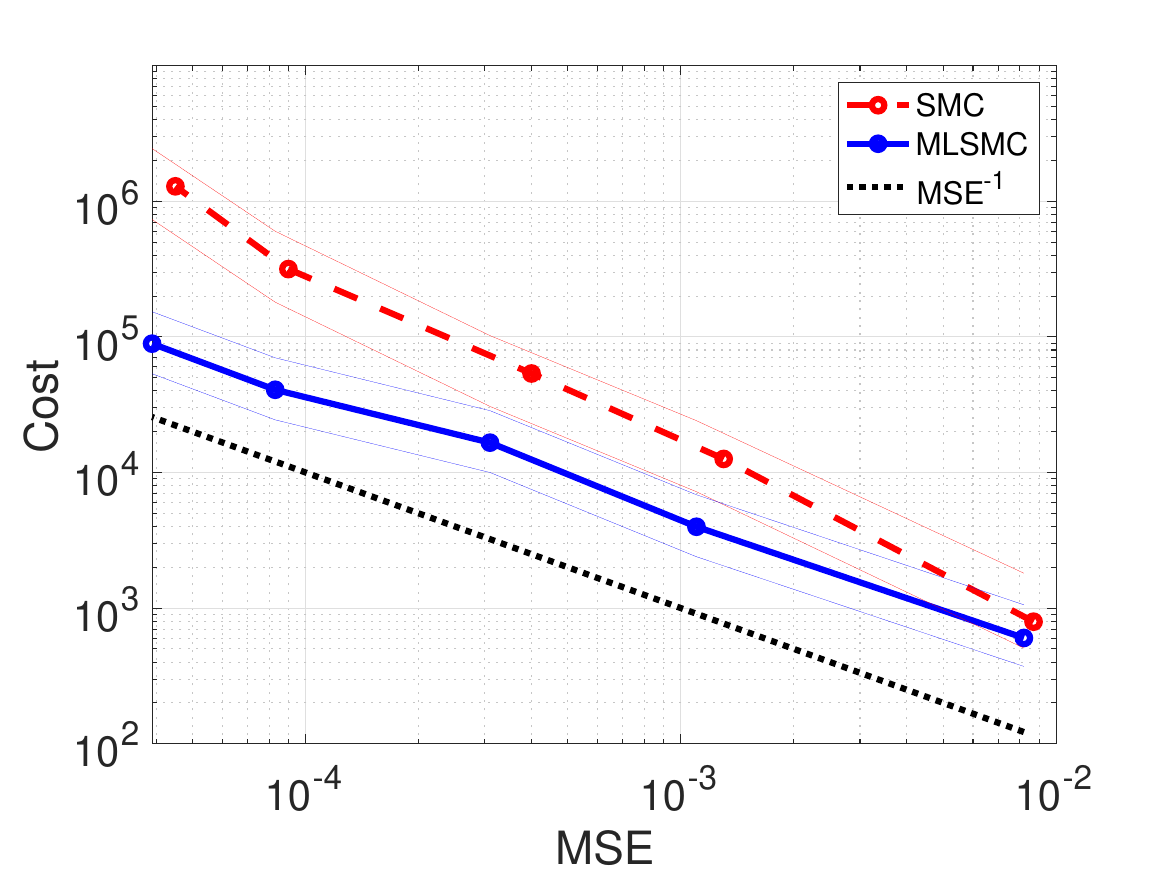}
\includegraphics[width=0.45\textwidth]{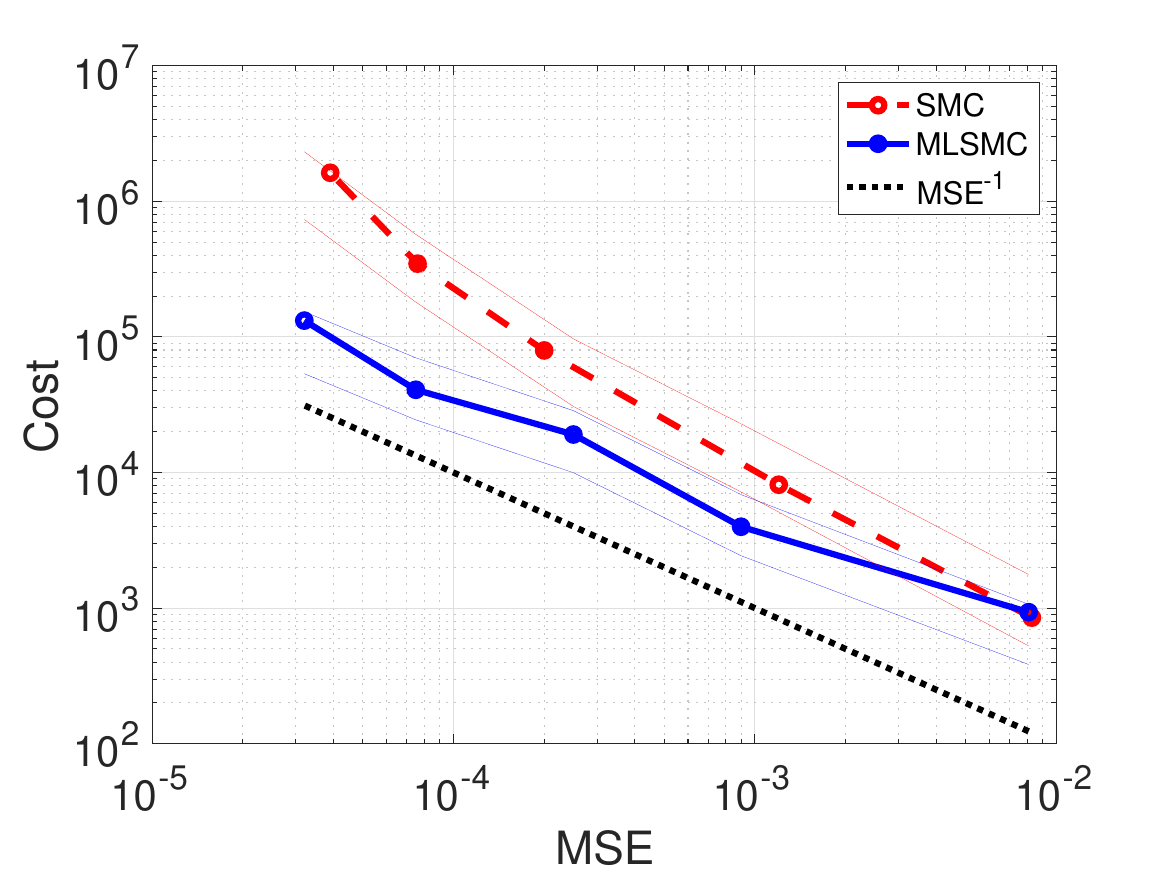}
%make this plot 3 below
\includegraphics[width=0.45\textwidth]{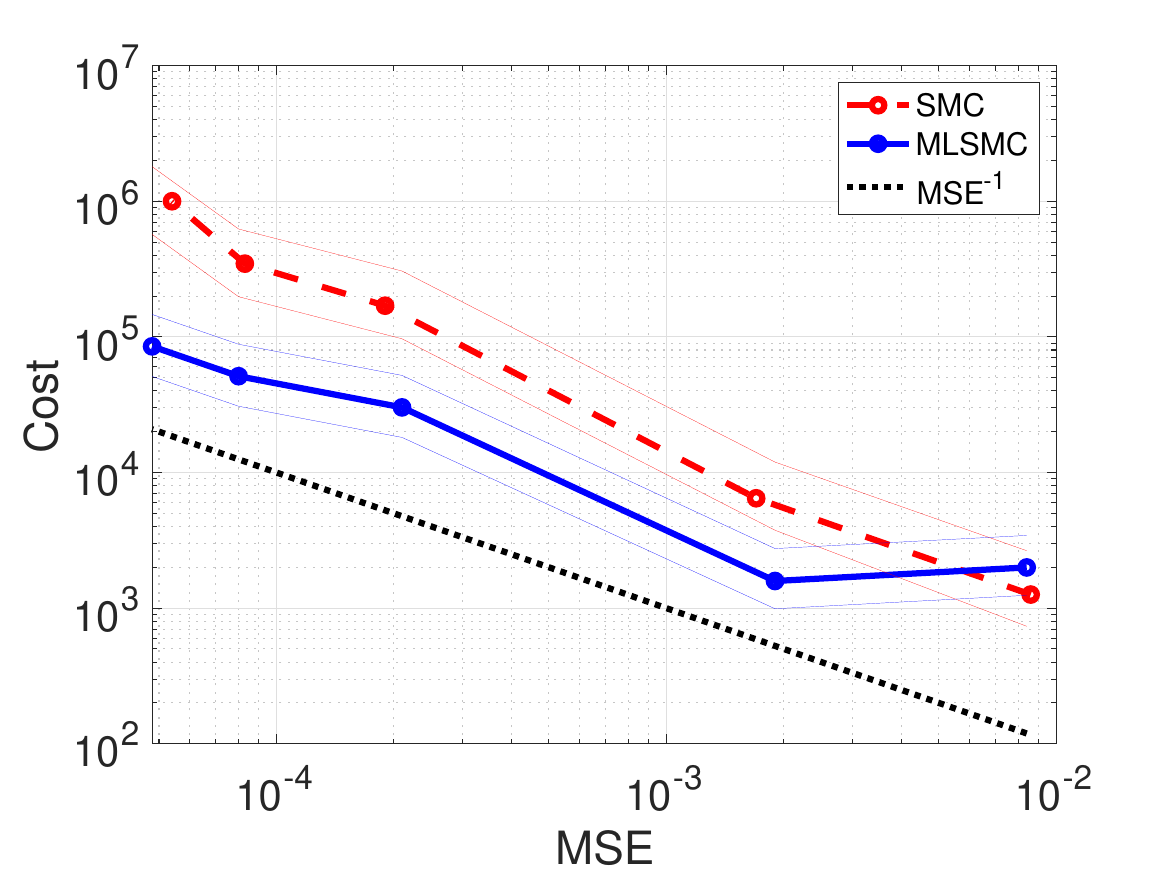}
\includegraphics[width=0.45\textwidth]{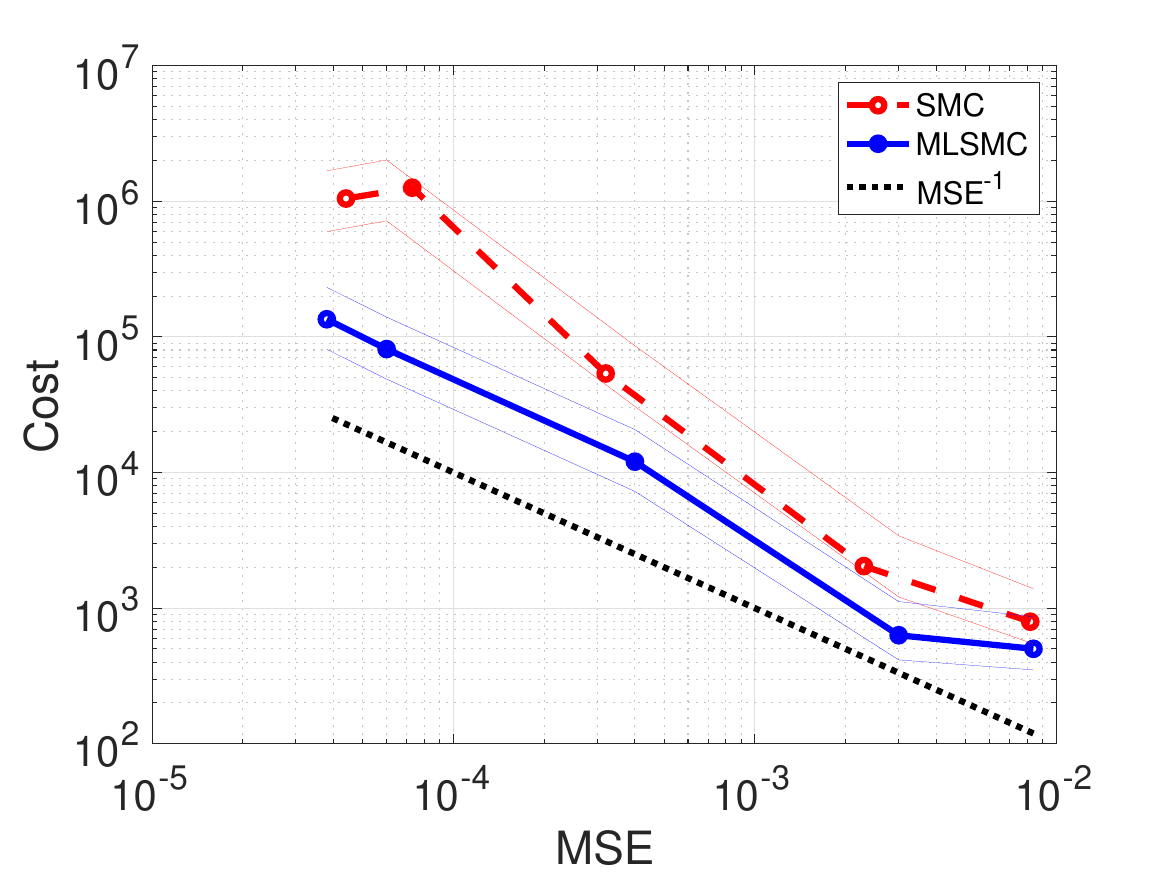}
 \caption{MNIST Classification problem: error vs cost plots for SMC and MLSMC, using TNN priors.
 Top left: $\alpha=3$. Top right: $\alpha = 2.0$. Bottom left: $\alpha=1.9$. Bottom right: $\alpha=1.7$.
  Credible sets are provided in the thin blue and red curves.}
  \label{fig:mnist1}
\end{figure}

\begin{figure}[h!]
\centering
\includegraphics[width=0.45\textwidth]{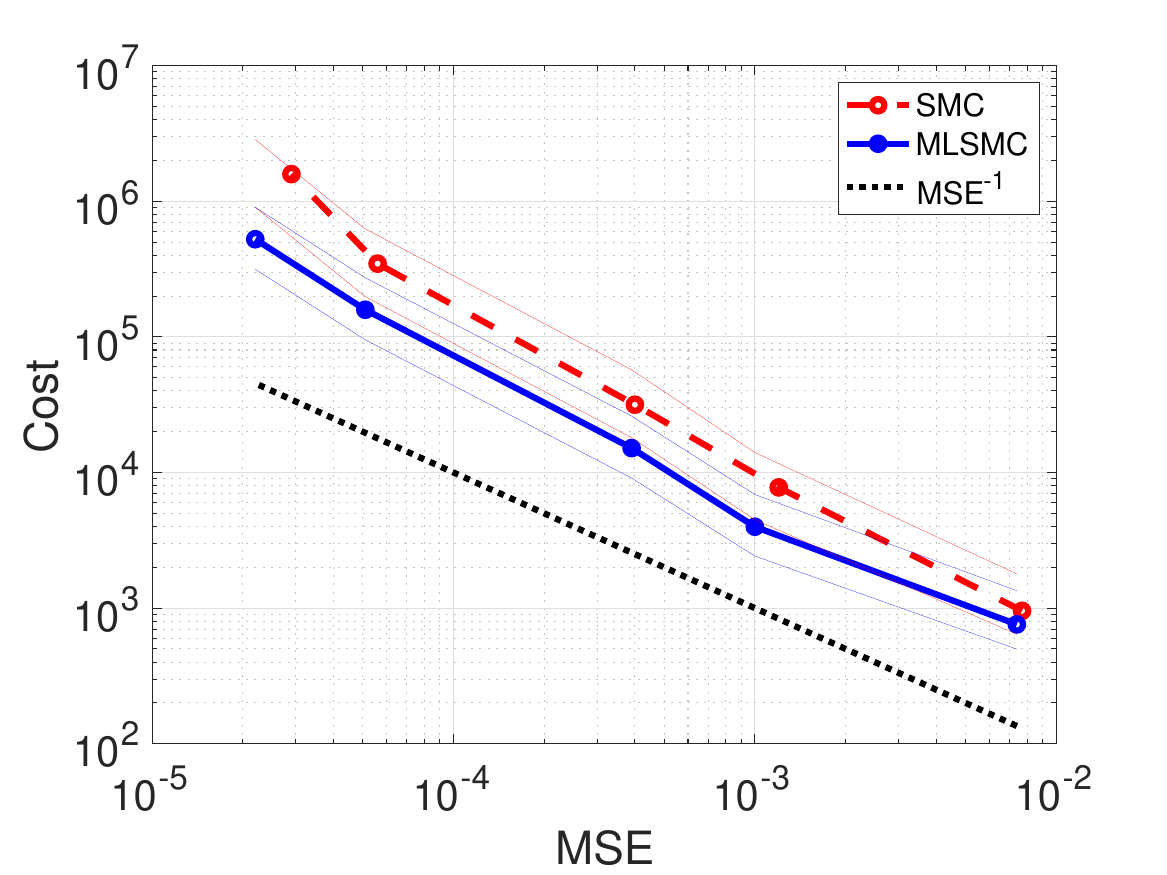}
\includegraphics[width=0.45\textwidth]{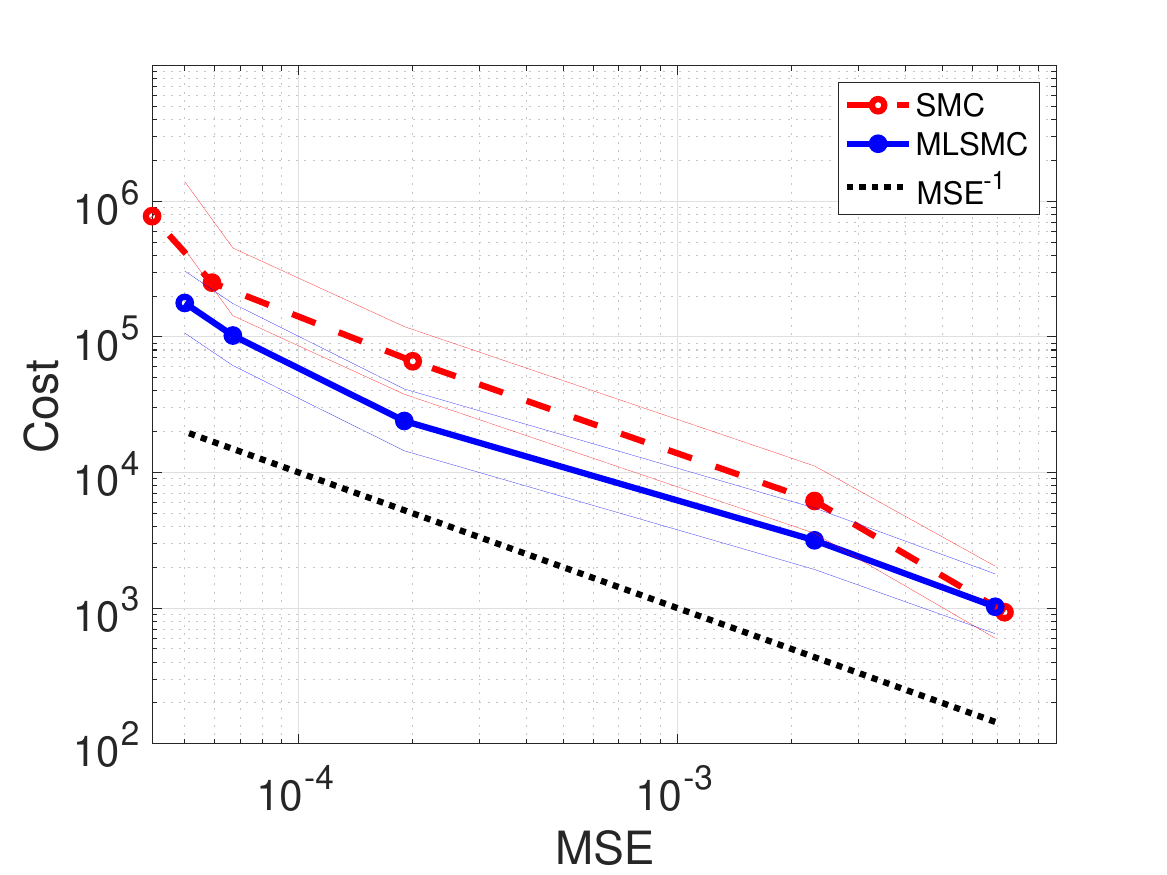}
 \caption{MNIST Classification problem: error vs cost plots for SMC and MLSMC, using TNN priors.
Left: $\alpha=1.4$. Right: $\alpha=1.1$.  Credible sets are provided in the thin blue and red curves.}
  \label{fig:mnist2}
\end{figure}

\subsection{Reinforcement Learning}
\subsubsection{Setup}
Our final numerical example will consist of an inverse reinforcement learning problem.
Unlike the previous two examples,
reinforcement learning is not concerned with pattern formations, but instead decision making. Specifically it is
used to solve stochastic optimal control problems. Before we continue with the specific example, we first
recall some common notation and provide our setup for Bayesian reinforcement learning, which is largely motivated and taken from \cite{sell2020dimension}.

A Markov decision process is defined by a controlled Markov chain $\{X_t\}_{t \in \mathbb{N}}$, referred to as the state process, the control process
$\{A_t\}_{t \in \mathbb{N}}$ and the an optimality criterion. The state process takes values in a bounded set $\mathcal{X} \subset \mathbb{R}^d$.
The control process is $\mathcal{A}$-valued with $\mathcal{A} = \{1,\ldots,M\}$. Therefore the state process propagates by
$$
X_{t+1}|(X_{1:t}=x_{1:t},A_{1:t}=a_{1:t}) \sim p(\cdot|x_t,a_t) \, ,
$$
where for any state-action pair $(x_t,a_t), p(\cdot|x_t,a_t)$ is a probability density. Now let $r:\mathcal{X} \rightarrow \mathbb{R}$ be the
reward function, then the accumulated reward given in a policy and initial state $X_1=x_1$ is
$$
C_{\mu}(x_1) = \mathbb{E}_{\mu}\Big[\sum^{\infty}_{t=1}\beta^tr(X_t)|X_1=x_1\Big] \, ,
$$
where $\beta \in (0,1)$ is a discount factor and $\mu:\mathcal{X} \rightarrow \mathbb{R}$ is the policy mapping.
A policy $\mu^*$
is optimal if $C_{\mu}^{*}(x_1) > C_{\mu}(x_1)$, for all $(\mu,x_1) \in \mathcal{X} \times \mathbb{R}$, and can be found by solving
Bellman's fixed point equation \cite{DPB85}.
Using this fixed point solution $v:\mathcal{X} \rightarrow \mathbb{R}$, we can derive the optimal policy by
\begin{equation}
\label{eq:optimal}
\mu^*(x) = \argmax_{a \in \mathcal{A}} \Big[  \int_{\mathcal{X}} p(x'|x,a)v(x') dx' \Big] \, ,
\end{equation}
that is, the optimal action at any state is the one that maximizes the expected value function at the next state.
{It is assumed that the state evolution is deterministic, i.e.,
there is a map {$\mathcal{T}: \cX \times \cA \rightarrow \cX$} such that
$p(x'|x,a)=\delta_{\mathcal{T}(x,a)}(x')$.}
%Then %It is assumed
Noise is added to model imperfect action selections
so that at each time step
the chosen action is a random variable given by
\begin{equation}
\label{eq:action}
A_t(x_t) = \argmax_{a \in \mathcal{A}}\Big[ v(\mathcal{T}(x_t,a))
%\int_{\mathcal{X}} p(x'|x,a)v(x')dx'
+ \epsilon_t(a)  \Big], \qquad \epsilon_t \sim \mathcal{N}(0,\sigma^2I) \, .
\end{equation}
Now what remains is to define our likelihood function associated to our example.
Our data will consist of a collection of noise corrupted state-action pairs %$y = \{y_t\}^T_{t=1}=
$\{x_t,a_t\}_{t=1}^T$,
 and the aim is to infer the value function \eqref{eq:optimal} that leads to the actions $a_t$ for the current state $x_t$. Using the noisy action selection process \eqref{eq:action}
 we have the likelihood function defined as
 \begin{equation}
 \label{eq:like}
 \mathcal{L}(a_{1:T}|x_{1:T},v,\sigma) =
 \prod^T_{t=1}p(a_t|x_t,v,\sigma) = \prod^T_{t=1}p(a_t|v_t,\sigma) \, ,
 \end{equation}
where $v_t\in\bbR^M$ is defined by $v_{t,k}=v(T(x_t,a=k))$.
%containing the relevant evaluations of the value function to calculate
%$p(a_t|x_t,v,\sigma)$
%the likelihood at $y_t$.
%For this experiment we will assume the optimal action
%is taken as $k=1$, therefore the probability $p(a=1|v,\sigma)$ where the
Following from \eqref{eq:action},
the factors in \eqref{eq:like} have a closed form expression
$$
p(a_t | v_t, \sigma) = \frac1{\sigma} \int_{\bbR} \phi\left(\frac{t-v_{t,a_t}}{\sigma}\right)
\prod_{\stackrel{i=1}{i\neq a_t}}^M \Phi\left(\frac{t-v_{t,i}}{\sigma}\right) dt \, ,
$$
where $\phi$ and $\Phi$ denote the standard normal PDF and CDF, respectively.
The derivation %for this likelihood
can be found in \cite{sell2020dimension}.

%Therefore we some computations, performed in \cite{sell2020dimension}, we can express our gradient of the likelihood, w.r.t. to the neural network
%parameter $\theta$, as
%\begin{equation}
%\label{eq:like2}
%\nabla_{\theta} \log p(y|v,\sigma) = \nabla_{\theta}\Big(\prod^T_{t=1}p(a_t|v_t,\sigma)\Big) = \nabla_{\theta} \sum^t_{t=1} \log p(a_t|v_t,\sigma) = \sum^T_{t=1}\nabla_{\theta} \log p(a_t|v_t,\sigma).
%\end{equation}
%\textcolor{blue}{**Put in a few references, and look at the Cheetah example.} %Ask Sumeet if it is inverse RL example or just RL?}

\begin{figure}[h!]
\centering
\includegraphics[width=0.35\textwidth]{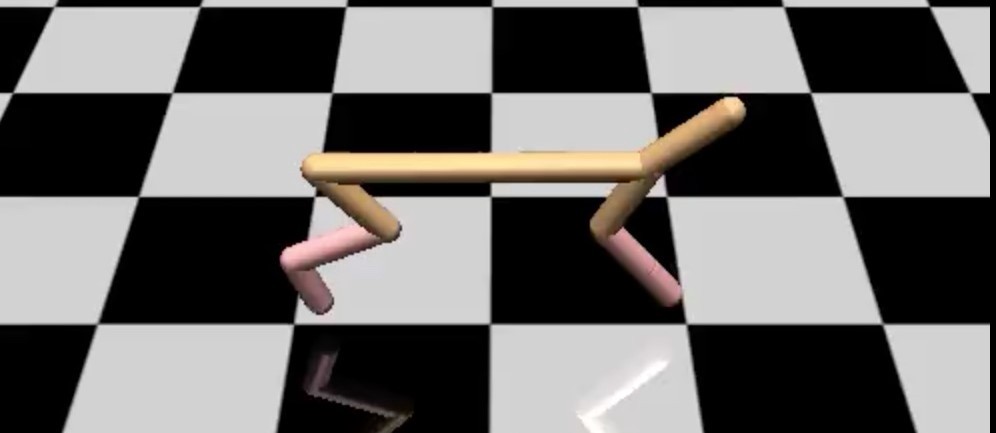}
 \caption{Plot of HalfCheetah, which has states $x_t \in \mathbb{R}^{17}$.
 Its goal is to run as quickly as possible, while not moving
 its body parts more than necessary.}
    \label{fig:half_cheetah}
\end{figure}

\subsubsection{Experiment}

For our  Bayesian reinforcement learning example,
we consider the HalfCheetah example \cite{TEM12},
where the goal %of reinforcement example
is for the HalfCheetah
to run as fast as possible without moving its body more than necessary. %We specify a
The state space %of $\textrm{dim}(\mathcal{X})=
is $\cX=\mathbb{R}^{17}$,
and the action space is $\cA=\{1,\dots,8\}$.
%to be $\textrm{dim}(\mathcal{A})= \mathbb{R}^{8}$.
We also take $T=100$ for the observations, where the data generation is taken similarly to \cite{sell2020dimension}, and $\sigma^2=0.01^2$.
As before we aim to show the benefit of MLMC for the SMC sampler,
where we choose values $\alpha=\{1.7, 1.9, 2, 3\}$ with levels $L \in \{3,4,\ldots,7\}$.
For our truth we again take a high-resolution solution to our problem, similarly to what was
done for the 2D spiral experiment.
These values are again chosen such that we can attain the canonical rate of convergence.
Our results are presented in Figure \ref{fig:reinf_results}
which compare both methodologies with TNN priors, for the different values of $\alpha$. Similarly to our previous results we notice a a bigger difference
is cost for lower values of the MSE, which indicate both error-to-cost rates are different. As done for the previous experiments,
 we plot the canonical rate in black to verify our methodology attains the rate.

Finally our last experiment is to verify that one cannot attain the canonical rate
if we assume that $\alpha<1.5$. In Figure \ref{fig:reinf_results2} we present similar experiments
to Figure \ref{fig:reinf_results}, but with modified values of $\alpha \in \{1.1,1.4\}$.
We can observe that the results are similar to the previously attained,
in the sense that the complexity grows at a faster rate than it does in
the canonical case, but still slower than for the single level approach.
%which is that the
%difference
%in cost to attain the same MSE is not as big as for values of $\alpha \geq 1.5$, which a
Again this verifies the theory.

\begin{figure}[h!]
\centering
\includegraphics[width=0.45\textwidth]{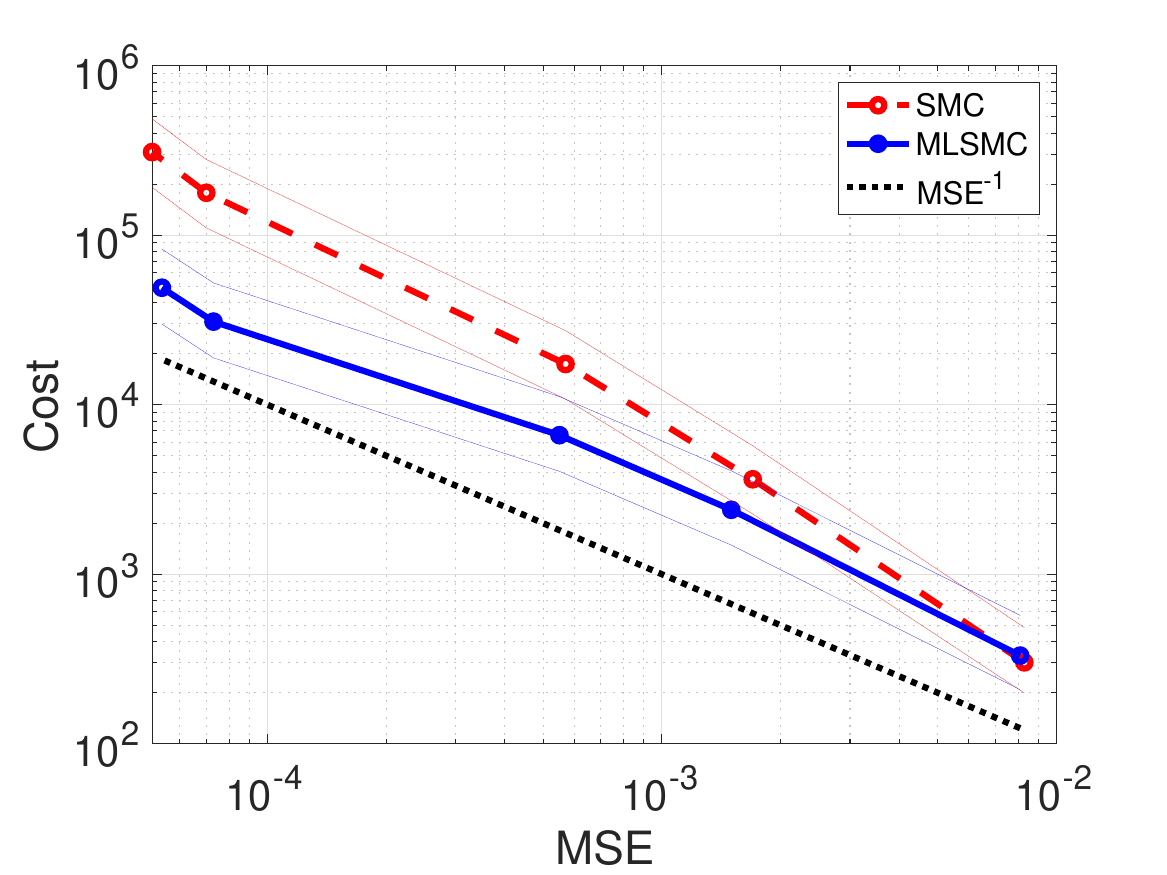}
\includegraphics[width=0.45\textwidth]{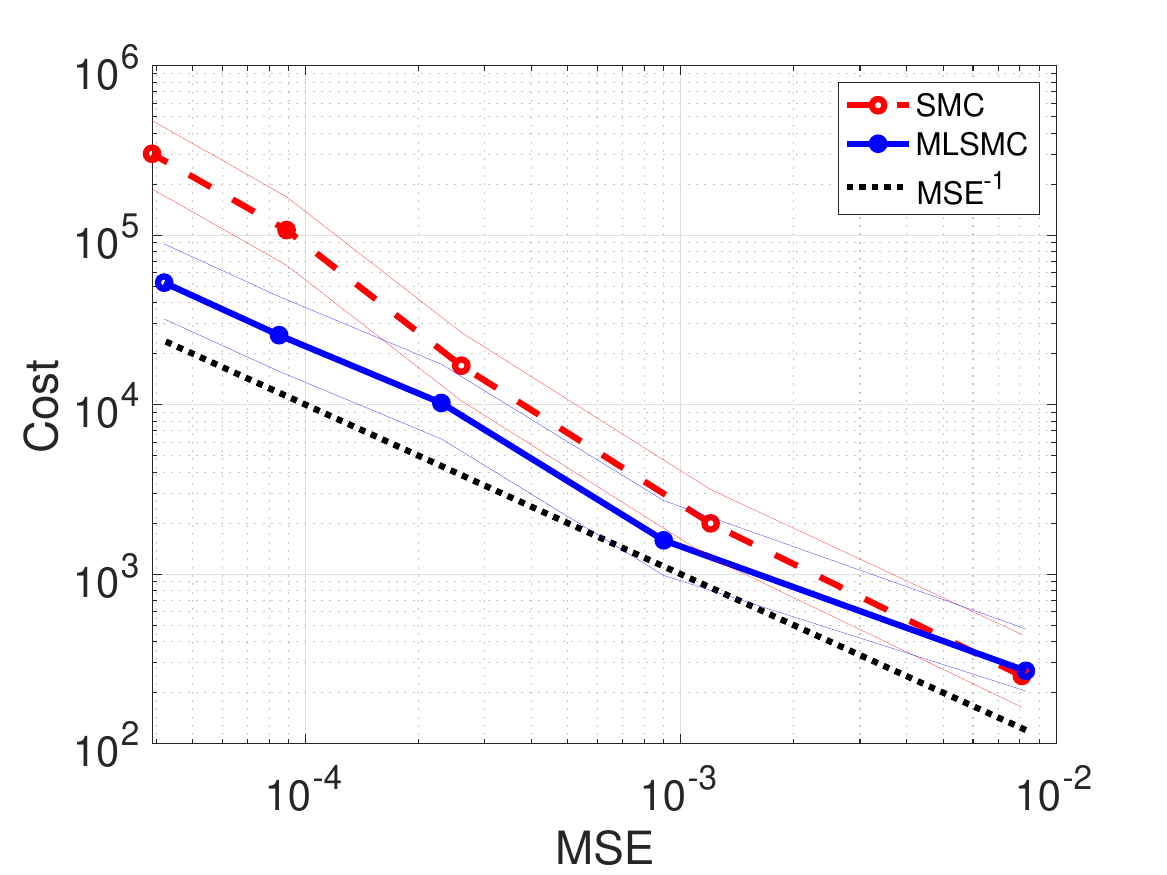}
\includegraphics[width=0.45\textwidth]{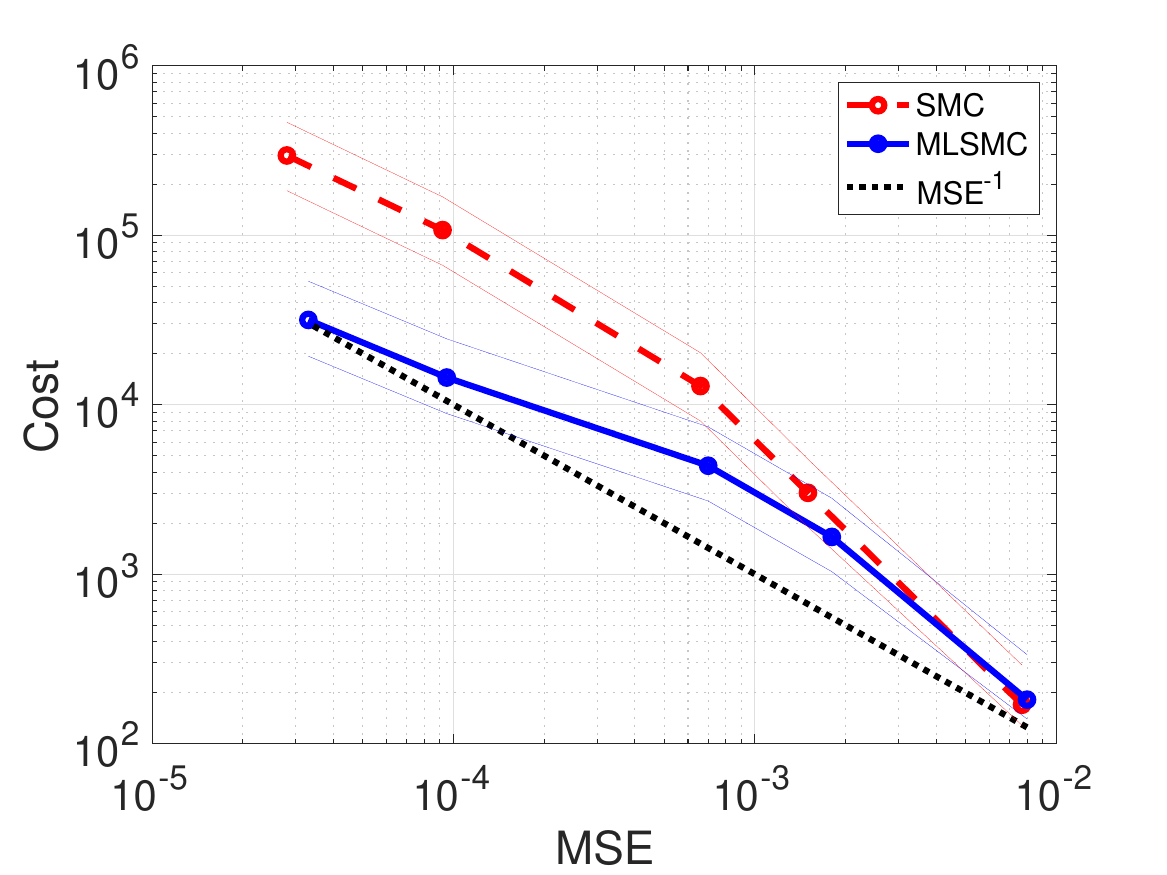}
\includegraphics[width=0.45\textwidth]{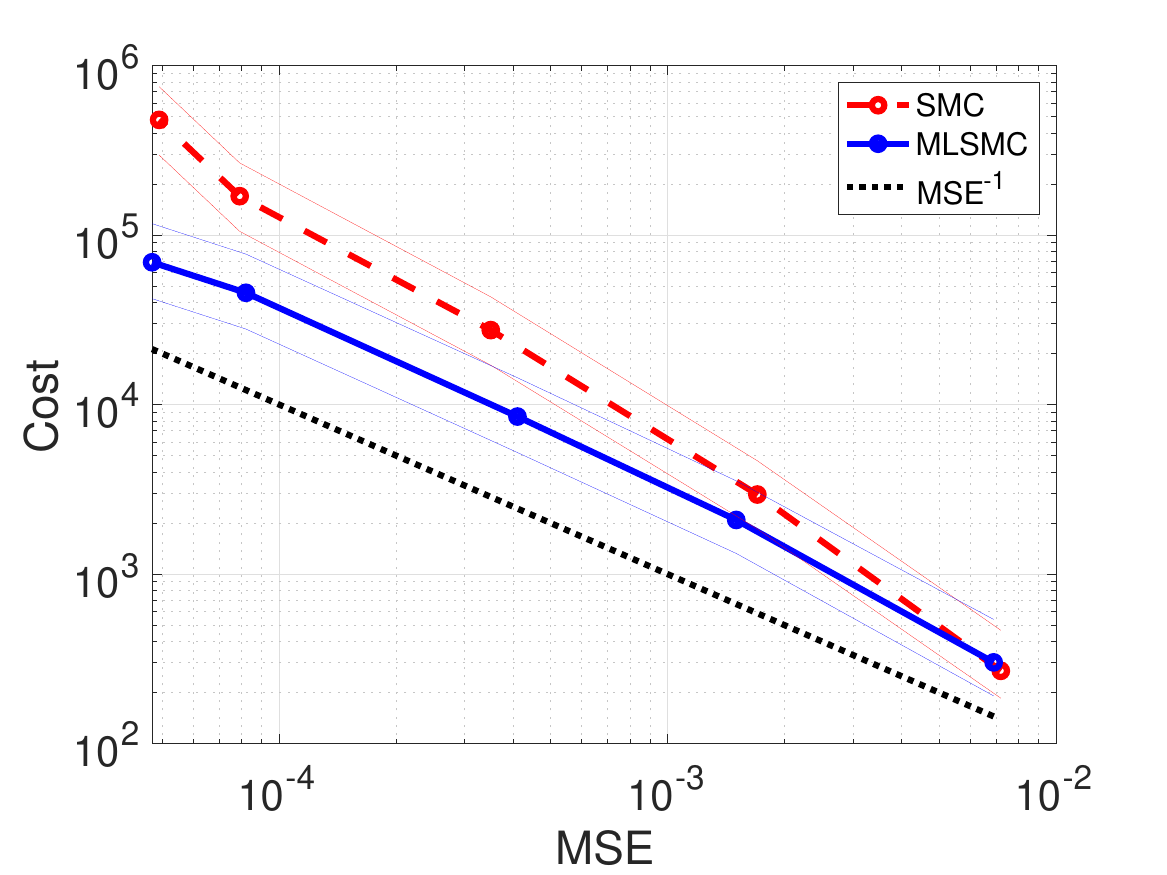}
 \caption{Reinforcement learning: error vs cost plots for SMC and MLSMC, using TNN priors.
 Top left: $\alpha=3$. Top right: $\alpha = 2.0$. Bottom left: $\alpha=1.9$. Bottom right: $\alpha=1.7$.
  Credible sets are provided in the thin blue and red curves.}
    \label{fig:reinf_results}
\end{figure}

\begin{figure}[h!]
\centering
\includegraphics[width=0.45\textwidth]{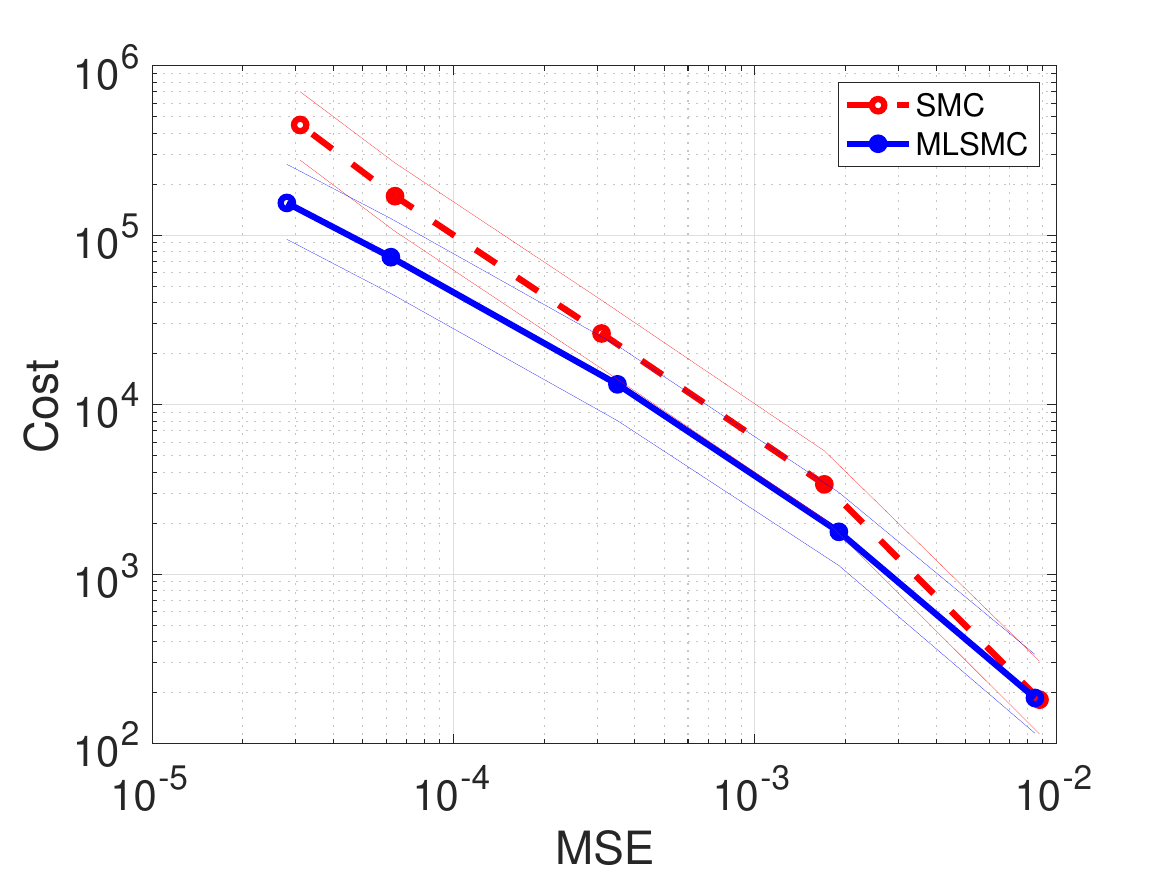}
\includegraphics[width=0.45\textwidth]{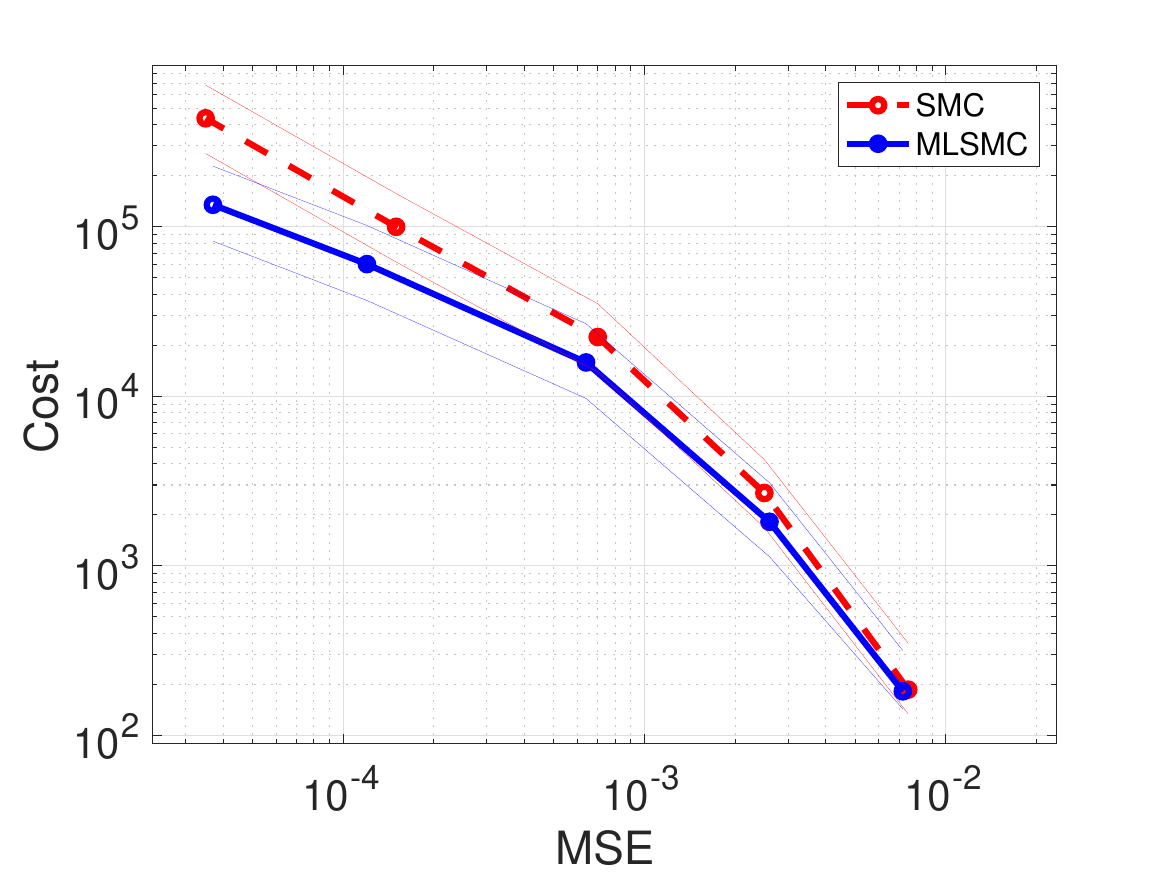}
 \caption{Reinforcement learning: error vs cost plots for SMC and MLSMC, using TNN priors.
Left: $\alpha=1.4$. Right: $\alpha=1.1$.  Credible sets are provided in the thin blue and red curves.}
    \label{fig:reinf_results2}
\end{figure}

\subsection{\textcolor{black}{Sensitivity Analysis}}

\textcolor{black}{In this subsection, we would like to conclude with a simple sensitivity analysis on various parameter choices.
Our analysis will be based on how the computational complexity, related to the MSE-to-Error rate is effected, 
as this is our main metric we are concerned with (as related to the MLMC literature).}
\textcolor{black}{
We modify a number of parameter choices such as the activation function, where we include both the ReLU and Tanh activation function $\sigma$,
as discussed in Section \ref{sec:tnn}. We also modify the dimension of our problem, and the parameter $\alpha$ which we have already tested,
but we compare again (assuming it remains TNN).  All these experiments are conducted on the MNIST
dataset example. Our final parameter we modify is the maximum value of $l$, which we denote as $L$. We use a base case to compare the effect
of modifying such parameters. Our results are presented in Table \ref{table:1}.}

\begin{table}[h!]
\centering % used for centering table
\textcolor{black}{
\begin{tabular}{c c c c c c} % centered columns (4 columns)
\hline\hline %inserts double horizontal lines
Case & $\alpha$ & $d$ (dimension) & $\sigma$ & $L$  & MSE-Error Rate \\ [0.5ex] % inserts table
%heading
\hline % inserts single horizontal line
1 & 4 & 100 & ReLU & 7  & $\mathcal{O}(\epsilon^{-2.04})$  \\ % inserting body of the table
2 & 2 & 100 & ReLU & 7  &  $\mathcal{O}(\epsilon^{-2.02})$ \\ % inserting body of the table
3 & 4 & 400 & ReLU & 7  &  $\mathcal{O}(\epsilon^{-2.07})$ \\ % inserting body of the table
4 & 4 & 100 & tanh & 7  &   $\mathcal{O}(\epsilon^{-2.04})$\\ % inserting body of the table
5 & 4 & 100 & ReLU & 3  &  $\mathcal{O}(\epsilon^{-2.34})$ \\ % inserting body of the table
\hline \\ %inserts single line
\end{tabular}}
\caption{\textcolor{black}{Sensitivity analysis for the complexity rates, based on various parameter choices.}} % title of Table
\label{table:1} % is used to refer this table in the text
\end{table}
\textcolor{black}{As we can observe from Table \ref{table:1}, we are comparing various parameter choices to the base setting we used for the MNIST dataset, which is Case 1. 
The canonical rate that we should obtain is $\mathcal{O}(\epsilon^{-2})$, which is what we see when modifying most of the parameters. This demonstrates the robustness of our methodology. 
From the different cases, it is clear that Case 5 is the worst, which differs through the choice of $L=3$. This is expected, as for small choices of $L$ the rate is expected to be worse. }

\begin{remark}[Parameter Tuning]
\textcolor{black}{It is important to note that first MCMC must be properly tuned within SMC. 
This can be done adaptively, and SMC is particularly convenient in this sense, because at each MCMC stage we have an ensemble of particles with approximately the correct distribution which are used to estimate moments for accelerating the MCMC. 
Then one needs to choose number of intermediate MCMC steps, and possibly between-level tempering schedule, which can both be done adaptively–the former by using an autocorrelation metric and the latter by targeting an effective sample size of say $N_l/2$. 
Finally, there is the question of how the MLMC works with different rates of convergence of the TNN, i.e. different $\alpha$.
 This is well understood and explored further in the examples.}
\end{remark}

\subsection{\textcolor{black}{Practical Applications}}

An obvious question is whether any of this is useful for practical applications.
To answer that question, we have executed two experiments, presented in the following subsections.

%\textcolor{blue}{
{\color{black}
\subsubsection{IMDb big dataset}

The first experiment is sentiment (binary) 
classification for the IMDb dataset of $50,000$
movie reviews\footnote{https://huggingface.co/datasets/stanfordnlp/imdb} \cite{maas2011}. 
The dataset is split into $N=25,000$ data for each of training and testing.
The first and last training examples, 
labelled with negative and positive sentiment respectively,
are given in Figure \ref{fig:imdb}.
Then we preprocess the reviews by embedding them into Euclidean space using SBERT
embeddings \cite{reimers-2019-sentence-bert},  
based on the models 
\texttt{all-mpnet-base-v2}
\footnote{https://huggingface.co/sentence-transformers/all-mpnet-base-v2}
 \cite{song2020mpnet}
(768 dimension) and 
\texttt{all-MiniLM-L6-v2}\footnote{https://huggingface.co/sentence-transformers/all-MiniLM-L6-v2} 
\cite{song2020mpnet}
(384 dimension).
In other words, frozen weights from \texttt{all-mpnet-base-v2}
(or \texttt{all-MiniLM-L6-v2}) model are used until the $768$ (or $384$) 
dimensional  \texttt{[CLS]} output, 
after which we have appended a Bayesian logistic regression model.
%(i) one hidden layer with $128$ neurons, (ii) ReLU activation, 
%(iii) a final linear layer, and (iv) softmax output.
Since we only have a single layer, the parameter prior is isotropic 
%independent 
$\overline{\pi}_1 = N(0,\sigma^2 I_{768})$, with $\sigma^2=0.1$.
% on biases and $\sigma^2=\cdots$ on weights.
%\begin{figure}[h!]
%\centering
%\includegraphics[width=0.60\textwidth]{imdb-fig.png}
%    \caption{Example negative and positive reviews from the IMDb training dataset.
%    The negative review is 288 words, and the middle is omitted (...).}
%    %Absolute error of the difference between the finer and coarse level, with ReLU activation function.
%    \label{fig:imdb}
%\end{figure}
%
%\begin{figure}[h!]
%\centering
%\includegraphics[width=\textwidth]{entropy.png}
%    \caption{Average total and epistemic entropy, split by correct and incorrect predictions on test data.}
%    %Absolute error of the difference between the finer and coarse level, with ReLU activation function.
%    \label{fig:entropy}
%\end{figure}

%\begin{figure}[h!]
%  \centering
%  % Subfigure (a)
%  \begin{subfigure}[b]{0.47\textwidth}
%    \centering
%    \includegraphics[width=.47\textwidth]{imdb-fig.png}
%    \caption{Example negative and positive reviews from the IMDb training dataset. The negative review is 288 words, and the middle is omitted (...).}
%    \label{fig:imdb}
%  \end{subfigure}
%  
%  \vspace{1em} % Adjust vertical space if needed
%  
%  % Subfigure (b)
%  \begin{subfigure}[b]{.47\textwidth}
%    \centering
%    \includegraphics[width=.47\textwidth]{entropy.png}
%    \caption{Average total and epistemic entropy, split by correct and incorrect predictions on test data.}
%    \label{fig:entropy}
%  \end{subfigure}
%  
%  \caption{(a) IMDb review examples and (b) entropy analysis.}
%  \label{fig:combined}
%\end{figure}

\begin{figure}[h!]
  \centering
  \begin{subfigure}[b]{0.47\textwidth}
    \centering
    \includegraphics[width=\linewidth]{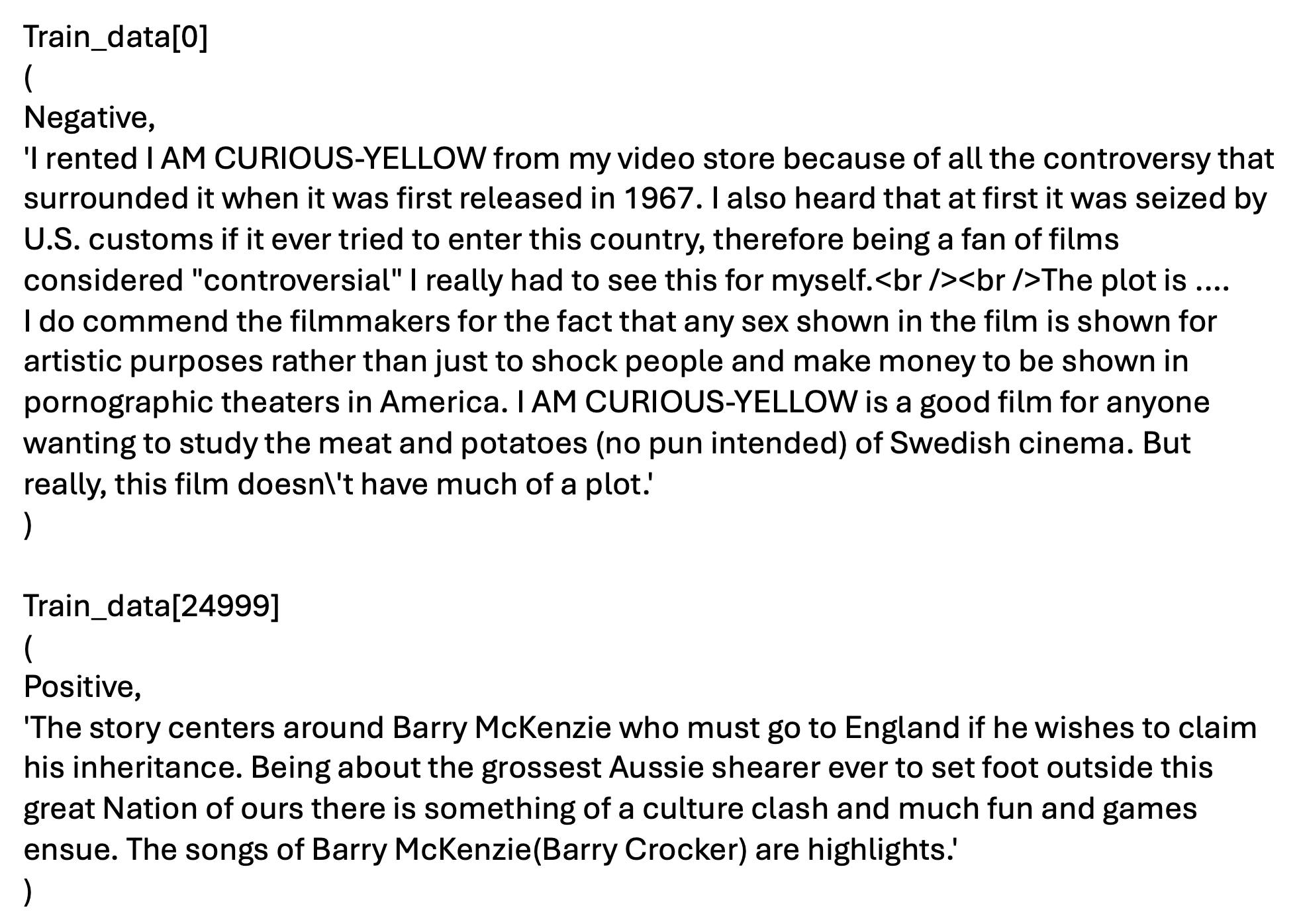}
    \caption{Example negative and positive reviews from the IMDb training dataset. The negative review is 288 words, and the middle is omitted (...).}
    \label{fig:imdb}
  \end{subfigure}
  \hfill
  \begin{subfigure}[b]{0.47\textwidth}
    \centering
    \includegraphics[width=\linewidth]{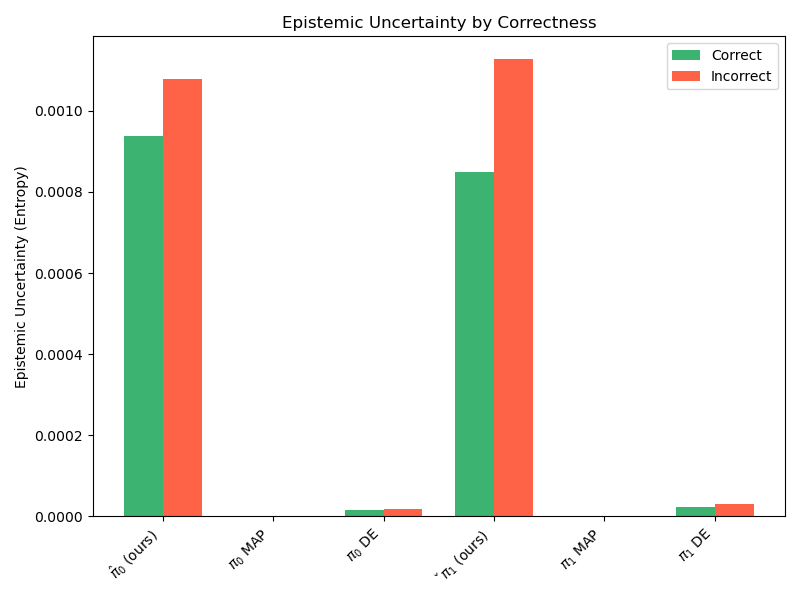}
    \caption{Average epistemic entropy on test data, split by correct and incorrect predictions.}
    \label{fig:entropy}
  \end{subfigure}
  \caption{(a) IMDb movie review dataset examples 
%  with sentiment labels 
  and (b) entropy analysis.}
  \label{fig:combined}
\end{figure}

We refer to the coarse and fine posteriors as $\pi_0$ and $\pi_1$, respectively.
One expects them to correlate, so we conduct a canonical correlation analysis, as follows.
Let $X_1 \in \bbR^{N \times 768}$ and $X_0 \in \bbR^{N \times 384}$
denote the training embeddings. Let $M=X_1 X_0^{\sf T}$, and 
denote its SVD as 
$$
U \Sigma V^{\sf T} = M  \, .
$$
Projecting $X_1U_{:,:384}$ 
onto the 384 non-zero singular values captures  96\% of the variance,
so it is reasonable to correlate the two models by rotating into a common coordinate frame:
%Hence we redefine 
$X_1 \leftarrow X_1 U$ and $X_0 \leftarrow X_0 V$.

We use SMC with $P_0=30$ with kernel defined by 
20 steps of HMC (adaptive step-size and minimum trajectory length $\tau=0.03$) \cite{Buchholz2021Adaptive}
and adaptive tempering (keeping ESS$\geq P_0/2$) to build the estimator of 
$\eta_0=q_1$. 
%This requires $\approx 100$ steps.
%and takes $\approx 290$ seconds (not accounting for the $30\times$
%parallelism which could be leveraged with appropriate resources). 
The resulting posterior predictive probability estimator of a positive sentiment classification, 
\begin{equation}\label{eq:imdb-est0}
\hat{p}_0(x,\cD) = \widehat{\pi}_0(f_0(x,\cdot)) = 
 \eta_0^{P_0}[f_0(x,\cdot)] = \frac1{P_0}\sum_{i=1}^{P_0} f_0(x,\theta_0^i) \, ,
%\approx 
\end{equation}
approximates the true posterior predictive probability estimator at level $0$:
$$
%\overline{\pi}_0(y=1 | x, \cD) = 
\int_{\theta_0 \in \Theta_0} p_0(y=1 | x, \theta_0) \overline{\pi}_0( d\theta_0 | \cD) =
\int_{\theta_0 \in \Theta_0} f_0(x, \theta_0) \pi_0(d\theta_0 ) 
%d\theta_0 
= \eta_0[f_0(x,\cdot)] (=\eta_1[f_0(x,\cdot)]) \, .
$$
Defining the usual hard threshold classier by 
%prediction by 
\[
\hat{y}(x) = \begin{cases}
1, & \text{if } \hat{p}_0(x,\cD) > 0.5\, , \\
0, & \text{otherwise} \, .
\end{cases} \, 
\]
we achieve an accuracy of $80\%$ on the test data, where
$$
{\sf Accuracy}=\frac1{|\cD_{\sf test}|}\sum_{(x,y) \in \cD_{\sf test}} {\bf 1}_{\{\hat{y}(x)=y\}} \, .
$$
Clearly the cost to run the algorithm targeting $\pi_1$ is (at least) double, given that it has double the number of parameters.
We build the 2-level MLSMC estimator in two ways. Let
$\theta_{1,:384} = \theta_0^i$, 
and augment to $\theta_1^i=(\theta_0^i,\tilde{\theta}_1^i)$ 
with $\overline{\pi}_1(\tilde{\theta}_1) =  q_1(|\theta_0)$,
i.e. 
\begin{equation}\label{eq:imdb-extend}
\theta^i_{1,384:} = \tilde{\theta}_1^i \sim N(0,\sigma^2 I_{384}) \, .
\end{equation}
Let $P_1=P_0$, and as in \eqref{eq:cond_disc} define
$$
G_1(\theta_1) = \frac{p_1(y_{1:N} | \theta_1)}{p_0(y_{1:N} | \theta_0)} \, .
$$
The 2 estimators, each of which are almost free to 
%seconds to 
build from $\widehat{\pi}_0$, are as follows.
\begin{itemize}
\item[(i)] The one-level simple self-normalized importance sampling estimator \eqref{eq:ml_est}:
\begin{equation}\label{eq:one-level}
\widehat{\pi}_1(f_1) = \frac{\eta_1^{P_1}(G_1 f_1)}{\eta_1^{P_1}(G_1)} \, .
\end{equation}
\item[(ii)] Let $\overline \theta_0 \sim K_0(\theta_0,\cdot)$ denote another mutation step with all $P_1$
particles, define 
$\overline{\eta}^{P_1}_1$ as the corresponding empirical measure 
%and $\overline{\eta}^{N_1}_1$ as the extension 
extended to $\overline{\theta}_1$ by sampling from the prior as in \eqref{eq:imdb-extend}, 
and build the two-level version of \eqref{eq:ml_est} as
%Let
$$
\widehat{\overline{\pi}}_1(f_1) = \eta_0^{P_0}(f_0) + 
\underbrace{\left(\frac{\overline{\eta}_1^{P_1}(G_1 f_1)}{\overline{\eta}_1^{P_1}(G_1)} - \overline{\eta}_1^{P_1}(f_0)\right )}_{\sf increment}\, .
$$
This may seem slightly contrived, but notice that it is quite natural when $\eta_0$ itself comes from an SMC sampler,
since this is what happens in one stage of the algorithm:
\begin{itemize}
\item build an estimator of level $l$ from mutated samples at level $l-1$, 
\item construct the level $l$ increment using these samples.
\end{itemize}
If we attempted to build such a two-level estimator 
%was built 
without an additional mutation then we would arrive with 
the one-level estimator \eqref{eq:one-level} again. 
\end{itemize}

The results on test data are presented in Table \ref{tab:imdb}.
Accuracy is defined above. 
The negative log-likelihood (NLL) of a posterior predictive estimator $\hat{p}$ is the following
$$
{\sf NLL} = -\sum_{(x,y) \in \cD_{\sf test}} y \log(\hat{p}(x)) + (1-y) \log(1-\hat{p}(x)) \, .
$$
Note that {\sf NLL}$/|\cD_{\sf test}|$, the test empirical risk, 
is an estimator of the {\em generalization error} (expected loss over the true data distribution).
The area under the receiver operating characteristic curve (AUC-ROC), 
which plots the true positive rate against the false positive rate over thresholds in $[0,1]$,
corresponds to the probability that 
a randomly chosen positive example is assigned a higher predictive score than a randomly chosen negative example.
%for randomly selected positive and negative examples, 
%the predictive output for the positive is greater than the predictive output for the negative. 
\begin{table}[h]
\begingroup
\color{black}
\centering
\begin{tabular}{lccc}
\hline
\textbf{Estimator} & \textbf{Accuracy} & \textbf{NLL} & \textbf{AUC-ROC}  \\
\hline
$\widehat{\pi}_0$ (coarse estimator) & 0.7974 & 11100 & 0.8810  \\
$\widehat{\pi}_1$ (importance sampling) & 0.8660 & 8426 & 0.9391   \\
$\widehat{\overline{\pi}}_1$ ($\widehat{\pi}_0$+mutated coarse increment) & 0.8699 & 8178 & 0.9417  \\
\hline
\end{tabular}
\caption{Performance of different estimators.}
\label{tab:imdb}
\endgroup
\end{table}

\begin{remark}[Extensibility]
%The astute reader will observe that t
The setting of the example above does 
not precisely correspond to the construction in the paper, in the sense that it only
involves 2 levels and does not leverage the TNN.
%On the one hand this shows 
This 
%is not meant to pull the wool over the eyes of the reader, but rather to 
illustrates the generality of the approach.
%where a particular practical application 
A given application may not map immediately to the version of the approach 
presented in an obvious way.
Indeed, for the application above we originally tried a single hidden layer of growing width,
but realized that the improvements were invisible to the frequentist performance `metrics'
presented in Table \ref{tab:imdb}: they were essentially equivalent with width 0, i.e. 
Bayesian logistic regression. So we had to be more creative.
%Hence the version presented.
%Always there is a way, and for this practical application, this is the way.
%, which are built on the frequentist framework.
\end{remark}

\noindent {\bf Uncertainty quantification. }
%The greatest benefit of 
At convergence, 
Bayesian methods deliver Bayes optimal estimators of quantities as shown in Table \ref{tab:imdb},
however 
%one might argue that 
the greatest benefit of Bayesian methods may lie in the ability to quantify uncertainty (UQ).
To that end, we propose to quantify the 
%predictive total and 
epistemic entropy, defined as follows 
\cite{hullermeier2021aleatoric,depeweg2018decomposition,shaker2020aleatoric,krause2025probabilistic}
\begin{eqnarray}
H_{\sf tot}(x) &=& - \hat{p}(x)\log (\hat{p}(x)) - (1-\hat{p}(x)) \log( 1- \hat{p}(x) ) \, , \\
H_{\sf al}(x) &=& \widehat{\pi} \left [ -f(x,\cdot )\log( f(x,\cdot) ) - (1- f(x,\cdot) ) \log( 1- f(x,\cdot) ) \right ] \, , \\
H_{\sf ep}(x) &=& H_{\sf tot}(x) - H_{\sf al}(x) \, .
\end{eqnarray}
Epistemic entropy quantifies uncertainty which can be reduced with more data,  
and is {\em only} captured by Bayesian methods.
It can be viewed as the mutual information between parameter and predictive posterior random variables for input $x$,
and as such is $0$ by definition for point estimators, which yield deterministic predictive estimators.
The average epistemic entropy over the test data, split by correct and incorrect predictions, 
are presented in Figure \ref{fig:entropy}. 
It is clear that this quantity is predictive of misclassifications,
and it is also clear that the epistemic entropy goes down for correct and up for incorrect predictions as the quality of the estimator improves,
which supports the use of high-quality Bayesian estimation for applications where reliability is especially important.
The results show clearly that our {\em consistent Bayesian} method significantly outperforms
the maximum a posteriori (MAP) estimator, as well as a deep ensemble \cite{lakshminarayanan2017simple} (DE) of 30 MAP estimators. 
%Each such classical estimator takes 3-7 seconds to compute, which means it is cheaper but on par with our method.
It is noted that deep ensembles are considered to be the state of the art in scalable Bayesian deep learning, and have 
%meanwhile 
been shown to outperform other scalable approaches
\cite{lakshminarayanan2017simple,wilson2022evaluating,gustafsson2020evaluating}. 
%\cite{}.

For the present example the levels are too far apart to provide a significant 
effective sample size, so we introduce another SMC sampler to bring in the new dimensions
%principal components 
one batch at a time.
With only 20 additional steps, we achieve a full effective sample size for 
the approximation of $\pi_1$, which we denote $\check{\pi}_1$
(note that the cost is still lower than targeting $\pi_1$ directly).
This example further illustrates that the methodology in its most elegant form may need 
to be massaged for practical applications, but remains a valuable ideal target.

\subsubsection{Return of MNIST, digits $0,\dots,7$}

\begin{table}[h!]
\begingroup
\color{black}
\centering
\begin{tabular}{lccccc}
\hline
\textbf{Estimator} & \textbf{Accuracy} & \textbf{AUC-ROC}  & \textbf{$H_{\sf ep}$ (ID)} & \textbf{$H_{\sf ep}$ (OOD)}\  \\
\hline
MAP &  0.9823 & 0.8763 & 0 & 0 \\
DE &  0.9845 & 0.9195  &  0.073 &  0.3350  \\
SMC & 0.9842 & 0.9317 & 0.124 &  1.7831  \\
MLSMC & 0.9837 & 0.9406 & 0.144 & 1.8292  \\
\hline
\end{tabular}
\caption{Performance of different estimators for MNIST dataset.}
\label{tab:mnist}
\endgroup
\end{table}
%As a final example we consider the MNIST dataset, where 
We conclude this section with an example where we filter {\em the whole MNIST dataset} 
to digits $0,\dots,7$. 
%We test Accuracy
%and the area under the receiver operating characteristic curve (AUC-ROC) 
%on the standard test dataset, 
%against the maximum a-posteriori (MAP) estimator and deep ensemble (DE) \cite{},
%the state-of-the-art for approximate Bayesian deep learning.
%We also measure epistemic entropy on the held-out out-of-domain (OOD) 8,9 digits.
We compare MAP, DE, SMC, and 7-level MLSMC (setting is the same as in Section \ref{sec:mnist}), 
on a range of classical `metrics' which include: 
accuracy, AUC-ROC, and average $H_{\sf ep}$, as described in the previous subsection. 
For the latter, we also consider 
%the average prediction for 
a sample of out-of-domain (OOD) 8,9 digits.
%epitemic entropy.
%where we consider a DE of 30 MAP estimators.
Table \ref{tab:mnist} provides the comparison.
%, where we notice that for a
Accuracy all methods are similar. 
However for the AUC-ROC and the entropy, we start to notice an improvement which is what we expect. This is similar
to the results for the IMDB dataset.

%\begin{tabular}{lrrrrrrr}
%\toprule
%Estimator & Accuracy & NLL & AUC ROC & Total Unc. (Correct) & Epis. Unc. (Correct) & Total Unc. (Incorrect) & Epis. Unc. (Incorrect) \\
%\midrule
%$\hat{\pi}_0$ (ours)  & 0.6135 & 9.37e-4 & 0.6675 & 0.0011 \\
%$\pi_0$ MAP  & 0.5843 & 0.0000 & 0.6588 & 0.0000 \\
%$\pi_0$ DE  & 0.5843 & 0.0000 & 0.6587 & 0.0000 \\
%$\check{\pi}_1$ (ours)   & 0.5520 & 0.0009 & 0.6563 & 0.0011 \\
%$\pi_1$ MAP  & 0.5120 & 0.0000 & 0.6477 & 0.0000 \\
%$\pi_1$ DE  & 0.5118 & 0.0000 & 0.6476 & 0.0000 \\
%\bottomrule
%\end{tabular}

%\begin{tabular}{lrrrrrrr}
%\toprule
%Estimator & Accuracy & NLL & AUC ROC & Total Unc. (Correct) & Epis. Unc. (Correct) & Total Unc. (Incorrect) & Epis. Unc. (Incorrect) \\
%\midrule
%$\hat{\pi}_0$ (ours) & 0.7848 & 0.5196 & 0.8663 & 0.6135 & 0.0009 & 0.6675 & 0.0011 \\
%$\pi_0$ MAP & 0.7907 & 0.4940 & 0.8735 & 0.5843 & 0.0000 & 0.6588 & 0.0000 \\
%$\pi_0$ DE & 0.7913 & 0.4939 & 0.8736 & 0.5843 & 0.0000 & 0.6587 & 0.0000 \\
%$\check{\pi}_1$ (ours)  & 0.8625 & 0.4144 & 0.9357 & 0.5520 & 0.0009 & 0.6563 & 0.0011 \\
%$\pi_1$ MAP & 0.8701 & 0.3805 & 0.9427 & 0.5120 & 0.0000 & 0.6477 & 0.0000 \\
%$\pi_1$ DE & 0.8699 & 0.3804 & 0.9426 & 0.5118 & 0.0000 & 0.6476 & 0.0000 \\
%\bottomrule
%\end{tabular}

}

\section{Conclusion}\label{sec:conc}
The development of machine learning methodologies is of high relevance now, due to the availability of data and modern advanced algorithms.
In this work we considered the application of multilevel Monte Carlo (MLMC) for various Bayesian machine learning problems, 
%where we exploited the use of 
leveraging trace-class neural network (TNN) priors. These priors had been previously used on a range of Bayesian inference tasks, where now we combined
this with the methodology of multilevel sequential Monte Carlo (MLSMC) samplers, which is based on MLMC. In particular we motivated the use of such
priors in a multilevel setting, where we were able to firstly prove that one does attain the canonical Monte Carlo rate, 
unlike others priors based on neural network methodologies, but also providing
a bound on mean square error (MSE), using the methodology described above reducing the MSE-to-cost rate. 
Numerical experiments were conducted on a range of common machine learning
problems, which includes regression, classification and reinforcement learning, where we were able to verify the reduction in computational cost to achieve a particular order of MSE.

For future considerations of work, one natural direction is to see if one can extend this work to multi-index Monte Carlo \cite{HNT16}, which has shown to gain efficiency over
MLMC methods. This requires more, especially related to choosing the optimal set, but can be viewed as a natural extension. Another direction is to consider
other applications beyond this work, such as clustering. In a Bayesian context, such popular methods would be the likes of Bayesian hierarchical clustering \cite{HG05,JHS05} which
is related to mixture models. One could also exploit more advanced Monte Carlo proposal, based on gradient information, which could enhance the performance,
such as related to the reinforcement learning example. Such examples could include, Metropolis-Hastings adjusted algorithm, or Hamiltonian Monte Carlo. This examples, and others,
will be conducted for future work. As eluded to in the numerical experiments, one could consider alternative ways in which TNN priors, for a wider range of $\alpha \in \mathbb{R}$,
could be analyzed where one attains a canonical rate of convergence. 
It would be of interest to derive a full complexity analysis, for the application of MLMC to deep neural
approximations of data-driven models \cite{GHR21,LMM21}.
%\textcolor{black}{
%It is well-known that Monte Carlo methods can perform poorly in very high dimensions. As a result it would be of interest to see how our methodology can be applied to these settings, which go beyond the problems we have tested. 

\textcolor{black}{
Deep learning is often applied to unstructured data with very high input dimension, which can be extremely computationally intensive and go 
%goes well 
beyond the cases considered here. 
%Indeed 
As an example, the work \cite{post} leverages 100s of \textcolor{black}{tensor processing units 
%(TPU, 
(Google's proprietary alternative to graphics processing unit)} to apply MCMC to CIFAR-10 and IMDb datasets. 
Our approach is more efficient, but we do not have access to such resources.
We do however present a simplified version of our method for the latter dataset,
{\em which we run on a laptop}, and we achieve a comparable accuracy.
%but still not 
%nearly 
%enough to look at such models on a laptop.
%One potential remedy is to consider p
Our approach is also amenable to parallelization \cite{liang}, which can help to leverage supercomputers even more effectively.
%However, for the present work, we do not have access to such resources.
%Finally one aspect we did not test was the modification of the depth 
%and width-size within the inference procedure of the neural network. This flexibility could provide an understanding on balancing computational complexity and accuracy more effectively. 
%More work is required beyond this work, however this is a fruitful direction that can provide considerable benefits. 
It may also be useful to modify the depth of the network at the same time as the width, and find an optimal balance there. 
%This requires further investigation, which requires a careful construction
%of a modified MLMC scheme. 
For modern CNNs and LLMs \cite{newbishop} the TNN architecture will need to be modified to impose decay more cleverly. 
Candidates for CNN include kernel width (up-to input dimension) and number of channels.
Candidates for LLMs include (i) embedding dimension, (ii) number of heads in the multi-head attention layers,
%number of multi-head attention layers, 
and (iii) width of the hidden layer of the feedforward layers.
Beyond that, data can be split into levels hierarchically \cite{unbiased_giro}.
Following from this, an exciting final direction, related to the points above, 
%is that we would like to test our methodology on 
is to apply the method to more advanced practical applications, since the value of a method lies in its ultimate impact. 
%This is to give a real full verification in high-dimensional settings, 
The purpose of this work 
%rather, 
is only to introduce 
a new methodology which is achievable through the combination of different mathematical entities.}
%}

\section*{Acknowledgement}
AJ was supported by CUHK-SZ UDF01003537. NKC is also supported by a City University
of Hong Kong Start-up Grant, project number 7200809. SSS holds the Tibra Foundation professorial chair and gratefully acknowledges research
funding as follows: This material is based upon work supported by the Air Force Office of Scientific
Research under award number FA2386-23-1-4100.

\appendix

\section{Proofs for Proposition \ref{prop:main_res}}\label{app:proofs}

The proof of Proposition \ref{prop:main_res} essentially follows that of \cite[Theorem 3.1]{beskos2017multilevel}, except there are some modifications required.
We mainly provide details of these additional calculations, but we remark that to fully understand the proof, one must have read  \cite{beskos2017multilevel}.
{\color{black}We recall that $\eta_1 = q_1$ and that for any $l\in\{2,\dots,L\}$ one can establish that
$$
\eta_l(d\theta_l) = \frac{\int_{\Theta_1\times\cdots\times\Theta_{l-1}} \{\prod_{p=1}^{l-1}G_p(\theta_p)\}q_1(\theta_1)\prod_{p=1}^{l-1}M_p(u_{p-1},du_p)d\theta_1 }{\int_{\Theta_1\times\cdots\times\Theta_{l}} \{\prod_{p=1}^{l-1}G_p(\theta_p)\}q_1(\theta_1)\prod_{p=1}^{l-1}M_p(u_{p-1},du_p)d\theta_1}.
$$}
{\color{black}One can show that for any $l\in\{1,\dots,L\}$ and any measurable $\varphi_{l}:\textcolor{black}{\Theta}_l\rightarrow\mathbb{R}$ that is $\pi_{l}-$integrable that
$$
\pi_{l}(\varphi_{l}) = \int_{\Theta_{l+1}}\varphi(\theta_{l})\eta_{l+1}(d\theta_{l+1})
$$
for $l\in\{1,\dots,L-1\}$. Note that we need only approximate expectations associated to $\pi_l$, $l\in\{1,\dots,L-1\}$
using our multilevel identity.}
From herein $C$ is a finite constant whose value may change on each appearance, but, does not depend upon $l$.
We will also make use of the $C_p$ inequality. For two real-valued random variables
$X$ and $Y$ defined on the same probability space, with expectation operator $\mathbb{E}$, suppose that for some fixed $p\in(0,\infty)$, $\mathbb{E}[|X|^p]$
and $\mathbb{E}[|Y|^p]$ are finite, then the $C_p-$inequality is
\begin{equation}
\label{eq:cp}
\mathbb{E}[|X+Y|^p] \leq C_p\Big(\mathbb{E}[|X|^p]+ \mathbb{E}[|Y|^p]\Big),
\end{equation}
where $C_p=1$, if $p\in(0,1)$ and $C_p=2^{p-1}$ for $p\in[1,\infty)$.

\begin{proof}[Proof of Proposition \ref{prop:main_res}]
For any $l\in\{2,3,\dots,L\}$, we have the decomposition{\color{black}
$$
\frac{\eta_{l}^{P_{l}}(G_{l}f_l)}{\eta_{l}^{P_{l}}(G_{l})}-\eta_{l}^{P_{l}}(f_{l-1}) - \left\{
\frac{\eta_{l}(G_{l}f_l)}{\eta_{l}(G_{l})}-\eta_{l}(f_{l-1})
\right\}= \sum_{j=1}^3 T_j^{P_{l}},
$$
where
\begin{eqnarray*}
T_1^{P_{l}} & = & -\frac{\eta_{l}^{P_{l}}(G_{l}f_l)}{\eta_{l}^{P_{l}}(G_{l})}[\eta_{l}^{P_{l}}-\eta_{l}]\left(\frac{Z_l}{Z_{l-1}}G_{l}-1\right),\\
T_2^{P_{l}} & = & [\eta_{l}^{P_{l}}-\eta_{l}]\left(f_l\left\{\frac{Z_l}{Z_{l-1}}G_{l}-1\right\}\right),\\
T_3^{P_{l}} & = & [\eta_{l}^{P_{l}}-\eta_{l}](f_l-f_{l-1}).
\end{eqnarray*}}
As a result, by using the $C_2-$inequality, from \eqref{eq:cp} we can consider{\color{black}
\begin{equation}\label{eq:main_eq}
\mathbb{E}[(\widehat{\pi}_L(f_L)-\pi_L(f_L))^2] \leq 
C\Big(
\mathbb{E}\Big[\Big(\frac{\eta_{1}^{P_{1}}(G_{1}f_1)}{\eta_{1}^{P_{1}}(G_{1})}-\frac{\eta_{1}(G_{1}f_1)}{\eta_{1}(G_{1})}\Big)^2\Big]
+ \sum_{j=1}^3 \mathbb{E}[(\sum_{l=2}^LT_j^{P_{l}})^2]
\Big).
\end{equation}}
 For the first-term on the R.H.S.~of \eqref{eq:main_eq} by standard results (see e.g.~\cite[Lemma A.3.]{beskos2017multilevel})
we have {\color{black}
$$
\mathbb{E}\Big[\Big(\frac{\eta_{1}^{P_{1}}(G_{1}f_1)}{\eta_{1}^{P_{1}}(G_{1})}-\frac{\eta_{1}(G_{1}f_1)}{\eta_{1}(G_{1})}\Big)^2\Big] \leq \frac{C}{P_1}.
$$}
For the summands on the R.H.S.~of \eqref{eq:main_eq} we can apply Remark \ref{rem:main_rem} to conclude the proof.
\end{proof}

To give our technical results, we require some notations. To connect with the appendix in \cite{beskos2017multilevel}, we use the same subscript
conventions (i.e.~$n,p$ instead of $l$). Let $p\in\mathbb{N}$ then {\color{black}$Q_{p}(\theta,d\theta') = G_p(\theta)M_{p}(\theta,d\theta')$ and for $1\leq p \leq n\in\mathbb{N}$ we
set for $\theta_p\in\Theta_p$ and $\varphi:\Theta_n\rightarrow\mathbb{R}$ bounded and measurable (write the collection of such functions as $\mathsf{B}_b(\Theta_n)$)
$$
Q_{p,n}(\varphi)(\theta_p) := \int_{\Theta_{p+1}\times\cdots\times\Theta_n} \varphi(\theta_n) \prod_{q=p}^{n-1} Q_q(\theta_{q},d\theta_{q+1}),
$$
with $Q_{n,n}$ the identity operator.
Then for $(\theta_p,\varphi_n)\in\Theta_p\times\mathsf{B}_b(\Theta_n)$ set
$$
D_{p,n}(\varphi_n)(\theta_p) = \frac{Q_{p,n}(\varphi_n-\eta_n(\varphi_n))(\theta_p)}{\eta_p(Q_{p,n}(1))},
$$
with $D_{n,n}(\varphi_n)(\theta_n) =\varphi_n(\theta_n)-\eta_n(\varphi_n)$. Then we make the definitions, with $\mu$ a probability measure on $\Theta_{p-1}$,
$(\varphi_p,\varphi_n)\in\mathsf{B}_b(\Theta_p)\times \mathsf{B}_b(\Theta_n)$, $\Phi_1(\eta_{0}^{P_{0}})(\varphi_1)=\int_{\Theta_1}\varphi_1(\theta_1)q_1(\theta_1)d\theta_1$ 
\begin{eqnarray*}
\Phi_p(\mu)(\varphi) & = & \frac{\mu(G_{p-1}M_{p-1}(\varphi_{p}))}{\mu(G_{p-1})}, \quad p\geq 2\\
V_p^{P_p}(\varphi_p) & = & \sqrt{P_p}[\eta_p^{P_p}-\Phi_p(\eta_{p-1}^{P_{p-1}})](\varphi_p),\\
R_{p+1}^{P_p}(D_{p,n}(\varphi_n)) & = & \frac{\eta_p^{P_p}(D_{p,n}(\varphi_n))}{\eta_p^{P_p}(G_p)}[\eta_p-\eta_{p}^{P_p}](G_p).
\end{eqnarray*}
The following decomposition is then well-known (see \cite{beskos2017multilevel}) for $\varphi_n\in\mathsf{B}_b(\Theta_n)$
$$
[\eta_n^{P_n} - \eta_n](\varphi_n) = \sum_{p=1}^n \frac{V_p^{P_p}(D_{p,n}(\varphi_p))}{\sqrt{P_p}} + \sum_{p=1}^{n-1} R_{p+1}^{P_p}(D_{p,n}(\varphi_n)).
$$}
Given the terms in \eqref{eq:main_eq} it is then necessary to deal with the decomposition above. This has largely been done in the proof of Theorem 3.1.~in \cite{beskos2017multilevel}. However, in order to use the proof there, one must provide an appropriate adaptation of \cite[Lemma A.1.~(i)-(iii)]{{beskos2017multilevel}}
and this is the subject of the following result. Below, for a scalar random variable $Z$, we use the notation $\|Z\|_r=\mathbb{E}[|Z|^r]^{1/r}$.

\begin{lemma}\label{lem:tech_lem}
Assume (A\ref{ass:1}). Then there exists a $C<\infty$, possibly depending upon $r$ in (A\ref{ass:1}) 5.,  and $\zeta\in(0,1)$ such that for any $1\leq p \leq n$, \\ $\varphi_n\in\left\{f_{n}-f_{n-1}, \tfrac{Z_{n}}{Z_{n-1}}G_n-1,
f_{n}\left(\tfrac{Z_{n}}{Z_{n-1}}G_n-1\right)\right\}$ and $\beta$ as in (A\ref{ass:1}) 3.){\color{black}
\begin{enumerate}
\item{$\sup_{\theta_p\in\Theta_p}|D_{p,n}(\varphi_n)(\theta_p)|\leq C\zeta^{n-p}n_n^{-\beta/2}$, ($1\leq p <n$).}
\item{$\|V_p(D_{p,n}(\varphi_n)\|_r\leq C\zeta^{n-p}n_n^{-\beta/2}$, ($r$ as in (A\ref{ass:1}) 5.).}
\item{$\|R_{p+1}^{P_p}(D_{p,n}(\varphi_n)\|_r\leq C\zeta^{n-p}n_n^{-\beta/2}$, ($0\leq p <n$, $r$ as in (A\ref{ass:1}) 5.).}
\end{enumerate}}
\end{lemma}

\begin{proof}
Throughout the proof, we only consider the case $\varphi_n = \tfrac{Z_{n}}{Z_{n-1}}G_n-1$. The other cases can be dealt with in a similar manner.
We start with 1.~and noting that{\color{black}
 \begin{equation}\label{eq:prf3}
\frac{Z_{n}}{Z_{n-1}}G_n-1 = \frac{Z_{n}}{Z_{n-1}}(G_n-1) + \frac{Z_{n}}{{Z_{n-1}}} - 1.
\end{equation}
Therefore
$$
D_{p,n}(\varphi_n)(\theta_p) = \frac{Z_{n}}{Z_{n-1}}D_{p,n}(G_n-1)(\theta_p) \leq C|D_{p,n}(G_n-1)(\theta_p)|,
$$
so we need only work with $D_{p,n}(G_n-1)(\theta_p)$. Now, we note that
\begin{eqnarray*}
D_{p,n}(G_n-1)(\theta_p) & = &  \frac{\eta_p(Q_{p,n-1}(1))}{\eta_p(Q_{p,n}(1))}
\Bigg\{D_{p,n-1}(Q_{n-1}(G_n-1))(\theta_p) - \\ & & 
\frac{\eta_{n-1}(G_{n-1}M_{n-1}(G_n-1))}{\eta_{n-1}(G_{n-1})}D_{p,n-1}(G_{n-1})(\theta_p)\Bigg\}.
\end{eqnarray*}
Now by using (A\ref{ass:1}) 1.~\& 2.~we have
\begin{equation}\label{eq:prf1}
\frac{\eta_p(Q_{p,n-1}(1))}{\eta_p(Q_{p,n}(1))}
\leq C.
\end{equation}
 In addition
 $$
 M_{n-1}(G_n-1)(\theta_{n-1}) = \int_{\Theta_n} \left(\frac{p_{n}(y_{1:N}|\theta_n')}{p_{n}(y_{1:N}|\theta_{n-1}')} - 1\right)M_{n-1}(\theta_{n-1},d\theta_n').
 $$
Now by using (A\ref{ass:1}) 1.~followed by 
 (A\ref{ass:1}) 4.~and then  (A\ref{ass:1}) 5.
 \begin{equation}\label{eq:prf2}
 |M_{n-1}(G_n-1)(\theta_{n-1})| \leq C \int_{\Theta_n} |f_{n}(x,\theta_n')-f_{n-1}(x,\theta_{n-1}')|M_{n-1}(\theta_{n-1},d\theta_n') \leq Cn_n^{-\beta/2}.
 \end{equation}
 Therefore, we have 
\begin{eqnarray*}
D_{p,n}(G_n-1)(\theta_p) & = & \frac{\eta_p(Q_{p,n-1}(1))}{\eta_p(Q_{p,n}(1))}
\Bigg\{
D_{p,n-1}\left(\frac{Q_{n-1}(G_n-1)}{\|Q_{n-1}(G_n-1)\|_{\infty}}\right)(\theta_p)
\|Q_{n-1}(G_n-1)\|_{\infty}- \\ & &
\frac{\eta_{n-1}(G_{n-1}M_{n-1}(G_n-1))}{\eta_{n-1}(G_{n-1})}D_{p,n-1}(G_{n-1})(\theta_p)
\Bigg\}
\end{eqnarray*}
where for any $\varphi_p\in\mathsf{B}_b(\Theta_p)$, $\|\varphi_p\|_{\infty}=\sup_{\theta_p\in\Theta_p}|\varphi_p(\theta_p)|$. Application of \eqref{eq:prf1} and
\cite[Lemma A.1.~(i)]{beskos2017multilevel} yields
$$
|D_{p,n}(G_n-1)(\theta_p)| \leq C\zeta^{n-p}\left(\|Q_n\{G_n-1\}\|_{\infty}+\|G_{n-1}\|_{\infty}\right).
$$
Then using (A\ref{ass:1}) 1.~and \eqref{eq:prf2} yields
$$
\sup_{\theta_p\in\Theta_p}|D_{p,n}(\varphi_n)(\theta_p)|\leq C\zeta^{n-p}n_n^{-\beta/2}.
$$
}
For the proof of 2.~the case {\color{black}$1\leq p <n$} follows immediately from 1.~and the proof in \cite[Lemma A.1.~(ii)]{beskos2017multilevel}. So we need only consider $n=p$, which
reads
$$
\sqrt{P_n}\mathbb{E}[|[\eta_n^{P_n}-\Phi_n(\eta_{n-1}^{P_{n-1}})](\varphi_n)|^r]^{1/r},
$$
and then using \eqref{eq:prf3} we need only consider
$$
\sqrt{P_n}\mathbb{E}[|[\eta_n^{P_n}-\Phi_n(\eta_{n-1}^{P_{n-1}})](G_n-1)|^r]^{1/r}.
$$
Using the conditional Marcinkiewicz-Zygmund inequality gives the upper-bound
$$
\|V_n(D_{n,n}(\varphi_n)\|_r\leq C\mathbb{E}[|G_n(\theta_n^1)-1|^r]^{1/r}.
$$
Taking conditional expectations w.r.t.~{\color{black}$M_{n-1}$} and using \eqref{eq:prf2} yields
$$
\|V_n(D_{n,n}(\varphi_n)\|_r\leq Cn_1^{-\beta/2},
$$
which is the desired result.

For the proof of 3.~this follows immediately from 1.~and the proof in \cite[Lemma A.1.~(iii)]{beskos2017multilevel}.
\end{proof}

\begin{remark}\label{rem:main_rem}
Given the results in Lemma \ref{lem:tech_lem}, one can follow the proofs of \cite[Theorem 3.1.]{beskos2017multilevel} to deduce that for any $j\in\{1,2,3\}${\color{black}
$$
\mathbb{E}[(\sum_{l=2}^LT_j^{P_{l}})^2] \leq C\left(\sum_{l=2}^L \frac{1}{P_{l}n_l^{\beta}}+ \sum_{2\leq l<q\leq L} \frac{1}{(n_ln_q)^{\beta/2}}
\left\{\frac{\zeta^{q-1}}{P_{l}} + \frac{1}{P_{l}^{1/2}P_{q}}\right\}\right),
$$}
with the notations as in the statement of Proposition \ref{prop:main_res}. 
\end{remark}

\bibliographystyle{siamplain}

\begin{thebibliography}{10}

%\bibitem{sobolev}
%Robert A. Adams and John JF Fournier. 
%\newblock {\em Sobolev spaces.} 
%Elsevier, 2003.





\bibitem{agrawal2020wide}
Devanshu Agrawal, Theodore Papamarkou, and Jacob~D. Hinkle.
\newblock Wide neural networks with bottlenecks are deep gaussian processes.
\newblock {\em J. Mach. Learn. Res.}, 21:1--66, 2020.

\bibitem{bach}
Francis Bach.
\newblock On the equivalence between kernel quadrature rules and random feature
  expansions.
\newblock {\em J. Mach. Learn. Res.}, 18(1):714--751, 2017.

%\bibitem{barron}
%Andrew R. Barron.
%\newblock Approximation and estimation bounds for artificial neural networks.
%\newblock {\em Machine learning} 14: 115-133, 1994.
%

\bibitem{DPB85}
Dimitri P. Bertsekas
\newblock{\em Dynamic Programming and Optimal Control.}
Athena Scientific, 1985.

\bibitem{beskos2017multilevel}
Alexandros Beskos, Ajay Jasra, Kody J. H. Law, Raul Tempone, and Yan Zhou.
\newblock Multilevel sequential Monte Carlo samplers.
\newblock {\em Stochastic Processes and their Applications}, 127(5):1417--1440,
  2017.

\bibitem{beskos2018multilevel}
Alexandros Beskos, Ajay Jasra, Kody J. H. Law, Youssef Marzouk and Yan Zhou.
\newblock Multilevel sequential Monte Carlo with
dimension independent likelihood informed proposals.
\newblock {\em SIAM/ASA Journal on Uncertainty Quantification}, 6:762--786,
  2018.


\bibitem{bishop}
Christopher~M. Bishop.
\newblock Pattern recognition and machine learning (information science and
  statistics), 2006.

\bibitem{newbishop}
Christopher M. Bishop and Hugh Bishop. 
{\em Deep learning: Foundations and concepts}. 
Springer Nature, 2023.

\bibitem{blei}
David M. Blei, Alp Kucukelbir, and Jon D. McAuliffe. 
\newblock  Variational inference: A review for statisticians. 
\newblock {\em Journal of the American statistical Association} 
112.518: 859-877, 2017.

\bibitem{blundell}
Charles Blundell, et al. 
\newblock Weight uncertainty in neural network. 
{\em International Conference on Machine Learning. PMLR}, 
1613-1622, 2015.

\bibitem{Buchholz2021Adaptive}
Alexander Buchholz, Nicolas Chopin, and Pierre E. Jacob.  
Adaptive tuning of Hamiltonian Monte Carlo within Sequential Monte Carlo.  
\newblock{\em Bayesian Analysis}, 16(3):745--771, 2021.



\bibitem{nnkernel}
Youngmin Cho and Lawrence Saul.
\newblock Kernel methods for deep learning.
\newblock In Y.~Bengio, D.~Schuurmans, J.~Lafferty, C.~Williams, and
  A.~Culotta, editors, {\em Advances in Neural Information Processing Systems},
  volume~22. Curran Associates, Inc., 2009.

% \textcolor{blue}{
%\bibitem{DHJ22}
% Chenguang Dai, Jeremy Heng, Pierre Jacob and Nick Whiteley. ABSOULTE BULLSHIT
%An invitation to sequential Monte Carlo samplers.
%\newblock{\em Journal of the American Statistical Association}, 117(539):1587--1600, 2022.}
%

\bibitem{cotter}
Simon L. Cotter, Gareth O. Roberts, Andrew M. Stuart, and David White. 
MCMC methods for functions: modifying old algorithms to make them faster. 
Statistical Science 28:3, 424-446, 2013. 



 \bibitem{delmoral}
 Pierre Del Moral.
 \newblock \emph{Feynman-Kac Formulae.}
 \newblock Springer, 2004.
 
 
 \bibitem{DDJ06}
 Pierre Del Moral, Arnaud Doucet and Ajay Jasra.
 Sequential Monte Carlo samplers.
 \newblock{\em Journal of the Royal Statistical Society: Series B (Statistical Methodology)}, 68(3), 411--436, 2006.
 
 \bibitem{ddj2012}
 Pierre Del Moral, Arnaud Doucet and Ajay Jasra.
 \newblock On adaptive resampling strategies for sequential Monte Carlo methods.
 \newblock \emph{Bernoulli}, 18:252--278, 2012.


\bibitem{del2017multilevel},
{Pierre Del Moral, Ajay Jasra and Kody JH Law},  
\newblock{Multilevel sequential Monte Carlo: Mean square error bounds under verifiable conditions},
\newblock \emph{Stochastic Analysis and Applications}, {35}:{3}, {478--498}, {2017}.

\bibitem{depeweg2018decomposition}
Stefan Depeweg, Jose-Miguel Hernandez-Lobato, Finale Doshi-Velez, and Steffen Udluft.
\newblock Decomposition of uncertainty in Bayesian deep learning for efficient and risk-sensitive learning.
\newblock {\em International conference on machine learning}, pages 1184--1193, 2018.


 \bibitem{DP97}
Pedro Domingos and Michael Pazzani.
On the Optimality of the Simple Bayesian Classifier under Zero-One Loss.
\newblock{\em Machine Learning}, 29, 103--130, 1997.

 \bibitem{DGS18}
Matthew M. Dunlop, Mark A. Girolami, Andrew M. Stuart, Aretha L. Teckentrup.
How deep are deep Gaussian processes?
\newblock{\em J. Mach. Learn. Res.}, 19 (54), 1--46, 2018.

 \bibitem{DLS21}
 Matthew M. Dunlop, Chen Li and Georg Stadler.
 Bayesian neural network priors for edge-preserving inversion.
 \newblock {\em arXiv preprint arXiv:2112.10663}, 2021.
 
 \bibitem{EJJ22}
 Weinan E, Jiequn Han and Arnulf Jentzen.
 Algorithms for solving high dimensional PDEs: from nonlinear Monte Carlo to machine learning.
 \newblock{\em Nonlinearity} 35, 278--310, 2022.
 
  \bibitem{E2}
 Weinan E., Chao Ma, and Lei Wu. 
 \newblock 
 Barron spaces and the compositional function spaces for neural network models. \newblock {\em arXiv preprint arXiv:1906.08039}, 2019.
 
 
 
\bibitem{girolamigpclass}
Maurizio Filippone and Mark Girolami.
\newblock Pseudo-marginal bayesian inference for gaussian processes.
\newblock {\em IEEE Transactions on Pattern Analysis and Machine Intelligence},
  36(11):2214--2226, 2014.

\bibitem{FGG97}
Nir Friedman, Dan Geiger and Moises Goldszmidt.  
Bayesian network classifiers.
\newblock{\em Machine Learning},  29, 131--163, 1997.

\bibitem{FS21}
Masahiro Fujisawa and Issei Sato.
Multilevel Monte Carlo variational inference.
\newblock{\em J. Mach. Learn. Res.}, 22, 1-44, 2021.

\bibitem{gal}
Yarin Gal and Zoubin Ghahramani. 
\newblock Dropout as a bayesian approximation: Representing model uncertainty in deep learning.
{\em International conference on machine learning. PMLR}, 
1050-1059, 2016.



\bibitem{MBG08}
Michael B. Giles.
Multilevel Monte Carlo path simulation.
\newblock{\em Op. Res.}, 56, 607--617, 2008.


\bibitem{MBG15}
Michael B. Giles.
\newblock{Multilevel Monte Carlo methods.}
\newblock{\em  Acta Numerica}, 24, 259--328, 2015.


\bibitem{GHR21}
Thomas Gerstner, Bastian Harrach, Daniel Roth and  Martin Simon.
Multilevel Monte Carlo learning.
Arxiv preprint arxiv:2102.08734, 2021.
\textcolor{black}{
\bibitem{calib}
Chuan Guo, Geoff Pleiss,  Yu Sun and Kilian Q. Weinberger.
On calibration of modern neural networks.
\newblock{\em ICML'17: Proceedings of the 34th International Conference on Machine Learning}, 70,
1321--1330, 2017.}


\bibitem{gustafsson2020evaluating}
Fredrik K. Gustafsson, Martin Danelljan, and Thomas B. Schon.
\newblock Evaluating scalable Bayesian deep learning methods for robust computer vision.
\newblock {\em Proceedings of the IEEE/CVF conference on computer vision and pattern recognition workshops}, pages 318--319, 2020.

\textcolor{black}{
\bibitem{HT20}
Bobby He, Balaji Lakshminarayanan and Yee Whye Teh.
Bayesian deep ensembles via the neural tangent kernel.
\newblock{\em  In Advances in Neural Information Processing System},1010--1022, 2020.}


\bibitem{hullermeier2021aleatoric}
Eyke Hüllermeier and Willem Waegeman.
\newblock Aleatoric and epistemic uncertainty in machine learning: An introduction to concepts and methods.
\newblock {\em Machine learning}, 110(3):457--506, 2021.

\bibitem{krause2025probabilistic}
Andreas Krause and Jonas Hübotter.
\newblock Probabilistic Artificial Intelligence.
\newblock {\em arXiv preprint arXiv:2502.05244}, 2025.

\bibitem{lakshminarayanan2017simple}
Balaji Lakshminarayanan, Alexander Pritzel, and Charles Blundell.
\newblock Simple and scalable predictive uncertainty estimation using deep ensembles.
\newblock {\em Advances in neural information processing systems}, 30, 2017.

\bibitem{HG05}
Katherine A. Heller and Zoubin Ghahramani.
Bayesian hierarchical clustering.
\emph{Proceedings of the 22nd international conference on Machine learning}, 297--304, 2005.

\bibitem{dnn}
Ian Goodfellow, Yoshua Bengio, and Aaron Courville.
\newblock {\em Deep learning}.
\newblock MIT press, 2016.


\bibitem{HNT16}
Abdul Lateef Haji-Ali, Fabio Nobile and Raul Tempone.
Multi-index Monte Carlo: when sparsity meets sampling.
\emph{Numerische Mathematik}, 132(4), 767--806, 2016.

\bibitem{HTF01}
 Trevor Hastie, Robert Tibshirani and Jerome H. Friedman.
\newblock{\em The Elements of Statistical Learning}.
Springer, 2001.

\bibitem{hein}
Stefan Heinrich.
{Multilevel {Mo}nte {C}arlo methods}.
In \emph{Large-Scale Scientific Computing}, (eds.~S. Margenov, J. Wasniewski \&
P. Yalamov), Springer, 2001.


\bibitem{heng2021unbiased}
Jeremy Heng, Ajay Jasra, Kody~J. H. Law, and Alexander Tarakanov.
\newblock On unbiased estimation for discretized models.
\newblock {\em arXiv preprint arXiv:2102.12230}, 2021.
\textcolor{black}{
\bibitem{HE21}
Lara Hoffmann and Clemens Elster.
\newblock Deep ensembles from a Bayesian perspective.
\newblock ArXiv preprint arXiv:2105.13283, 2021.}

\bibitem{hornik}
Kurt Hornik, Maxwell Stinchcombe, and Halbert White. 
\newblock Multilayer feedforward networks are universal approximators.
\newblock {\em Neural networks} 2.5 (1989): 359-366.
\textcolor{black}{
\bibitem{post}
Pavel Izmailov, Sharad Vikram, Matthew D. Hoffman  and Andrew Gordon Wilson.
What Are Bayesian Neural Network Posteriors Really Like?
\newblock{\em  International Conference on Machine Learning (ICML)}, 2021.
}


\bibitem{JHS05}
Ajay Jasra, Chris C. Holmes and David A. Stephens.
MCMC methods and the label switching problem in Bayesian mixture modelling.
\emph{Stat. Sci.}, 20, 50--67, 2005.


\bibitem{JKL17}
Ajay Jasra, Kengo Kamatani, Kody J. H. Law and Yan Zhou.
\newblock{Multilevel particle filters}.
\newblock{\em SIAM J. Numer. Anal.}, 55(6), 3068--3096, 2017.

\bibitem{jasra2021unbiased}
Ajay Jasra, Kody~J. H. Law, and Deng Lu.
\newblock Unbiased estimation of the gradient of the log-likelihood in inverse
  problems.
\newblock {\em Statistics and Computing}, 31(3):1--18, 2021.

\bibitem{kontoyiannis2005large},
{Ioannis Kontoyiannis  and SP Meyn},
\newblock{Large deviations asymptotics and the spectral theory of multiplicatively regular Markov processes.},
\newblock {\em Electronic Communications in Probability [electronic only]}, {10}, {61--123}, {2005}.
\textcolor{black}{
\bibitem{LPB17}
Balaji Lakshminarayanan, Alexander Pritzel, and Charles Blundell.
Simple and scalable predictive uncertainty estimation using deep ensembles.
\newblock{\em In Advances in Neural Information Processing Systems}, 6402--6413, 2017}
\bibitem{lee2018deep}
Jaehoon Lee, Yasaman Bahri, Roman Novak, Samuel~S. Schoenholz, Jeffrey
  Pennington, and Jascha Sohl-Dickstein.
\newblock Deep neural networks as gaussian processes.
\newblock In {\em International Conference on Learning Representations}, 2018.

\bibitem{sngp}
Jeremiah Liu, et al. 
\newblock Simple and principled uncertainty estimation with deterministic deep learning via distance awareness.
\newblock {\em Advances in Neural Information Processing Systems} 
33: 7498-7512, 2020.

\bibitem{LMM21}
Kjetil O. Lye, Siddhartha Mishra and Roberto Molinaro.
A multi-level procedure for enhancing accuracy of machine learning algorithms.
\newblock{\em Euro. Journal of Applied Mathematics}, 32(3), 436--469, 2021.

\bibitem{mackay}
David~J. C. MacKay.
\newblock {\em Information theory, inference and learning algorithms}.
\newblock Cambridge university press, 2003.

\bibitem{mackay1}
David J. C. MacKay. 
\newblock Developments in probabilistic modelling with neural networks—ensemble learning.
\newblock {\em Neural Networks: Artificial Intelligence and Industrial Applications: Proceedings of the Third Annual SNN Symposium on Neural Networks, Nijmegen, The Netherlands}, 
14–15 September 1995. London: Springer London, 1995.

\bibitem{matthews2018gaussian}
Alexander G de~G Matthews, Jiri Hron, Mark Rowland, Richard~E. Turner, and
  Zoubin Ghahramani.
\newblock Gaussian process behaviour in wide deep neural networks.
\newblock In {\em International Conference on Learning Representations}, 2018.


\bibitem{maas2011}
Andrew L. Maas, Raymond E. Daly, Peter T. Pham, Dan Huang, Andrew Y. Ng, and Christopher Potts.
{Learning Word Vectors for Sentiment Analysis}.
In \emph{Proceedings of the 49th Annual Meeting of the Association for Computational Linguistics: Human Language Technologies}, 
%Portland, Oregon, USA, Association for Computational Linguistics, 
June 2011, pp. 142--150.


\bibitem{murphy}
Kevin~P. Murphy.
\newblock Probabilistic machine learning: an introduction.
MIT press, 2022.

\bibitem{neal}
Radford~M Neal.
\newblock {\em Bayesian learning for neural networks}, volume 118.
\newblock Springer Science \& Business Media, 2012.
\textcolor{black}{
\bibitem{liang}
Xinzhu Liang, Joseph M. Lukens, Sanjaya Lohani, Brian T. Kirby, Thomas A. Searles and Kody J. H. Law.
SMC Is All You Need: parallel strong scaling.
arxiv preprint, arXiv:2402.06173, 2024.
}
\bibitem{LPS14}
 Gabriel Lord, Catherine.E. Powell and Tony Shardlow.
\newblock{ \em An Introduction to Computational Stochastic PDEs.}
Cambridge Texts in Applied Mathematics,
\newblock 2014.
\textcolor{black}{
\bibitem{PS17}
 Nicholas G. Polson and Vadim Sokolov.
Deep learning: A Bayesian perspective.
\newblock{\em Bayesian Anal.}, 12(4): 1275--1304, 2017.}


\bibitem{RA07}
Deepak Ramachandran and Eyal Amir.
Bayesian inverse reinforcement learning.
\newblock{\em IJCAI}, 7, 2586--2591, 2007.
%\bibitem{rahimi}
%Ali Rahimi and Benjamin Recht.
%\newblock Weighted sums of random kitchen sinks: replacing minimization with
 % randomization in learning.
%\newblock {\em Advances in Neural Information Processing Systems}, 1313--1320. Citeseer, 2008.


\bibitem{reimers-2019-sentence-bert}
Nils Reimers and Iryna Gurevych.  
Sentence-BERT: Sentence embeddings using Siamese BERT-networks.  
\newblock{\em Proceedings of the 2019 Conference on Empirical Methods in Natural Language Processing and the 9th International Joint Conference on Natural Language Processing (EMNLP-IJCNLP)}, pages 3982--3992, 2019.


\bibitem{IR01}
Irina Rish.
An empirical study of the naive Bayes classifier.
\newblock{\em Empir. methods Artif. Intell. Work.}, IJCAI, 22230, 41--46, 2001.

\bibitem{sell2020dimension}
Torben Sell and Sumeetpal~S. Singh.
\newblock Dimension-robust function space mcmc with neural network priors.
\newblock {\em arXiv preprint arXiv:2012.10943}, 2020.



\bibitem{shaker2020aleatoric}
Mohammad Hossein Shaker and Eyke Hüllermeier.
\newblock Aleatoric and epistemic uncertainty with random forests.
\newblock {\em International Symposium on Intelligent Data Analysis}, pages 444--456, 2020.


\bibitem{SC21}
Yuyang Shi, Rob Cornish.
On multilevel Monte Carlo unbiased gradient estimation for deep latent variable models.
\newblock{\em Proceedings of The 24th International Conference on Artificial Intelligence and Statistics, PMLR}130:3925--3933, 2021.

\bibitem{song2020mpnet}
Kaitao Song, Xu Tan, Tao Qin, Jianfeng Lu, and Tie-Yan Liu.  
Mpnet: Masked and permuted pre-training for language understanding.  
\newblock{\em Advances in Neural Information Processing Systems}, 16857--16867, 2020.  
\newblock\texttt{sentence-transformers/all-mpnet-base-v2}

\bibitem{unbiased_giro}
Heiko Strathmann, Dino Sejdinovic, and Mark Girolami. 
Unbiased Bayes for big data: Paths of partial posteriors.
{\em arXiv preprint arXiv:1501.03326}, 2015.


\bibitem{TEM12}
Emanuel Todorov, Tom Erez, and Yuval Tassa.
Mujoco: A physics engine for model-based control.  
\newblock{\em IEEE/RSJ International Conference on Intelligent Robots and Systems, pages 5026-5033}. IEEE, 2012.
 
\bibitem{VVM19}  
Mariia Vladimirova, Jakob Verbeek, Pablo Mesejo, and Julyan Arbel.
Understanding priors in Bayesian neural networks at the unit level.
\newblock{\em International Conference on Machine Learning}, 6458--6467. PMLR, 2019.

\bibitem{wang2020minilm}
Wenhui Wang, Furu Wei, Li Dong, Hangbo Bao, Nan Yang, and Ming Zhou.  
Minilm: Deep self-attention distillation for task-agnostic compression of pre-trained transformers.  
\newblock{\em Advances in Neural Information Processing Systems}, 5776--5788, 2020.


\bibitem{williams}
Christopher~K. I. Williams.
\newblock Computation with infinite neural networks.
\newblock {\em Neural Computation}, 10(5):1203--1216, 1998.



\bibitem{wilson2022evaluating}
Andrew Gordon Wilson, Pavel Izmailov, Matthew D. Hoffman, Yarin Gal, Yingzhen Li, Melanie F. Pradier, Sharad Vikram, Andrew Foong, Sanae Lotfi, and Sebastian Farquhar.
\newblock Evaluating approximate inference in Bayesian deep learning.
\newblock {\em NeurIPS 2021 Competitions and Demonstrations Track}, pages 113--124, 2022.

\bibitem{ZX19}
Zhiqing Xiao.
\newblock{\em Reinforcement Learning: Theory and Python Implementation}.
China Machine Press, 2019.

\end{thebibliography}

\end{document}